\documentclass[11pt,mnsc]{amsart}
\usepackage{commands}
\usepackage[margin=1in]{geometry}
\usepackage{caption}
\usepackage{subcaption}
\usepackage[foot]{amsaddr}
\usepackage[english]{babel}
\usepackage[latin1]{inputenc}
\usepackage[numbers]{natbib}
\usepackage[normalem]{ulem}

\DeclareOldFontCommand{\bf}{\normalfont\bfseries}{\mathbf}
\DeclareOldFontCommand{\cal}{\normalfont\bfseries}{\mathcal}

\def\orange{\textcolor{black}}

\pdfoutput=1
\usepackage{color}

\usepackage{setspace}

\begin{document}
\title[Optimal Bubble Riding]{Optimal Bubble Riding: \\A Mean Field Game with Varying Entry Times}

 \author{Ludovic Tangpi}
 \author{Shichun Wang}
 \address{Princeton University, Operations Research and Financial Engineering. Both authors gratefully acknowledge financial support from the NSF grant DMS-2005832 and the NSF CAREER award DMS-2143861.}

\maketitle

\begin{center}
\today
\end{center}

\begin{abstract}
\onehalfspacing
Recent financial bubbles such as the emergence of cryptocurrencies and ``meme stocks'' have gained increasing attention from both retail and institutional investors. 
In this paper, we propose a game-theoretic model on optimal liquidation in the presence of an asset bubble. Our setup allows the influx of players to fuel the price of the asset. 
Moreover, traders enter the market at possibly different times and take advantage of the uptrend at the risk of an inevitable crash. In particular, we consider two types of crashes: an \emph{endogenous} burst which results from excessive selling, and an \emph{exogenous} burst which cannot be anticipated and is independent from the actions of the traders. 
    
The popularity of asset bubbles suggests a large-population setting, which naturally leads to a mean field game (MFG) formulation. We introduce a class of MFGs with varying entry times. In particular, an equilibrium depends on the entry-weighted average of conditional optimal strategies. To incorporate the exogenous burst time, we adopt the method of progressive enlargement of filtrations. We prove existence of MFG equilibria using the weak formulation in a generalized setup, and we show that the equilibrium strategy can be decomposed into before-and-after-burst segments, each part containing only the market information. We also perform numerical simulations of the solution, which allow us to provide some intriguing results on the relationship between the bubble burst and equilibrium strategies.
\end{abstract}



\section{Introduction}
\onehalfspacing
Asset bubbles from various financial markets have recently gained much attention. For instance, in 2021 the short squeeze on Reddit triggered a so-called ``meme stock'' frenzy, whose impact surprised the financial industry by its unprecedented magnitude. Certain cryptocurrencies also experienced explosive price increase in the past decade. In particular, the price of Bitcoin grew 25 times in 2017, a trend reminiscent of the tech boom (also known as the ``dot--com'' bubble) at the turn of the last century, during which the Nasdaq index rose 400\%. Bubbles of this size can be detected \cite{LPPLBitcoin19, NFT22, BubblevsDiscovery20, YOLOTrading22} and are attractive to traders despite the risk from an inevitable future price correction. In this paper, we shall analyze the notion of \textit{optimal bubble riding} in a symmetric, large population game setting. \par

The soaring prices certainly attract non-institutional traders. Retail investors took up nearly 25\% of U.S. equity trading volume in Q1 2021, compared to 15\% in 2019 \cite{CreditSuisse22}. A similar trend was observed during the dot-com bubble, when a large wave of retail investors poured into tech stocks. What is puzzling, however, is that sophisticated investors such as hedge funds were also riding the dot-com bubble, despite anticipating a major price correction (see e.g. \cite{HedgeFunds04, TechBubble2011} for empirical evidence). Such evidence challenges the traditional perspective which considers bubble riding merely as ``irrational exuberance'' that results from the herding behavior of uninformed investors. Researchers have thus been interested in understanding rational bubble modeling. 

Among the most influential works on the topic, the theoretical model by \citet{AbreuBrunnermeier03} on profitable bubble riding offers insight to the persistence of asset bubbles in the presence of rational arbitrageurs (see also \cite{RobustBubbleModel12, FiniteHorizon16} for extensions of this model). 
\citet{RidingBubbleConvex18} use a dynamic trading model to show that sophisticated, risk-averse money managers can invest in overvalued non-benchmark asset. The Log-Periodic Power Law (LPPL) model \cite{JLS1,JLS2, JLS3,LPPL10}, sometimes referred to as the Johansen-Ledoit-Sornette (JLS) model \cite{JLS4}, characterizes bubbles by faster-than-exponential growth of the asset price with a finite singularity point that results from the positive feedback loop from investors. A more mathematical finance perspective is developed by \citet{Jarrow2007} and \citet{Protter2013}, who model bubbles using the theory of strict local martingales and various no-arbitrage arguments. Nonetheless, these models do not link the growth of the bubble directly to the inflow of traders, nor do they take advantage of the large population feature. 
Moreover, to our knowledge, the present paper is the first that looks into \emph{optimal} bubble riding in a stochastic, non-cooperative, and continuous-time game perspective.

In contrast to the works mentioned above, we integrate the size and the dynamics of the population into our model. Our basic framework borrows from the seminal work of \citet{AlmgrenandChriss01} on price impact and its multiplayer game extension by \citet{CarlinLoboViswanathan07}. The core of the price impact model is an optimal execution problem, where traders need to liquidate their assets in a given interval of time. This framework is suitable for bubble riding as well because most ``riders'' simply wish to take advantage of the growth, and thus they do not intend to hold the overpriced asset. In an optimal liquidation setting, traders optimize their trading speed $\alpha$ to minimize the temporary trading cost while lowering their inventory. The existence of an asset bubble intensifies the urgency in light of the inevitable yet unpredictable moment when it will burst. The asset bubble's presence amplifies the urgency to liquidate, given the impending but uncertain moment of its burst. At the same time, however, the longer the trader holds onto the asset, the more she benefits from the growth of the bubble, provided that she pulls out before the crash. 

There are three main features of our model on top of the optimal execution framework: the varying entry times, a bubble component fueled by the influx of players, and a random burst time. Inspired by what \citet{AbreuBrunnermeier03} refer to as the ``awareness window'', we fix some $\eta > 0$ that measures the level of ``heterogeneity'' between traders, that is, how spread-out their entry times are. The bubble starts at time $t=0$, and the players begin trading at various entry times in the interval $[0, \eta]$. There are two potential reasons to the delay of entry. The first one is information asymmetry among the traders. It is possible that the mispricing has not yet caught some traders' attention, or that some of them have insider information and begin trading earlier than others. Another reason, perhaps more interesting, is related to risk preferences. Some traders might not yet believe in the strength of the bubble, so they choose to wait longer before entering the game. For example, using data on brokerage accounts, \citet{whoparticipated21} found that investors who joined the GameStop frenzy earlier had a history of investing in ``lottery-like'' stocks. However, since we are more interested in the Nash equilibrium and the group behavior, we consider the exact reason of entry as an extrinsic factor. In the setup of \citet{AbreuBrunnermeier03}, players enter uniformly on $[0, \eta]$. Here we only assume that the entry times are independent and identically distributed (iid) with a fixed distribution known to all players. The traders do not optimize their strategies over entry times. In fact, as soon as the agent decides that the bubble is worth riding, the optimal action is to enter immediately. The decision to enter should thus depend on the traders' awareness of the existence of the bubble and her confidence in its strength. In our accompanying paper \cite{TangpiWang23}, for example, we use the price of the bubble asset as a proxy for these two factors at the expense of a weaker result and a more abstract proof. We do not model these two factors explicitly in this paper

Our model emulates the five stages (Displacement, Boom, Euphoria, Profit-taking, and Panic) of a typical bubble identified by \citet{MinskyBook} and detailed in the book by \citet{aliber2017manias}. In particular, the Boom phase describes the period when the price increases due to waves of traders coming into the market. New players that join the game should boost the price differently from the standard market impact model. Therefore, we assume the price dynamics of the asset before the burst to have a bubble component plus a fundamental (or non-bubble) component; see \eqref{Price-Dynamics-PreBurst}. Specifically, the bubble component depends on the number of players in the game. As a result, a trader could have a large relative disadvantage even if she joins only slightly after a big crowd does. After the crash, the price dynamics loses the bubble component, and the model reduces to a standard optimal execution problem. 

Regarding the burst of the bubble, it is commonly believed that there are two types of market crashes: endogenous and exogenous. An \emph{endogenous} burst (sometimes called a Minsky moment) refers to the collapse caused by the cumulative selling pressure. For instance, \citet{TechBubble2011} showed, using individual stock peaks during the dot-com bubble, that institutions began to sell rapidly on the day of the peak which triggered the drop in prices. In contrast to a trader-driven burst, an \emph{exogenous} crash refers to a commonly observable but unanticipated event from outside the system that causes the price to plummet. A recent example took place in Q1 2020 when the COVID-19 pandemic triggered the so-called Coronavirus crash. The fact that it was an exogenous event is confirmed by \citet{ExovsEndo20}, in which the authors find no indication of endogenous factors behind the crash. Mirroring the approach of \citet{AbreuBrunnermeier03}, we allow for both types of bursts. Specifically, the true crash occurs at the minimum of the two burst times and leads to a sudden drop in price. The size of the drop does not depend on the exact type of the burst.

We define the endogenous burst time as the first time that the average inventory among the traders hits a given threshold level. The exogenous burst time is modeled as a random time, but not necessarily a stopping time in the underlying market filtration. We allow investors to continue trading after the burst until some fixed horizon. Each player chooses her continuous liquidation strategy to minimize the total cost \eqref{Definition:Modelobj} up to time $T$. As everyone's strategies should depend on the price, hence also the crash, the endogenous burst time couples all the players through their inventory levels. Similarly, the price impact term depends on the strategies of all the players. These complications make the game mathematically intractable to analyze, especially when the number of players gets large. A well-known workaround is to look at its mean field game (MFG) counterpart (see e.g. \cite{CardaliaguetMFG18, PriceImpactMFG2,PriceImpactMFG3}).

Originally introduced by \citet{LasryLions1, LasryLions2, LasryLions33} and, independently, by \citet{Huang1, Huang2}, MFGs are heuristic infinite population limits of symmetric stochastic differential games; see e.g. \cite{DelarueConv, TangpiConv, Lacker16, CardaliaguetDelarueLasryLions19} for rigorous proofs of the convergence of finite population games to MFG in various settings. As the number of players approaches infinity, the interdependence weakens due to the symmetry among the players. Instead of reacting to other players' actions, each player responds to a fixed distribution, and the equilibrium translates into a fixed point condition. 
We refer the readers to the monographs of \citet{CarmonaBookI, CarmonaBookII} for a more detailed discussion on MFGs and additional references. We adopt the weak formulation approach proposed by \citet{CarmonaLacker15}. 

The MFG that naturally emerges from our model has two distinct features which have been studied in the MFG literature. The first one is the notion of varying entry times. Suppose the representative player joins the game at time $t^* \in [0, \eta]$. Then her strategy and state processes are both kept at $0$ before $t^*$. In a classical MFG, an equilibrium requires the laws of the optimal controls and state processes coincide with the distribution of that of all players. But if players enter the game at different times, clearly their controls, or inventory levels, cannot share a common law. Both the control and state variables are fixed at $0$ for someone who has not yet entered the game. Therefore, the correct notion of consistency should aggregate the players by marginalizing the  \emph{conditional} laws over entry times. Specifically, the fixed point condition in Definition \ref{definition:MFGequilibrium} is imposed on the marginal law by integrating the regular conditional probabilities. In other words, MFG equilibrium is defined as an \emph{entry-weighted average} of optimal strategies. 

Another feature of our MFG comes from the exogenous burst time $\tau$. As mentioned earlier, $\tau$ should be an unanticipated but observable event, which bears resemblance to default times in the credit risk literature \cite{CreditRiskBook04}. To put it in mathematical terms, we are looking for a totally-inaccessible stopping time. However, such stopping times do not exist with respect to the natural choice of filtration which is generated by a Brownian motion. We will therefore solve the MFG on an enlarged filtration by adopting the method of progressive enlargement, a common technique used in credit risk analysis \cite{PE1, PE2}. On this new filtration, we can make an arbitrary random time on the original probability space into a totally inaccessible stopping time. However, we need to bring the solution back to the original filtration so that the optimal strategies of the traders still only depend on the available market information. This is achieved by adopting results of \citet{KharroubiLim14}, \citet{Ankirchner10}, and \citet{Pham10}.

 In this paper, we first prove the existence of MFG equilibrium (Theorem \ref{Theorem:MFGExistence0}) for fixed entry time using Brouwer-Schauder-Tychonoff fixed point theorem. This approach should be compared to \citet{CarmonaLacker15} and \citet{CampiFischer18}. We solve the MFG on the progressively enlarged filtration and decompose the equilibrium strategy into before-and-after-burst segments. Moreover, each part is progressively measurable with respect to the original filtration. Then, we prove the existence of MFG equilibrium with varying entry times (Theorem \ref{Theorem:MFGExistence}) through the perspective of probability kernels. Since the game introduced here is apparently new, we prove the existence results in a more generic formulation for mathematical completeness, which we then apply to the bubble riding game. 

The paper is structured as follows. In Section \ref{Section:Model}, we describe the $N$-player model for bubble riding as well as its limiting MFG. We then proceed to discuss the price dynamics and the bursting mechanism. Section \ref{Section:PE} describes the general setup under filtration enlargement. We state the assumptions and the main existence result for fixed entry time and prove it in Section \ref{Section:FP}. Section \ref{Section:RET} deals with varying entry times, and we prove the existence result in that setting. Finally in Section \ref{Section:ModelRevisit}, we revisit the bubble riding problem and prove the existence of MFG equilibrium for the model, and we provide a discussion of our numerical simulation results. We conclude by giving an outline of the proof for using MFG equilibrium to construct approximate Nash equilibrium of the finite player game. Some proofs of technical results on progressive enlargement of filtration are included in the Appendix \ref{Section:Appendix}.


\section{The Bubble Riding Model and First Main Results}\label{Section:Model}
\subsection{The N-Player Game}
\label{sec:Nplayer.game}
A crowd of $N$ agents trade on the same asset that is believed to be overpriced and has a persistent upward trend, resembling a bubble structure. The bubble starts at time $t = 0$, but the players have not necessarily started trading. We assume that the players begin trading at times independent from each other after the bubble starts. There could be two potential reasons to this delay of entry. It is possible that the mispricing has not yet caught the trader's attention, or the trader does not yet believe in the strength of the bubble. We make a simplifying assumption that the entry time is independent of everything else; though the second scenario suggests a price-dependent entry time, a case explored in our accompanying paper \cite{TangpiWang23}.

\subsubsection{Entry Time}
We let $\eta > 0$ be a known parameter that measures the level of ``heterogeneity'' between the $N$ traders. The players begin trading at various entry times in the interval  $[0, \eta]$, referred to as the ``awareness window'' by Abreu and Brunnermeier \cite{AbreuBrunnermeier03}. Let $\bT = (\T^1, \dots, \T^N) \sim \nu^N$ be a vector of i.i.d entry times, each with a fixed probability measure $\nu$ on $[0, \eta]$ with cumulative distribution function $F_\T$. This distribution reflects the differences among traders in terms of their information and beliefs on the underlying asset. We impose the assumption that $F_{\T}(0) > 0$. In other words, we assume that there are people who are already actively trading before the bubble starts, or they might have some information to anticipate an upcoming upward trend. Let $F_{\bT}^N$ denote the empirical CDF of the $N$ entry times. Conditioning on $\bT = \bt = (t^1, \dots, t^N)$, at each time $t$, the number of players already in the game is $$N_{in}(t) \ce \sumN\ind{t^i\leq t}.$$ Using the empirical CDF, we can write $$\frac{N_{in}(t)}{N} = \avg\sumN \ind{t^i \leq t} \eqqcolon F_{\bT}^N(t; \bt).$$ 
It is worth noting that the players do not optimize over their entry times in our model.
\subsubsection{Trading Rate and Inventory}
Let $\lambda_0$ denote some fixed initial probability measure on $\R_+$ with finite second moment and positive first moment. At time $t^i$, player $i$ has an initial inventory $X_{t^i}^i = \iota^i \iid \lambda_0$. Suppose that there is a common horizon $T \geq \eta > 0$, regardless of whether the bubble has burst or not. By choosing her trading rate $\alpha^i = (\alpha^i_t)_{t^i \leq t \leq T}$, the player can control her inventory trajectory by $$dX_t^i = \alpha_t^idt + \sigma dW^i_t,\quad X_t^i = 0 \text{ on } t<t^i, \quad X^i_{t^i} = \iota^i$$ 
where $W^i,\dots, W^N $ are independent $1$--dimensional Brownian motions shifted to start from $t = t^i$, which correspond to the random streams of demand \cite{CarmonaLacker15,LealThesis}. A constant $\sigma >0$ is assumed for simplicity, which can be somewhat generalized (see Section \ref{Section:PE}). A positive $\alpha_t$ corresponds to buying, and a negative $\alpha_t$ corresponds to selling. Obviously, the total supply of this stock is finite. For other practical reasons such as transaction fees, we assume that at any time a player's trading rate is in a compact interval $A \subset \R$. Before entry, naturally the trader has $0$ inventory and does not have control over it. Hence we require $\alpha_t^i = 0$ on $t < t^i$ for each $i \in \{1, \dots, N\}$.

\subsubsection{Permanent Price Impact}
Trades convey information which has a long-term impact on the price dynamics. We use the $N$-player extension of the celebrated price impact model developed by Almgren and Chriss \cite{AlmgrenandChriss01}, where the instantaneous price impact depends on the trading rate. The impact function is usually assumed to be concave, a fact repeatedly supported by strong empirical evidence and theoretical models \cite{BOUCHAUD200957,Amazing}. Let $\rho: A \to \R$ be a locally bounded function for the instantaneous impact with $\rho(0) = 0$. The true impact on the price is the average impact among the traders who have entered the game, which is defined as
\begin{equation}\label{empirical_impact}
\frac{1}{ N_{in}(t)}\sumN \rho(\alpha^i_t) = \frac{1}{N} \sumN \rho(\alpha^i_t)/F_{\bT}^N(t; \bt) = \frac{1}{F^N_{\bT}(t; \bt)}\int_A \rho(a) \theta_t^N(da) \eqqcolon \qv{\rho, \theta^N_t}_{F_{\bT}^N},
\end{equation}
where $\theta^N_t \ce \frac{1}{N}\sumN \delta_{\alpha_t^i}$ is the empirical distribution of trading rates at time $t$. Before the first entry, we set the permanent price impact to be $0$. Similarly, we define the average inventory
\begin{equation}\label{empirical_meanmu}
    \bar{\mu}_t^N \ce \begin{cases}
    0 & \text{ if } N_{in}(t) = 0\\
    \avgt \sumN X_t^i = \avg\sumN X_t^i/F^N_{\bT}(t; \bt) = \frac{1}{F^N_{\bT}(t; \bt)}\int_\R x \mu_t^N(dx) & \text{ if } N_{in}(t) \neq 0
    \end{cases}
\end{equation}
where $\mu^N_t \ce \avg \sumN\delta_{X_t^i}$ is the empirical probability distribution of inventory at time $t$. Note that the interaction terms should only depend on the inventory and control processes after entry. Therefore, we end up with a factor of $\frac{1}{F_{\bT}^N(t; \bt)}$ to scale up the simple average of $X^i$ and $\alpha^i$, which are $0$ for those who are not in the game. Since we assume $F_\mathcal{T}(0)>0$, we have $N_{in} > 0$ almost surely for all $t$ as $N \to \infty$. Hence the values of $\qv{\rho, \theta_t}_{F_{\bT}^N}$ and $\bar{\mu}_t^N$ before first entry do not matter.

\subsubsection{Price Dynamics and Bubble Component}
The price dynamics of the asset is assumed to have two components:
\begin{enumerate}
    \item  Reference price component given by $$dQ_t \ce \qv{\rho, \theta^N_t}_{F_{\bT}^N}dt + \sigma_0dW_t^0, \qquad Q_0 = P_0,$$
    where $P_0$ is a known constant, and $W^0$ is a standard 1-dimensional Brownian motion independent from $(W^i, \mathcal{T}^i)$ for $i \in \{1, \dots, N\}$. The process $Q$ serves a similar role to the reference price in the Johansen-Ledoit-Sornette model \cite{JLS3, JLS4}, or the \emph{fundamental value} process in \cite{AbreuBrunnermeier03}. This part of the price will not be impacted by the bubble burst.
    \item Bubble trend component, which is given by a bubble trend function $b: [0, T] \times [0, 1] \to \R$. The second argument in $b(t, r)$ stands for the rate of entry, where we plug in the empirical CDF of entry times at time $t$. The bubble trend should be increasing in $r$ since the bubble is fueled by the persistent influx of players, which is the main driver of the price dynamics before the burst. 
\end{enumerate}
Let $P^+_t$ denote the price of the stock before the burst of the bubble. It is assumed to follow the dynamics 
\begin{equation}\label{Price-Dynamics-PreBurst}
    dP^+_t \ce \underbrace{b(t, F^N_{\bT}(t; \bt))dt}_{\text{bubble trend}}+\underbrace{\sqbra{ \qv{\rho, \theta^N_t}_{F_{\bT}^N}dt + \sigma_0dW_t^0}}_\text{reference price}, \qquad P^+_0 = P_0,
\end{equation}
As defined in \citet{Protter2013}, the difference between $P^+$ and $Q$ is referred to as the \emph{bubble component}, specifically 
\begin{equation}\label{bubble_component}
\gamma_t \ce P^+_t - Q_t = \intt b(s, F^N_{\bT}(s; \bt))ds.  
\end{equation}
Here we give two examples of the bubble trend functions. 
\begin{example}[Exponential Bubble]
    \citet{AbreuBrunnermeier03} assumed a fully deterministic model with $P^+_t = \gamma_t = P_0\exp(\ell t)$, where $\ell > 0$ governs the speed of bubble growth. To incorporate an entry-dependent growth, one can generalize the model by setting $$\gamma_t \ce P_0\exp(\ell_tt), \quad \ell_t \ce \ell F^N_\T(t).$$
    Then the bubble function $b$ is approximated for large $N$ by taking time derivative $$b(t, r) \approx lP_0\exp(lrt)(f_\T(t)t + r)$$
    where $f_\T$ is the density function of $\nu$.
\end{example}
\begin{example}[LPPL Bubble]
    \citet{JLS1} proposed what is known as the JLS model with a mean trajectory following the log-periodic power law (LPPL): 
    \begin{equation*}
    \log(\E[\gamma_t])= A + B(T - t)^{\ell}\{1 + C\cos[\omega \log(T - t) + \phi]\}
    \end{equation*}
    with critical time set to $T$. The parameter $\ell \in (0, 1)$ measures the power law acceleration of prices, which relates to the number of traders interacting with each other \cite{EverythingLPPL}. To allow players to dynamically enter the game, set $\ell_t = \ell F^N_\T(t)$ with $\ell \in (0, 1)$. For large $N$, $F^N_\T$ is nearly deterministic. 
    Thus, we can approximate the bubble component by its expectation: \begin{equation}\label{eq:LPPL}
        \gamma_t =  \exp\pa{A + B(T - t)^{\ell_t}\{1 + C\cos[\omega \log(T - t) + \phi]\}}
    \end{equation}
    and obtain the bubble trend function $b$ by taking time derivative. Note that if all players are present from the beginning (fixed entry with $\nu_T = \delta_0$), the model reduces to a standard LPPL in expectation. 
\end{example}
After the bubble bursts, the price loses a proportion of the bubble component and drops to an after--burst value. The price dynamics in our model will be given in Equation \eqref{Price-Dynamics} below. Beforehand, let us first describe the bubble burst time.
\subsubsection{Bubble Burst Time}\label{Subsubsection:BubbleBurst}
It is commonly believed that there are two types of mechanisms that cause the burst of a bubble: endogenous and exogenous. 
An \emph{endogenous} burst refers to the collapse caused by the cumulative selling pressure, as modeled by Abreu and Brunnermeier in \cite{AbreuBrunnermeier03}. 
\begin{definition}[$N$-Player endogenous burst time]\label{Def:N-playerendoburst}
We define the $N$-player endogenous burst time as $$\bar{\tau}^N(\mu^N) \ce \inf\cbra{t > \min_{i \in \{1, 2, \dots, N\}}t^i: \inf_{s \in [0, t]}\bar{\mu}^N_s \leq \zeta_t} \wedge T$$
where $\zeta$ is a deterministic, continuous and strictly increasing function of $t$ with $\zeta_0 \in (0, \E[\iota])$.
\end{definition}
The initial threshold $\zeta_0$ should take up most of the threshold value. It is strictly positive to ensure the inevitability of the burst as the average inventory approaches $0$. We require $\zeta_0 < \E[\iota]$ because the bubble should not burst the moment it starts. We explain in the remarks below the two properties of the endogenous burst time that follow from the definition.
\begin{remark}[Increasing fragility]\label{Remark:endoburst-fragility}
The system should become increasingly vulnerable to extensive selling as the bubble grows, so we require the threshold to be strictly increasing. We also use the running minimum of $\bar\mu^N$, instead of the mean trajectory itself, to allow the market to ``remember'' the time when buying pressure is at its lowest, enabling it to react to the bubble based on the weakest moment. On the other hand, we expect that in most scenarios the average inventory will eventually become monotone decreasing due to the pressure of the burst and terminal cost, which we discuss in the next section. For a more detailed discussion on the monotonicity of average inventory, see \cite{CardaliaguetMFG18}. 
\end{remark}

\begin{remark}[Burst by T]\label{Remark:endoburst-T}
Notice that by this definition, the bubble will always burst at or before $T$. The model does not say anything about the asset after the burst. As such, our model is similar to the classical JLS bubble model \cite{JLS1, JLS2, JLS3, SornetteJohanssenBouchaud96}, in which a critical point marks the end of the bubble and potentially a transition point to a new regime. Also, requiring the bubble to burst \emph{endogenously} by $T$ is not too restrictive in the present setting.
In fact, we impose a terminal cost that encourages liquidation before $T$ , which by definition triggers the endogenous burst.
\end{remark}

An \emph{exogenous} burst, sometimes called a sunspot event, refers to a commonly observable event that triggers the cascading price adjustments from outside the system. We model it by a \emph{random time} $\tau$ on $\R_+$ with distribution $\nu_\tau$ that is independent from the rest of the model. In other words, there is no information within the system that can help anticipate this event. 

Whichever burst takes place first, the price crashes. Therefore, the actual burst time is $$\tau^* \ce \bar{\tau}^N(\mu^N) \wedge \tau.$$ 
The bursting event consumes a fraction of the bubble component, which Abreu and Brunnermeier \cite{AbreuBrunnermeier03} refer to as the ``size'' of the bubble, or the ``loss amplitude'' in the JLS model \cite{JLS4}. We model the size of the bubble by a continuous, monotone increasing deterministic function $\beta: \R_+ \to [0, 1]$. Note that in the traditional JLS model, the loss amplitude is assumed to be a constant fraction \cite{JLS4}. At the time of burst $\tau^*$, the price $P$ drops by $\beta_{\tau^*}\gamma_{\tau^*}$ and thus, trader $i$ experiences the loss of $$X^i_{\tau^*}\beta_{\tau^*}\gamma_{\tau^*} = X^i_{\tau^*}\beta_{\tau^*}\int^{\tau^*}_{0}b(t, F^N_{\bT}(t; \bt))dt.$$
Note that the bubble component does not depend on entry time. In other words, the late joiners face the same level of risk as the early joiners. Adding this jump into the price dynamics, we have
\begin{equation*}
    P_t = P^+_t(1-D_t^*)  + Q_tD_t^* + \gamma_{\tau^*}(1-\beta_{\tau^*})D^*_t,
\end{equation*}
where $D^*_t = \ind{\tau^* \leq t}$ is jump process of the bursting event, which gives us
\begin{equation}
    \begin{split}\label{Price-Dynamics}
            dP_t & = (1-D_t^*)dP^+_t - P_t^{+}dD_t^* + D_t^*dQ_t + Q_tdD_t^* + \gamma_{\tau^*}(1-\beta_{\tau^*})dD^*_t\\
            & = (1-D_t^*)dP^+_t + D_t^*dQ_t - \gamma_{\tau^*}\beta_{\tau^*}dD^*_t.
    \end{split}
\end{equation} 
\subsubsection{Objective and Equilibrium}
Player $i$'s cash process is modeled by $$dK^i_t = -\alpha^i_t[P_t + \kappa \alpha_t^i]dt$$where $\kappa > 0$ is the linear temporary price impact parameter \cite{AlmgrenandChriss01}. Under the usual self-financing condition, the wealth $V^i$ of this player follows
\begin{align*}
    dV^i_t & = dK^i_t + X^i_tdP_t + P_tdX^i_t\\
    & = \pa{-\kappa(\alpha_t^i)^2 + X_t^ib(t, F^N_{\bT}(t; \bt))\ind{t < \tau^*} +X_t^i\qv{\rho, \theta_t^N}_{F^N_\T}}dt - X^i_t\gamma_t\beta_tdD^*_t\\
    &\quad + \sigma_0X_t^idW_t^0 + \sigma P_t dW_t^i.
\end{align*}
The players are allowed to continue trading until $T$, even if the burst has already happened. For a fixed $\phi > 0$, we impose a quadratic running inventory cost $\phi (X_t^i)^2$ which Cartea \textit{et al.} \cite{Carteaetal17} refer to as \emph{ambiguity aversion}. Since the goal for the traders is to liquidate everything by $T$, we also impose a quadratic terminal inventory penalty $c(X^i_T)^2$ with $c > 0$ to encourage selling. Adding these costs to the negative of increase in wealth, we have the total cost of player $i$: $$-(V_T^i - V_{t^i}^i) + \int_{t^i}^T\phi (X_t^i)^2dt + c(X_T^i)^2.$$

\subsubsection{Probabilistic Setup}\label{Subsubsection:ProbSetup}
We formalize the elements mentioned in the previous sections. Consider a single probability space $(\Omega, \F, \P)$ supporting several mutually independent random elements including the exogenous burst time $\tau \sim \nu_\tau$, an infinite sequence of $\{\iota^i\}_{i \in \N} \iid \nu$ and Brownian motions $\{W^i\}_{i \in \N}$. For $N \in \N$, let $\mathbb{F}_N$ be the $\P$-completed natural filtration of the first $N$ Brownian motions initially enlarged by the $(\iota^1, \dots, \iota^N)$. Let $\mathbb{G}_N$ denote the progressively enlargement of $\mathbb{F}_N$ by $\ind{t \leq \tau}$. Under $\mathbb{G}$, stopping time $\tau$ is totally inaccessible\footnote{Details on the construction of $\mathbb{G}_N$ will be given in section \ref{Subsection:PEforExoBurst}.}. This is important since accessible stopping times could be utilized to create arbitrage opportunities, as noted by \citet{Protter2013}. To incorporate varying entry time, consider the product probability space $(\Omega \times [0, \eta]^N, \F \otimes \B([0, \eta]^N), \P \otimes \nu^N)$. Then the sequence of i.i.d entry times $\{\T^i\}_{i \in \N}$ are simply the canonical processes on the respective coordinate. 

We now define admissibility of the control process. Given a filtration $\mathbb{H}$ on $(\Omega, \F, \P)$, let $\mathscr{P}(\mathbb{H})$ denote the $\sigma$-algebra of $\mathbb{H}$-progressively measurable subsets of $[0, T] \times \Omega$, and denote by $\A_{\mathbb{H}}$ the $\mathscr{P}(\mathbb{H})$-measurable, $A$-valued processes. For a given vector $\bt \in [0, \eta]^N$, we call its $i$'th coordinate $t^i$.
\begin{definition}\label{Definition:AdmissibleControlN}
    Given $t^* \in [0, \eta]$, define $\A(t^*; \mathbb{G}_N)$ as the set of all admissible controls with entry time $t^*$, that is:
    $$\A(t^*; \mathbb{G}_N) \ce \{\beta \in \A_{\mathbb{G}_N}: \beta_t = 0 \text{ for } t \in [0, t^*)\}.$$
    For $i = 1, \dots, N$, an admissible strategy $\alpha^i: [0, T] \times \Omega \times [0, \eta]^N \to A$ is a $\mathscr{P}(\mathbb{G}_N) \otimes \B([0, \eta]^N)$-measurable function such that for $\nu^N$-almost every $\bt \in [0, \eta]^N$, we have $\alpha^i(\bt) \ce \alpha^i(\cdot, \cdot, \bt) \in \A(t^i; \mathbb{G}_N)$. Let $\A_{i;\mathbb{G}_N}$ denote all such strategies.
\end{definition}
A profile of varying-entry strategies $\balpha = (\alpha^1, \dots, \alpha^N)$ is collectively admissible if it is in $\A^N_{\mathbb{G}_N} \ce \prod_{i = 1}^N\A_{i; \mathbb{G}_N}$. When the underlying filtration is clear, we abbreviate the notation by dropping $\mathbb{G}_N$ and simply use $\A$ or $\A(t^*)$ for progressively measurable processes on $[0, T] \times \Omega$, and we use $\A_i$ or $\A^N$ for jointly measurable processes on $[0, T] \times \Omega \times [0, \eta]^N$.
\subsubsection{Varying Entry Equilibrium}
Suppose the random entry times $\bT = (\T^1, \dots, \T^N) \sim\nu^N$ have realizations $\bt = (t^1, \dots, t^N) \in [0, \eta]^N$. Let $\balpha = (\alpha^1, \dots, \alpha^N) \in \A^N$ denote the strategy profile of the $N$ traders. We assume that the traders are risk neutral. The objective that player $i$ tries to minimize is the expected total cost
\begin{equation}\label{Nobj}
    J^{N, i}(\balpha; \bt)\ce \E\sqbra{-(V_T^i - V_{t^i}^i) + \int_{t^i}^T\phi (X_t^i)^2dt + c(X_T^i)^2}
\end{equation}
over $\A(t^i)$, the set of admissible strategies for her entry time. 

\begin{remark}
    Our model tacitly assumes that there are noise traders in the market right from the beginning ensuring supply and demand on the asset. 
    Therefore, only the ``sophisticated'' investors $i\in \{1, \dots,N\}$ influence the growth of the bubble, and we do not impose a net-zero inventory condition.
    In particular, as often done in the literature on optimal execution (see e.g. \cite{AlmgrenandChriss01,CardaliaguetMFG18,bookAlgoandHFT15} and the references therein), the price process is not formed at equilibrium, but rather pre-specified, albeit controlled.
\end{remark}

Let $\mathcal{P}(A)$ denote the space of probability measures on $A$. We define the running costs before burst $f^+:[0, T]\times \R \times [0, 1] \times \mathcal{P}(A) \times A \to \R$ and after burst $f^-:[0, T]\times \R \times \mathcal{P}(A) \times A \to \R$ as
\begin{align}
\label{running_cost_f+}f^+(t,x,r, q, a) &= \kappa a^2 + \phi x^2 -\qv{\rho, q}_{F^N_\T}x -b(t, r)x \\
\label{running_cost_f-}f^-(t,x, q, a) &= \kappa a^2 + \phi x^2- \qv{\rho, q}_{F^N_\T}x,
\end{align}
At the time of burst $\tau^* \leq T$, the player lose an additional $X^i_{\tau^*}\beta_{\tau^*}\gamma_{\tau^*}$.  Since $A$ is bounded and $b, g$ continuous, $X^i_t$ and $P_t$ are square integrable on $[0, T]$. In particular, $\int X^idW^0$ and $\int PdW^i$ are martingales on $[0, T]$. Therefore, we can rewrite (\ref{Nobj}) as 
\begin{equation}\label{eq:Objective_N}
\begin{split}
J^{N,i}(\balpha; \bt) &\ce \E\left [ \int_{t^i}^{\tau^*}f^+(t, X^i_t, F^N_\T(t; \bt), \theta^N_t, \alpha^i_t(\bt))dt \right.\\
& \left .+ \int_{\tau^*}^T f^-(t, X_t^i, \theta_t^N, \alpha_t^i(\bt))dt + X^i_{\tau^*}\beta_{\tau^*}\gamma_{\tau^*}+ c(X_T^i)^2 \right].   
\end{split}
\end{equation}
Observe that integrating from $t^i$ is equivalent to integrating from $0$ because both $\alpha^i_t$ and $X^i_t$ are enforced to be $0$ before $t^i$, and $f^+(t, 0, r, q, 0) = f^-(t, 0, q, 0) = 0$ for all $(t, r, q) \in [0, T] \times \R \times \mathcal{P}(A)$. Similarly, in the trivial case where $\tau^* < t^i$ (i.e. the bubble bursts endogenously before entry), we have $J^{N,i} \equiv 0$. 

Given $\balpha \in \A^N$ and $\beta \in \A(t^i)$, denote by $\balpha^{\beta;-i}$ the strategy profile $$(\alpha^1, \dots, \alpha^{i-1}, \beta, \alpha^{i+1}, \dots, \alpha^N) \in \A^N.$$ 
We now define the Nash equilibrium for varying entry times.
\begin{definition}\label{Definition:NplayerNE*}
A strategy profile $\hat{\balpha} = (\hat{\alpha}^1, \dots, \hat{\alpha}^N) \in \A^N$ is a Nash equilibrium of the $N$-player game with \emph{varying entry time} if for $\nu^N$-almost every $\bt \in [0, \eta]^N$:
$$J^{N, i}(\hat{\balpha}; \bt) = \inf_{\beta \in \A(t^i)}J^{N,i}(\hat{\balpha}^{\beta;-i}; \bt) \text{ for every } i \in \{1, 2, \dots, N\}.$$ 
\end{definition}

In this paper we are particularly interested in the case where the number of agents is very large. Given that such problems are both mathematically and numerically intractable for large $N$, we consider the MFG counterpart which we now describe.

\subsection{The Limiting MFG}\label{Subsection:MFGModel} 
As $N\to \infty$, intuitively the individual impact becomes weaker. As a result, MFGs considers a single representative player who is playing against a \emph{distribution} of players. Let $(\Omega, \F, \P)$ be some probability space that supports a standard Brownian motion $W$ and mutually independent random variables $\iota, \tau$ with respective laws $\lambda_0, \nu$ representing the initial state and exogenous burst time, respectively.

Just as in Section \ref{Subsubsection:ProbSetup}, consider the product probability space $(\Omega \times [0, \eta], \F \otimes \B([0, \eta]), \P \otimes \nu)$. We naturally extend $\iota, \tau, W$ on the product space by putting $$\iota(\omega, t^*) = \iota(\omega), \quad \tau(\omega, t^*)= \tau(\omega), \quad W(\omega, t^*) = W(\omega)$$ so that the initial inventory, exogenous burst and Brownian motion retain their original laws. Define $\T(\omega, t^*) = t^*$ the random entry time of the representative player, which is independent from the other variables. Let $\mathbb{F}$ be the $\P$--completion of the natural filtration of the process $(\iota, W_t)_{t \in [0, T]}$ and let $\mathbb{G}$ be the smallest filtration containing $\mathbb{F}$ and making $\tau$ a totally inaccessible stopping time. Following Definition \ref{Definition:AdmissibleControlN}, we define admissibility of the MFG control.
\begin{definition}\label{Definition:AdmissibleControl}
An admissible strategy $\alpha: [0, T] \times \Omega \times [0, \eta] \to A $ for varying entry times is a $\mathscr{P}(\mathbb{G}) \otimes \B([0, \eta])$--measurable function such that for $\nu$--almost every $t^* \in [0, \eta]$, we have $\alpha({t^*}) \ce \alpha(\cdot, \cdot, t^*) \in \A(t^*)$. Let $\A^*$ denote all such strategies.
\end{definition}

The formulation of the control problem in the $N$-player game is known as the `` strong'' formulation. The associated MFG can be formulated similarly and naturally leads to solving a coupled system of forward-backward stochastic differential equations \cite{FBSDE13,CarmonaBookI, CarmonaBookII}. While the analysis of such systems is typically quite involved, it would be further complicated in our setting due to the dependence of $\tau^*$ in $\mu$ and the presence of random time $\tau$. Therefore, we choose to adopt the ``weak'' formulation of the game that is more suitable in the present case. It is well-known that under mild conditions, strong and weak formulations of stochastic control problems coincide (see e.g. \cite{WeakFormulationEquivalence22, LudoEquivalence20, KarouiNguyenJeanblanc87}).
In the context of MFGs, the weak formulation was first analyzed in the works of Carmona and Lacker \cite{CarmonaLacker15} and Campi and Fischer \cite{CampiFischer18}, from whom we borrow some intuition and ideas. We also refer to Possama\"i and Tangpi \cite{Possamai-Tangpi} for more recent developments. 

Let us define the driftless inventory 
\begin{equation}\label{eq:driftlessInv_model}
X_t \ce \iota + \sigma (W_t - W_\T) \text{ for } t \geq \T \quad \text{and} \quad X_t \ce 0 \text{ for } t \in [0, \T).    
\end{equation}
For $\alpha \in \A^*$, define the probability measure\footnote{ $\Q^1 \sim \Q^2$ denote that the two probability measures $\Q^1$ and $\Q^2$ are equivalent.} $\P^\alpha \sim \P \otimes \nu$ by  $$\frac{d\P^\alpha}{d\P \otimes \nu} \ce \mathcal{E}\pa{\int_0^\cdot \sigma^{-1}\alpha_sdW_s}_T,$$
and by $\mathcal{E}(\cdot)$ we denote the stochastic exponential of a martingale. In particular, if $M$ is a continuous $(\P, \mathbb{G})$-martingale starting from zero with quadratic variation $\qv{M}$, we have
\begin{equation*}
    \mathcal{E}(M)_t \ce \exp\pa{M_t - \frac{1}{2}\qv{M}_t}, \quad t \in [0, T].
\end{equation*}
By Girsanov's theorem and the boundedness of $A$, the process $W^\alpha$ defined by $$W^\alpha_t \ce W_t - \intt \sigma^{-1}\alpha_sds$$ is a $\P^\alpha$-Wiener process. The state process $X$ satisfies under $\P^\alpha$ the dynamics $$dX_t = \alpha_tdt + \sigma dW^\alpha_t \text{ for } t \geq \T \quad \text{and} \quad X_t = 0 \text{ for } t \in [0, \T).$$

Let $\X \ce C\pa{[0, T], \R}$ denote the space of all continuous functions from $[0, T]$ to $\R$ equipped with the sup--norm $\norm{\boldsymbol{x}}_\infty :=\sup_{t\in [0,T]}|\boldsymbol{x}_t|$ for $\boldsymbol{x}\in \X$. Let $D([0, T], \R)$ denote the space of all c\`adl\`ag functions from $[0, T]$ to $\R$. For $t^* \in {[0, \eta]}$, define 
\begin{equation*}
    \X^{t^*} \ce \Big\{\bx(t^*) \in D([0, T], \R): \bx_t(t^*)= 0 \text{ on } [0,t^*), \text{ continuous on } [t^*, T]\Big\}\text{ and } \X^{*} \ce \bigcup_{t^* \in [0, \eta]} \X^{t^*}.
\end{equation*}
For each $\bx(t^*) \in \X^*$, we require $t^*$ to be the largest value such that $\bx(t^*) \in \X^{t^*}$ to avoid redundancies. Suppose $t_1, t_2 \in {[0, \eta]}$. Let $\bx(t_1), \by(t_2) \in \X^{*}$. Notice that $d(\bx(t_1), \by(t_2) )= ||\bx(t_1) - \by(t_2)||_\infty$ is not a good metric on $\X^*$ because it does not allow two processes to be close unless $t_1 = t_2$. For each $\bx(t^*) \in \X^*$, we first define its continuous counterpart in $\X$ as
 $$\bar{\bx}_t(t^*) \ce \begin{cases}
    \bx_{t^*}(t^*) & t \in [0, t^*)\\
    \bx_t(t^*) & t \geq t^*.
\end{cases}
$$
We then define a more suitable metric on $\X^* \subset D([0, T], \R)$ to be 
\begin{equation}\label{supmetric}
    d_{\X^*}(\bx(t_1), \by(t_2)) \ce \|\bar{\bx}(t_1) - \bar{\by}(t_2)\|_\infty + |t_1 - t_2|.
\end{equation}
The metric space $(\X^*, d_{\X^*})$ where $X$ takes values is separable and complete (see Lemma \ref{Lemma:MetricSpace}). We sometimes also use the notation $\norm{\bx}_{\infty}$ for $\bx \in \X^*$ to mean $d_{\X^*}(\bx, \bzero)$.

Viewing the MFG as the $N \to \infty$ limit of the $N$--player game, we work with the natural limits of the empirical CDFs and empirical measures. In particular, the empirical laws $\mu^N$ and $\theta^N$ should be replaced by the law of $X$ and $\alpha$. For a Polish space $E$, let $\B(E)$ denote the Borel subsets of $E$, and let $\mathcal{P}(E)$ denote the set of all probability measures on $\B(E)$. Unless specified otherwise, we will equip $\mathcal{P}(E)$ with the topology of weak convergence of measures. 
Then $\mathcal{P}(E)$ is also a Polish space. Define the space of probability measures of inventory processes and controls as follows:
\begin{gather}
    \label{Definition:M}\M \ce \cbra{\mu \in \mathcal{P}(\X^*): \int_{\X^*} d_{\X^*}(\bx, \mathbf{0})\mu(d\bx) < \infty} \text{ and }\\
    \label{Definition:Theta}\Theta \ce \cbra{\theta \in \mathcal{P}([0, T] \times A) \text{ whose first projection is Lebesgue measure}}.
\end{gather}

Note that the disintegration theorem gives us a $dt$-almost everywhere unique family of probability measures with $\theta_t \in \mathcal{P}(A)$ such that $\theta(dt, da) = \theta_t(da)dt$. Suppose we are now given $(\theta, \mu) \in \Theta \times \M$ (as opposed to a given set of strategies and inventories of the other players in the finite player game). For $\mu \in \M$, let $\mu_t$ denote the image of $\mu$ under the $t$-projection $\X^* \ni \omega \mapsto \omega_t \in \R$. 

The empirical CDF of entry is also replaced by the CDF of $\nu$, simplifying the bubble component to $\gamma_t = \intt b(s, F_\T(s))ds$ which is no longer stochastic. Again, we assume that $F_\T(0) > 0$ to ensure that the bubble does not burst immediately. Intuitively, there should be initial players who are already trading the bubble asset in the beginning. Following \eqref{empirical_impact} and \eqref{empirical_meanmu}, we define the average impact and inventory among players in the game at time $t \in [0, T]$ as
\begin{equation}\label{eq:MFG_averages}
 \qv{\rho, \theta_t}_{F_\T} \ce \frac{1}{F_\T(t)}\int_A \rho(a)\theta_t(da), \quad \bar{\mu}_t \ce \frac{1}{F_\T(t)}\int_\R x\mu_t(dx), \quad (\theta, \mu) \in \Theta \times \M.
\end{equation}
We also define the endogenous burst time with respect to this average inventory $\bar{\mu}$, in analogy to the $N$-player version in Definition \ref{Def:N-playerendoburst}.
\begin{definition}\label{Def:MFGendoburst}
The endogenous burst time for the MFG is defined as $$\bar{\tau}(\mu) \ce \inf\cbra{t > 0: \inf_{s \in [0, t]}\bar{\mu}_s \leq \zeta_t} \wedge T$$
where $\zeta$ is a deterministic, continuous and strictly increasing function of $t$ with $\zeta_0 \in (0, \E[\iota])$.
\end{definition}
\noindent The actual burst time of the MFG is 
$$\tau^* \ce \bar{\tau}(\mu) \wedge \tau.$$
Let us recall from \eqref{Price-Dynamics} that the price process follows
\begin{equation*}
    P_t = P_0 + \intt b(s, F_\T(s))(1-D^*_{s})ds+\intt \qv{\rho, \theta_s}_{F_\T}ds +\sigma_0W_t^0 - \intt \beta_s\gamma_sdD^*_s.
\end{equation*}
Therefore, under self-financing condition the wealth process $V_t$ follows 
\begin{equation*}
    dV_t = \sqbra{-\kappa\alpha_t^2 + X_t\pa{b(t, F_\T(t))(1-D^*_{t})+ \qv{\rho, \theta_t}_{F_\T}}}dt-X_t\beta_t\gamma_tdD^*_t +\sigma_0X_tdW_t^0 + \sigma P_tdW^{\alpha}_t.
\end{equation*}
Given $(\theta, \mu) \in \Theta \times \M$, we then have
\begin{align}
    \notag
    J^{\theta, \mu}(\alpha) &\ce \E^{\P^\alpha}\sqbra{-(V_T - V_{0}) + \int_{\T}^T\phi (X_t)^2dt + c(X_T)^2} \\
    \label{Definition:Modelobj}
    &=\E^{\P^\alpha}\sqbra{g(X, \tau^*) + \int_\T^T f(t,X_t, \theta_t, \tau^*, \alpha_t)dt}
\end{align}
where the terminal cost $g(X, \tau^*)$ is the $\G_T$-measurable random variable 
\begin{equation}\label{terminalcost}
    g(X, \tau^*) \ce X_{\tau^*}\beta_{\tau^*}\gamma_{\tau^*}+ cX_T^2
\end{equation}
and $f: (t, x, q, \tau^*, a) \in [0, T] \times \R \times \mathcal{P}(A) \times \R_+  \times A \to \R$ is given by
\begin{equation}\label{runningcost}
    f(t,x, q, \tau^*, a) \ce \kappa a^2 + \phi x^2 - x\sqbra{b(t, F_\T(t))(1-D^*_{t})+ \qv{\rho, q}_{F_\T}}.
\end{equation}

\begin{definition}\label{Definition:ModelNE}
An equilibrium for the bubble riding MFG is a triplet $(\hat{\alpha}, \theta, \mu) \in \A^* \times \Theta \times \M$ such that $\hat{\alpha}$ minimizes the objective $J^{\theta,\mu}$ in \eqref{Definition:Modelobj} over $\A^*$ given $(\theta, \mu)$, and that $(\theta, \mu)$ satisfies the fixed point condition $(\theta_t, \mu) = (\theta_{t; \text{MFG}}, \mu_{\text{MFG}})\ dt$-a.e. where
   \begin{equation*}
       \mu_{\text{MFG}} \ce \P^{\hat{\alpha}}\circ \pa{X}^{-1} \text{ and }\ \theta_{t;\text{MFG}} \ce \P^{\hat{\alpha}}\circ \pa{\hat{\alpha}_t}^{-1} \text { for almost every } t \in [0, T].
   \end{equation*}
\end{definition}

\begin{remark}
Before the representative player enters, both her state $X$ and control $\alpha$ are kept at $0$. At entry, the player has initial state $\iota$ which is $\mathcal{F}_0$--measurable. Note that in the case of fixed entry time, the flow of measures $\mu$ of the state in equilibrium should satisfy $\mu_0 = \lambda_0$ by construction. 
This is no longer the case as the initial distribution is now ``diluted'' by the crowd of traders who have not entered the game and whose controls and states are kept at $0$.
\end{remark}
\begin{remark}\label{Remark:CommonNoise}
 One might argue that an admissible control should not only be $\mathbb{G}$-progressive, but also progressively measurable with respect to the common noise $W^0$. Our definition of admissible strategies $\A^*$ assumes that players do not take into account the common noise. 
   Considering common noise would drastically complicate the setup (see our accompanying paper \cite{TangpiWang23}).
   This measurability requirement of $\A^*$ is however not too restrictive.
     Indeed, since the running cost $f$, terminal cost $g$, and the dynamics of $X$ do not depend on the common noise, and additionally the players are risk-neutral, we show in Proposition \ref{Prop:CommonNoise} that any MFG equilibrium without common noise is a MFG equilibrium with common noise in our special setting.
     Therefore, it suffices to only consider $\mathbb{G}$-progressive controls. See Section \ref{Subsection:CommonNoise} for a more detailed discussion.
\end{remark}

\begin{assumption}{BT1}\label{Assumption:BT1model}
    The exogenous burst satisfies $\P(\tau > T) > 0$, and the compensator process\footnote{For more details, see Section \ref{Subsection:PEforExoBurst}.} $K$ for the jump process $D_t = \ind{\tau \leq t}$ takes the form $$dK_t = k_tdt$$ where the $\mathbb{F}$-predictable, non-negative intensity process $k$ is bounded on $[0, T]$.
\end{assumption}

\begin{remark}\label{Remark:BT1model}
Here we assume existence of a density for $\tau$, making $D$ a single jump Cox process. We take this standard assumption from \cite{KharroubiLim14}. If $ \tau$ has density function $f_{\tau}(t) = r_t\exp(\intt -r_sds)$ for a function $r_t$ non-negative on $[0, \infty)$, then Assumption \eqref{Assumption:BT1model} is satisfied with $k_t = r_t$. The burst intensity process $(k_t)_{t \in [0, T]}$ is assumed to be known by the players. 
    It represents the players' prior regarding the occurrence of an exogenous burst. 
\end{remark}

\begin{assumption}{MA}\label{Assumption:MA}~
\begin{enumerate}[label=(MA\arabic*)]
    \item  \label{Assumption:MA1} The set $A \subset \R$ is a compact interval that contains $0$.
    \item \label{Assumption:MA2}  The permanent price impact function $\rho: A \to \R$ is Borel measurable, locally bounded with $\rho(0) = 0$. The loss amplitude $\beta: [0, T] \to \R_+$ is continuous. The bubble trend function $b: [0, T] \times [0, 1] \to \R_+$ is Lebesgue integrable.
    \item \label{Assumption:MA3} The initial inventory $\iota$ has all moments.
\end{enumerate}
\end{assumption}
We state our main result for the bubble model.
\begin{theorem}\label{Theorem:ModelExistence}
    In addition to Assumptions \eqref{Assumption:BT1model} and \eqref{Assumption:MA}, also assume that $F_\T(0) > 0$ and $F_\T(\cdot)$ is continuous on $[0, \eta]$. Then there exists an equilibrium $(\hat{\alpha}, \theta, \mu)$ to the bubble riding MFG. Moreover, the optimal strategy can be decomposed as
\begin{equation}\label{eq:decomp.alpha.Theorem0}
        \hat{\alpha}_t(t^*) = \hat{\alpha}_t^+(t^*) \ind{t \leq \tau^*} + \hat{\alpha}^-_t(t^*,\tau^*) \ind{t > \tau^*}\quad dt\otimes\P\otimes\nu\text{--a.s.}
\end{equation}
    where the pre-burst control $\hat{\alpha}^+$ is $\mathscr{P}(\mathbb{F})\otimes \B([0, \eta])$-measurable, and the post-burst control $\hat{\alpha}^-$ is $\mathscr{P}(\mathbb{F})\otimes \B([0, \eta]) \otimes \B([0, T])$--measurable.
\end{theorem}

Before going any further, let us briefly comment on Theorem \ref{Theorem:ModelExistence} and its assumptions. Being an MFG equilibrium implies that the optimal control $\hat{\alpha}$ is $\mathscr{P}(\mathbb{G})\otimes \B{[0, \eta]}$-measurable.
Adapting existing fixed point arguments to the present setting allows to obtain a MFG equilibrium $\hat\alpha$ in the form of a $\mathbb{G}$--predictable process.
However, the filtration $\mathbb{G}$ is not a model input.
Equation \eqref{eq:decomp.alpha.Theorem0} allows us to decompose $\hat\alpha$ into two parts: before the burst occurs, the player follows the strategy $\hat\alpha^+$; after the burst, she uses $\hat\alpha^-(\cdot)$ which depends on the burst time. Both strategies are based only on the market information, i.e. the filtration $\mathbb{F}$.
Regarding the model assumptions,
\ref{Assumption:MA1} is typically assumed in MFGs, see e.g. \cite{CarmonaLacker15,CampiFischer18}.
\ref{Assumption:MA2} includes mild regularity conditions on the model inputs, and \ref{Assumption:MA3} is assumed to simplify the exposition.
Assumption \ref{Assumption:BT1} is discussed in Remark \ref{Remark:BT1model}. In addition, we assume that there exist initial players who enter at time $0$, who represent the force required to trigger the displacement phase of the bubble.

\begin{remark}\label{Remark:Approximation}
Regarding the link between the $N$--player game described in Section \ref{sec:Nplayer.game} and the MFG, it is well known that MFG equilibria in Theorem \ref{Theorem:ModelExistence} approximates the \emph{weak formulation} of the N-player game (see e.g. \cite[Theorem 4.2]{CarmonaLacker15}). In particular, a MFG equilibrium can be used to construct a so-called $\varepsilon$-approximate Nash equilibrium. We provide further details in Section \ref{Subection:Approximation}.
\end{remark}

Since the MFG considered in this paper appears not to have been studied in the literature so far (both for fixed and varying entry times), we propose to analyze it in a general framework. The proof of Theorem \ref{Theorem:ModelExistence} is given in Section \ref{Section:ModelRevisit} by applying Theorem \ref{Theorem:MFGExistence}, the generalized result.
The techniques developed below can undoubtedly apply to other examples of games where players enter the game at varying times or are subject to an unpredictable random shock. To better illustrate the interplay between the equilibrium strategies and the two types of bursts, we now provide some numerical observations on the fixed entry time equilibrium liquidation strategy under different settings.
\begin{figure}
    \centering
    \captionsetup[subfigure]{justification=centering}
    \begin{subfigure}[ht]{0.49\textwidth}
        \centering
        \includegraphics[width=\textwidth]{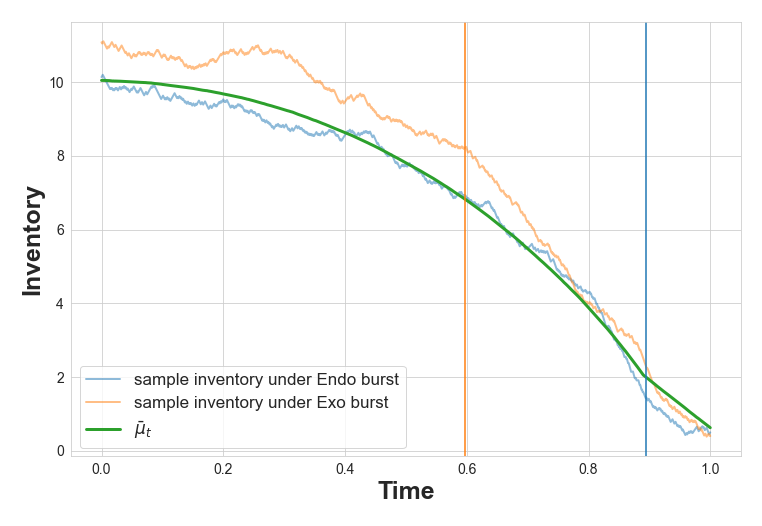}
        \caption{Inventory $X$ with Initial Distribution $N(10, 2)$}
        \label{fig:DefaultA}
    \end{subfigure}
        \hfill
    \begin{subfigure}[ht]{0.49\textwidth}
        \centering
        \includegraphics[width=\textwidth]{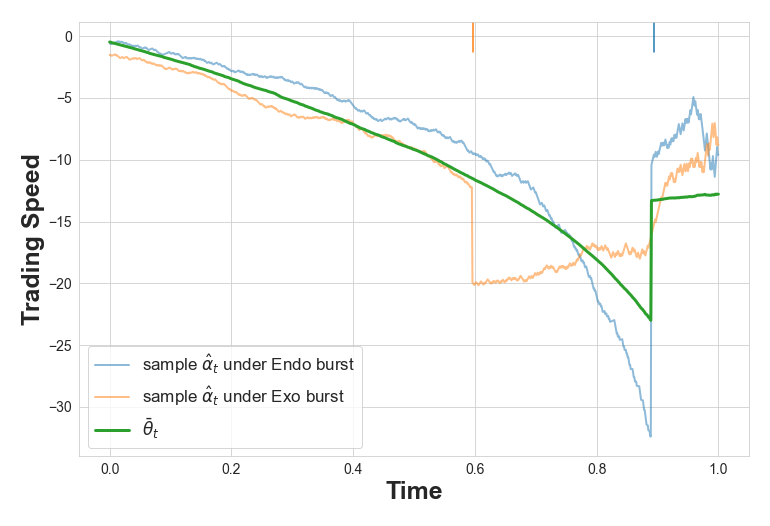}
        \caption{Signed Trading Speed $\alpha$}
        \label{fig:DefaultB}
    \end{subfigure}

    \caption{Equilibrium Strategies in Default Setting\\ Objective Value $J^{\theta, \mu}(\hat{\alpha})= -81.9$}
    \label{fig:Default}
\end{figure}

Anticipating a more detailed discussion on the simulation results in Section \ref{Subsection:Numerical}, notice that the presence of the bubble slows down the trading speed initially, as players attempt to ride the bubble first and collectively attack it later, which can be observed through the increasing selling speed before the endogenous burst (see endogenous burst sample path Figure \ref{fig:Default}). This bubble riding behavior leads to a concave-shaped inventory, in contrast to a convex inventory curve found in standard optimal execution MFGs (see \cite[Figure 3]{CardaliaguetMFG18}, \cite[Figure 6.2]{bookAlgoandHFT15} or our Figure \ref{fig:NoBubble} below).
On the other hand, when an early exogenous burst occurs, in which case the traders are caught off guard, we see a sudden increase in selling speed to avoid terminal penalty (see exogenous burst sample path in Figure \ref{fig:Default}). It turns out that the level of fear among the traders towards the likelihood of an exogenous shock plays a crucial role in the shape of their equilibrium strategies. For a more detailed discussion on the model results, the reader may skip the more technical sections below and go directly to Section \ref{Section:ModelRevisit}.


\section{General Framework for Fixed Entry Time}\label{Section:PE}
In this section, we temporarily fix the entry time and generalize the model inputs. Specifically, we set $\nu$ as the Dirac measure concentrated at $0$. 

\subsection{Progressive Enlargement of Filtration}\label{Subsection:PEforExoBurst}
Before we begin studying the general framework, let us say a few words on the appropriate filtered probability space on which $\tau$ is a totally inaccessible stopping time. The suitable filtration is obtained by progressive enlargement. Let us recall the filtered probability space $(\Omega, \F, \mathbb{F}, \P)$ from Section \ref{Subsection:MFGModel}. To construct a filtration under which $\tau$ is a totally inaccessible stopping time, define on $\R_+$ the jump process for exogenous burst time:
\begin{equation*}
    D_t \ce \ind{ \tau \leq t}.
\end{equation*}
Let $(\mathcal D_t)_{t \in [0,T]}$ denote the natural filtration of the exogenous burst time $(D_t)_{t\in [0, T]}$. Define the progressively enlarged filtration 
\begin{equation*}
    \G_t \ce \F_t \vee \mathcal D_t, \qquad \mathbb G = (\G_t)_{t \in [0, T]}.
\end{equation*}
Note that $\mathbb{G}$ is the smallest filtration which contains $\mathbb{F}$ and such that $\tau$ is a $\mathbb{G}$-stopping time. Since $\mathbb{F}$ is just the $\P$-completed Brownian filtration, by martingale representation theorem we have that any $\mathbb{F}$-martingale is continuous. By \cite{Azema93}, random time $\tau$ is a $\mathbb{G}$-totally inaccessible stopping time. Since we assume that $\tau$ is independent from $W$ and $\iota$, it follows from \cite[Proposition 1.3]{CarmonaBookII} that $\mathbb{F}$ is immersed in $\mathbb{G}$. In other words, the so-called (H)-hypothesis is satisfied: any square integrable $(\mathbb F, \P)$--martingale is a square integrable $(\mathbb G, \P)$-martingale. In particular, $W$ remains a ($\mathbb G, \P$)--Wiener process by L\'evy's characterization theorem. From now on, we work with the $\P$-completed filtered probability space $(\Omega, \F, \mathbb G, \P)$. Note that for any $\mathbb{G}$--predictable process $h$, there exists an $\mathbb{F}$-predictable process $f$ such that $h_t\ind{t \leq \tau} = f_t\ind{t \leq \tau}$. This process is unique as long as $\P(\tau > T) > 0$ (see \cite[Page 186 (a)]{bookDellacherie92}). Since $D$ is a $\mathbb{G}$--submartingale, by Doob--Meyer decomposition we can find a unique, $\mathbb{F}$--predictable, increasing compensator process $K$ with $K_0= 0$ and such that 
\begin{equation}\label{compensator MG}
M_t = D_t - \intt (1- D_{s-})dK_s
\end{equation}
is a $(\P, \mathbb{G})$--martingale, where $D_{t-} \ce \ind{\tau < t}$. Since $\T = 0$ almost surely, the set of admissible controls in Definition \ref{Definition:AdmissibleControl} effectively collapses to just $\A$, the set of $A$-valued, $\mathbb{G}$-progressively measurable processes. 

\subsection{Preliminaries, Assumptions, and Existence Result}\label{Subsection:weakfomulation}
Let us recall $\M$ and $\Theta$ introduced in \eqref{Definition:M} and \eqref{Definition:Theta}. Equip $\M$ with the topology induced by the Wasserstein metric
$$\W(\mu, \mu') \ce \inf_{\pi \in \Pi(\mu, \mu')}\int_{\X\times\X}\norm{\bx - \by}_\infty d\pi(\bx, \by)$$
where $\Pi(\mu, \mu')$ is the set of couplings of $(\mu,\mu')$.
Let $\bar\tau$ be a continuous function mapping $\M$ to $[0,T]$.
Our generalization of the game introduced in Section \ref{Subsection:MFGModel} is to solve the following MFG:
For any $(\theta,\mu) \in \Theta\times \mathcal{M}$, let $\hat\alpha$ solve
\begin{equation}
\label{eq:Def.MFG.general}
    \begin{cases}
        \inf_{\alpha\in \A}\E^{\P^{\alpha}}\Big[g(X, \bar\tau(\mu), \tau) + \int_0^Tf(s, X,\theta,\bar\tau(\mu),\tau, \alpha_s )\,ds \Big]\\
        dX_t = \sigma(t, X)\,dW_t\quad \frac{d\P^{\alpha}}{d\P} = \mathcal{E}(\int_0^\cdot \sigma^{-1}(s,X)b(s, X, \alpha_s)\,dW_s)_T,\quad \tau^*=\tau\wedge \bar\tau(\mu).
    \end{cases}
\end{equation}
A MFG equilibrium (for fixed entry time) is a triplet $(\hat\alpha, \theta,\mu) \in \A\times \Theta\times \mathcal{M}$ that solves the above system such that $\P^{\hat\alpha}\circ X^{-1} = \mu$ and $\P^{\hat\alpha}\circ \hat\alpha^{-1}_t = \theta_t$ for almost every $t \in [0, T]$. 

Observe that unlike the specific case of the bubble riding game, here we allow the coefficients of the game to be path--dependent.
In fact, we consider the following conditions:
\begin{assumption}{SD}\label{Assumption:SD}~
\begin{enumerate}[label={(SD\arabic*)}]
    \item \label{Assumption:SD1} The control space $A \subset \R$ is a compact interval that contains $0$;
    \item \label{Assumption:SD2} The function $b:[0, T] \times \X \times A \to  \R$ is continuous in $a$ for fixed $(t, \bx)$ and satisfies $b(t, \bx, \cdot) = b(t, \bx_{\cdot \wedge t}, \cdot)$ for all $(t, \bx)$. Moreover, there exists a constant $\ell_b > 0$ such that for all $(t, \bx, \bx', a) \in [0, T]\times \X \times \X \times A$,
    \begin{align*}
        \abs{b(t, \bx, a) - b(t, \bx', a)} &\leq \ell_b\pa{\norm{\bx_{\cdot \wedge t} - \bx'_{\cdot \wedge t}}_{\infty}}\\
        \abs{b(t, \bx, a)} &\leq \ell_b(1 + \norm{\bx_{\cdot \wedge t}}_\infty + |a|);
    \end{align*}
    \item \label{Assumption:SD3} The function
    $\sigma :[0, T] \times \X \to \R$
    satisfies $\sigma(t, \bx) = \sigma(t, \bx_{\cdot \wedge t})$ for all $(t, \bx)$, and there exists a constant $\ell_\sigma > 0$ such that for all $(t,\bx,\bx') \in [0, T] \times \X \times \X$,
    \begin{align*}
            |\sigma(t, \bx)- \sigma(t, \bx')| & \leq \ell_\sigma\norm{\bx_{\cdot \wedge t} - \bx'_{\cdot \wedge t}}_{\infty}\\
            |\sigma(t, \bx)| &\leq \ell_\sigma(1 + \norm{\bx_{\cdot \wedge t}}_{\infty});
    \end{align*}
    \item \label{Assumption:SD4} The distribution $\lambda_0$ of initial state $\iota$ has all moments. 
    For all $(t, \bx) \in [0, T] \times \X$ we have $\sigma(t, \bx) > 0$. Moreover, $\sigma^{-1}(t, \bx)b(t, \bx, a)$ is uniformly bounded in all variables. 
    \end{enumerate}
\end{assumption}
It is well known that under Assumption \ref{Assumption:SD3}, the driftless SDE
\begin{equation}\label{DriftlessState}
    dX_t = \sigma(t, X)dW_t, \qquad X_0 = \iota
\end{equation}
admits a unique strong solution (see e.g. \citet[Chapter V.3 Theorem 7]{bookprotter2005}), and $X$ is always referred to as the unique solution to \eqref{DriftlessState} throughout this section. 
Let us recall that by Girsanov's theorem, for all $\alpha\in \A$ the process $W^{\alpha}$ defined by 
\begin{equation*}
    W^\alpha_t \ce W_t - \intt \sigma^{-1}(s, X)b(s, X, \alpha_s)ds
\end{equation*} 
is a $(\P^\alpha, \mathbb{G})$--Brownian motion, and under $\P^\alpha$ the driftless state $X$ weakly solves the state equation $$dX_t = b(t, X, \alpha_t)dt + \sigma(t, X)dW_t^\alpha,\quad X_0 = \iota.$$
We make the following assumptions on $\tau$ and $\bar\tau$:
\begin{assumption}{BT} \label{Assumption:BT}\
\begin{enumerate}[label={(BT\arabic*)}]
    \item \label{Assumption:BT1} The exogenous burst satisfies $\P(\tau > T) > 0$. The compensator process $K$ defined in (\ref{compensator MG}) takes the form $$dK_t = k_tdt$$ where the $\mathbb{F}$-predictable, non-negative intensity process $k$ is bounded on $[0, T]$. 
    \item \label{Assumption:BT2} The endogenous burst time $\bar{\tau}: \M \to [0, T]$ is continuous.
\end{enumerate}
\end{assumption}

\noindent Now we specify the assumptions on the running cost $f$ and terminal cost $g$.

\begin{assumption}{C}\label{Assumption:C}~
\begin{enumerate}[label={(C\arabic*)}]
    \item \label{Assumption:C1} The running cost function $f: [0, T] \times \X \times \mathcal{P}(A) \times [0, T] \times \R_+ \times A \to \R$ is (jointly) Borel measurable and satisfies $f(t, \bx, \cdot, \cdot, \cdot) = f(t, \bx_{\cdot \wedge t}, \cdot, \cdot, \cdot)$. It can be decomposed as 
    \begin{equation*}
        f(t, \bx, q, s, \tau, a) = f^0(t, \bx_{\cdot \wedge t}, q, s, a) \ind{0 \leq t < \tau} + f^1(t,\bx_{\cdot \wedge t}, q, a) \ind{\tau \leq t}.
    \end{equation*}
    Furthermore, $f^0$ and $f^1$ satisfy the following separability properties:
    \begin{equation*}
        \begin{split}
            f^0(t, \bx, q, s, a) & = f_a(t, \bx, a) + f_b(t, \bx, q) + f_c(t, \bx, q)\ind{0 \leq t < s} \\
            f^1(t, \bx, q, a) & = f_a(t, \bx, a) + f_b(t, \bx, q).
        \end{split}
    \end{equation*}
    The functions 
    $f_a(t, \bx, \cdot)$, $f_b(t, \bx, \cdot)$ and $f_c(t, \bx, \cdot)$ are continuous for each $(t, \bx)$. In addition, there exists $p \geq 0$ and $\ell_f > 0$ such that for all $(t, \bx, q) \in [0, T] \times \X \times \mathcal{P}(A)$:
    $$|f_a(t, \bx, a)| + |f_b(t, \bx, q)| + |f_c(t, \bx, q)| \leq \ell_f\pa{1 + \norm{\bx_{\cdot \wedge t}}_\infty^p +\int_A |a|^p q(da)}.$$
    
    \item \label{Assumption:C2} 
    The terminal cost function $g:\X\times [0,T] \times \R_+ \to \R$ is (jointly) Borel measurable and can be decomposed as
    \begin{equation*}
        g(\bx, s, \tau) = g^0(\bx, s)\ind{\tau > T} + g^1(\bx,s \wedge \tau)\ind{\tau \leq T}.
    \end{equation*}
   Furthermore, the functions $g^0, g^1: \X\times [0, T] \to \R$ are continuous in the time variable, and there exists $p \geq 0 $ and $\ell_g > 0$ such that for all $(\bx, t) \in \X \times [0, T]$: $$|g^0(\bx, t)| + |g^1(\bx, t)| \leq \ell_g(1 + \norm{\bx}_{\infty}^p).$$

    \item \label{Assumption:C3} For each $(t, \bx, q,s, \tau,z) \in [0, T] \times \X \times \mathcal{P}(A) \times [0, T] \times \R_+ \times \R$, there is a unique $\hat{a} \in A$ such that 
    \begin{equation*}
        \hat a = \argmin_{a\in A}H(t, \bx, q,s, \tau, z, a),
    \end{equation*}
    where the Hamiltonian $H$ given by
    \begin{equation}\label{hamiltonian}
    H(t, \bx, q,s, \tau, z, a) = f(t, \bx,q, s, \tau, a) + b(t, \bx, a)\sigma^{-1}(t, \bx)z.
    \end{equation}    
\end{enumerate}
\end{assumption}
\noindent We denote the minimized Hamiltonian for each $(t, \bx,\theta,s, \tau, z)$ by
\begin{equation}\label{minhamiltonian}
    h(t, \bx,q,s, \tau, z) = \inf_{a \in A} H(t, \bx, q, z,s, \tau, a).    
\end{equation}
\begin{remark}\label{Remark:AssumptionC}
By \cite[Lemma 2.1 and Remark 2.1]{KharroubiLim14}, any $\mathbb{G}$-progressively measurable (resp. $\mathbb{G}$-predictable) process can be decomposed into before-and-after-$\tau$ components which are $\mathbb{F}$-progressively measurable (resp. $\mathbb{F}$-predictable). The simplifying condition \ref{Assumption:C1} additionally assumes that the running cost after burst does not explicitly depend on the burst time. 

The function $f_c$ is assumed to be an uncontrolled component that marks the change in cost dynamics after $\tau^*$. Consequently, the unique control in \ref{Assumption:C3} does not explicitly depend on $\bar{\tau}$ or $\tau$. Together with the separability property in \ref{Assumption:C1}, we implicitly require that $\hat{a}$ only depends on $t, \bx$ and $z$. 
Assumptions \ref{Assumption:SD1} and \ref{Assumption:C1} guarantee that the minimized Hamiltonian is always attained by $\hat{a}(t, \bx, z)$. We hence get the decomposition of $h$ as $h(t, \bx,q, \bar{\tau}, \tau, z) = h^0(t,\bx, q, \bar{\tau}, z)\ind{0\le t<\tau }  + h^1(t, \bx, q, z)\ind{\tau \leq t}$ with
\begin{equation}\label{decomposedminHam}
\begin{split}
    h^0(t,\bx, q, \bar{\tau}, z) & \ce H^0(t, \bx, q, \bar{\tau}, z, \hat{a}(t, \bx, z))\\
    h^1(t, \bx, q, z) & \ce H^1(t, \bx, q, z, \hat{a}(t, \bx, z))
\end{split}
\end{equation}
where
\begin{equation}\label{label:Hdecomposition}
    \begin{split}
    H^0(t, \bx, q, \bar{\tau}, z, a) & \ce f^0(t, \bx, q, \bar{\tau}, a) + b(t, \bx, a)\sigma^{-1}(t, \bx)z\\
    H^1(t, \bx, q, z, a) & \ce f^1(t, \bx, q, a) + b(t, \bx, a)\sigma^{-1}(t, \bx)z.
    \end{split}
\end{equation}
and \begin{equation*}
    H(t, \bx, q, \bar{\tau}, \tau, z, a) = H^0(t, \bx, q, \bar{\tau}, z, a)\ind{0 \leq t < \tau} + H^1(t,\bx,q,z, a)\ind{\tau \leq t}.
\end{equation*}
Note that under Assumption \ref{Assumption:SD4}, the functions $H^0$, $H^1$, $h^0$, and $h^1$ are all Lipschitz in $z$. 
\end{remark}
Because we assume that $f, H$ and $h$ depend on $\bar{\tau}$ and $\tau$ in a very specific way, we can omit these from the notation when there is no confusion. 
In this general setting, the existence theorem now takes the following form:
\begin{theorem}\label{Theorem:MFGExistence0}
Under Assumptions (\ref{Assumption:SD}), (\ref{Assumption:BT}), and (\ref{Assumption:C}), there exists a MFG equilibrium $(\hat{\alpha}, \theta, \mu)$ for fixed entry time. Moreover, the optimal control can be decomposed as 
$$\hat{\alpha}_t = \hat{\alpha}^+_t \ind{t \leq \tau^*} + \hat{\alpha}_t^-(\tau^*) \ind{t > \tau^*}$$
where $\hat{\alpha}^+$ and $\hat{\alpha}^-$ are $\mathscr{P}(\mathbb{F})$ and $\mathscr{P}(\mathbb{F}) \otimes \B([0, T])$-measurable, respectively.
\end{theorem}


\section{Existence Result for Fixed Entry Time: Proof of Theorem \ref{Theorem:MFGExistence0}}\label{Section:FP}

Throughout this section, we assume (\ref{Assumption:SD}), (\ref{Assumption:BT}), (\ref{Assumption:C}) and that $\T=0$.
The method used to prove Theorem \ref{Theorem:MFGExistence0} is based on BSDE arguments as customary in probabilistic approaches to stochastic control problems.
In the present case, we analyze the BSDEs on the \emph{enlarged filtration} $\mathbb{G}$. 
Therefore, it would be helpful to review relevant results of BSDE theory under progressive enlargement of filtration, notably well-posedness, comparison and a representation of solutions with respect to equations written in the small filtration $\mathbb{F}$. 
These results are slight extensions of results from \citet{KharroubiLim14}.
We provide the proofs in the appendix for completeness. 

\subsection{BSDE Results on Progressive Enlargement of Filtration}
\label{sec:BSDE-progress}
We begin by introducing a few notations for various spaces and norms. For a filtration $\mathbb H$, define the following spaces of processes on $[s, t] \subseteq [0, T]$:
\begin{itemize}
    \item Let $\mathcal{S}^2_{\mathbb H}[s,t]$ denote the space of $\R$-valued $\mathbb H$-progressively measurable, c\`{a}dl\`{a}g processes $Y$ on $[s,t]$ satisfying $$||Y||_{\mathcal{S}_\mathbb{H}^2} \ce \E\sqbra{\sup_{u \in [s, t]}\abs{Y_u}^2}^{\frac{1}{2}} < \infty.$$
    \item Let $\mathcal{H}^2_{\mathbb{H}}[s,t]$ denote $\R$-valued $\mathbb{H}$-predictable processes $Z$ on $[s, t]$ satisfying $$||Z||_{\mathcal{H}_\mathbb{H}^2} \ce \E\sqbra{\int_s^t |Z_u|^2du}^{\frac{1}{2}} < \infty.$$ 
    \item Let $\mathcal{H}_{\mathbb{H}, D}^2[s,t]$ denote $\R$-valued $\mathbb{H}$-predictable processes $U$ on $[s, t]$ satisfying $$||U||_{\mathcal{H}^2_{\mathbb{H},D}} \ce \E\sqbra{\int_s^t |U_u|^2dD_u}^{\frac{1}{2}} < \infty.$$ 
\end{itemize}
We drop $[s, t]$ from notation when considering the whole interval $[0, T]$. Let us recall the definition of $M$ in (\ref{compensator MG}). Consider a generic type of BSDEs on the enlarged filtration $\mathbb{G}$:
\begin{equation}\label{BSDE ProEnl0}
Y_t = \xi + \int_t^T G(s,Z_s)ds - \int_t^T Z_sdW_s - \int_t^T U_sdM_s
\end{equation}
where $\xi$ is a $\G_T$-measurable random variable, and $G$ is a $\mathscr{P}(\mathbb G)$-measurable function mapping $[0,T]\times \Omega\times \R$ to $\R$.

\begin{remark}\label{Remark:convention}
Note that on $\{t > \tau\}$ we have $$M_t = D_t - \intt (1-D_{s^-})dK_s = 1 - \int_0^\tau (1-D_{s^-})dK_s$$ which is constant. In fact it is equal to $M_0 = 0$ because $M$ is a $(\P, \mathbb{G})$ martingale. Additionally, since the driver $G$ does not depend on $U$, we can require that any solution to \eqref{BSDE ProEnl0} satisfies $$U_t = 0 \text{ for } t \in (\tau \wedge T, T].$$
\end{remark}

We first tackle the existence and uniqueness of solutions to BSDE (\ref{BSDE ProEnl0}). Let $\xi^0, \xi^1, G^0, G^1$ be the corresponding components from the respective decomposition of $\xi$ and $G$ in Remark \ref{Remark:AssumptionC}. By Assumption \ref{Assumption:BT1}, we can rewrite the BSDE \eqref{BSDE ProEnl0} as
\begin{equation}\label{BSDE ProEnl1}
Y_t = \xi + \int_t^T \sqbra{(G^0(s, Z_s)+U_sk_s)\ind{0\leq s < \tau} + G^1(s, Z_s, \tau)\ind{\tau \leq s}}ds - \int_t^T Z_sdW_s - \int_t^T U_sdD_s.
\end{equation}
We now impose conditions on the terminal condition and the driver.
\begin{assumption}{PEA}\label{Assumption:PEA}\ 
\begin{enumerate}[label={(PEA\arabic*)}]
    \item \label{Assumption:PEA1} The random variable $\xi^0 \in L^2((\Omega, \F_T, \P); \R)$. The mapping $\R_+ \ni \upeta \mapsto \xi^1(\upeta) \in L^2((\Omega, \F_T, \P); \R))$ is continuous, and
    $\sup_{\upeta \in [0, T]}\E\sqbra{|\xi^1(\upeta)|^2} < \infty$; \label{AssumptionPE:PEA1}
    \item \label{Assumption:PEA2}The function $G^0$ (resp. $G^1(\cdot)$) is Lipschitz continuous in $z$ uniformly in $t$ (resp. $t$ and $\tau$) with Lipschitz constant $\ell_G > 0$. 
    Moreover, they satisfy the integrability condition $$\sup_{\upeta \in [0, T]}\E\sqbra{\intT \abs{G^0(s, 0)}^2 + \abs{G^1(s, 0, \upeta)}^2ds} < \infty.$$ \label{AssumptionPE:PEA2}
\end{enumerate}
\end{assumption}
We are now ready to state the main BSDE results needed for our purpose.
\begin{theorem}[Existence and Uniqueness]\label{Theorem:BSDEmain}
    Under Assumptions \ref{Assumption:BT1} and \eqref{Assumption:PEA}, there exists a unique solution $(Y, Z, U) \in \mathcal{S}^2_{\mathbb G} \times \mathcal{H}_\mathbb{G}^2 \times \mathcal{H}_{\mathbb{G}, D}^2$ to the BSDE (\ref{BSDE ProEnl0}). 
    Moreover, $(Y,Z,U)$ satisfies the following decomposition for $\P\otimes dt$ almost every $(\omega,t)$:
    \begin{equation} \label{BSDEsolution}
    \begin{cases}
        Y_t = Y_t^0\ind{t < \tau} + Y_t^1(\tau)\ind{\tau \leq t}\\
        Z_t = Z_t^0\ind{t \leq \tau} + Z_t^1(\tau)\ind{\tau < t}\\
        U_t = (Y_t^1(t) - Y_t^0)\ind{t \leq \tau}
    \end{cases}
    \end{equation}
    where $(Y^0, Z^0)$ and $(Y^1(\upeta), Z^1(\upeta))$ are unique solutions in $\S^2_\mathbb{F} \times \H^2_\mathbb{F} \times \S^2_\mathbb{F}[\upeta \wedge T, T] \times \H^2_\mathbb{F}[\upeta \wedge T, T]$ to the following iterative system of Brownian BSDEs:
    \begin{equation}\label{BSDEsolutionBSDE}
    \begin{cases}
        Y_t^1(\upeta) = \xi^1(\upeta) + \int_t^T G^1(s, Z_s^1(\upeta), \upeta)ds - \int_t^TZ_s^1(\upeta)dW_s & \text{ for } \upeta \wedge T \leq t \leq T\\
        Y_t^0 = \xi^0 + \int_t^T\pa{G^0(s, Z^0_s) + k_sY_s^1(s) - k_sY_s^0}ds - \int_t^TZ^0_sdW_s & \text { for } 0 \leq t \leq T.
    \end{cases}
    \end{equation} 
    In particular, $Y^1(\cdot)$ and $Z^1(\cdot)$ are $\mathcal{P}(\mathbb{F}) \otimes \B(\R_+)$-measurable process.
\end{theorem}

\begin{proposition}[Comparison Principle]\label{Prop:BSDEcomparison}
    Let $(\xi,G), (\xi', G')$ be two sets of coefficients of the BSDE (\ref{BSDE ProEnl0}) that satisfy Assumption \eqref{Assumption:PEA}. Let $(Y, Z, U), (Y', Z', U')$ be their respective solutions in $\mathcal{S}^2_{\mathbb G} \times \mathcal{H}_\mathbb{G}^2 \times \mathcal{H}_{\mathbb{G}, D}^2$. If $\xi \leq \xi'$ $\P$--a.s. and for all $(t, z) \in [0, T]\times \R$ we have $G(t,z) \leq G'(t,z)$ $\P$--a.s., then $$Y_t\ \leq\ Y'_t \quad \P\text{--}\as \text{ for all } t \in [0, T].$$
\end{proposition}
\noindent The proofs are postponed to the appendix. Therein, we also prove a stability result which is useful for the proofs in the next section (see Proposition \ref{Prop:BSDEstability} in Appendix). 

\subsection{Spaces of Measures and Topologies}

We shall prove Theorem \ref{Theorem:MFGExistence0} by formulating the equilibria as fixed points to a mapping. In this part we introduce the appropriate domains of the mapping and analyze some of their properties.
\orange{
\begin{definition}\label{Definition:Mk0}
Let $K > 0$ be the uniform bound of $\sigma^{-1}b$ in Assumption \ref{Assumption:SD4}. Define $\M_K^0$ as the set of laws $\{\P^{\alpha} \circ X^{-1}: \alpha \in \A\}$. 
We equip $\M_K^0$ with the Wassertein distance.
\end{definition}}

Let us recall that $X$ is the unique solution to \eqref{DriftlessState}. Standard arguments give the following estimates, which we state as a lemma for later references.
\begin{lemma}\label{Lemma:fandXintegrability}
    Under Assumptions \eqref{Assumption:SD} and \eqref{Assumption:C}, for $p \geq 1$, we have $\E\Big[\|X\|^p_\infty\Big] < \infty$. Additionally, 
    \begin{equation*}
             \sup_{t \in [0, T], q \in \mathcal{P}(A), a \in A}\E\sqbra{\abs{f_a(t, X, a)}^2 + \abs{f_b(t, X, q)}^2 + \abs{f_c(t, X, q)}^2} < \infty.
     \end{equation*} 
     
\end{lemma}
\begin{lemma}\label{Lemma:Mk0PreCompact}
For every $K \geq 0$, the set $\M_K^0$ is convex and relatively compact in $\M$. 
\end{lemma}

\begin{proof}
We first show convexity. Take $\lambda \in (0, 1)$ and $\mu^1, \mu^2 \in \M_K^0$ with corresponding controls $\alpha^1, \alpha^2 \in \A$. Let us recall that $X$ is the driftless state process \eqref{DriftlessState}.
Since $\mu^1= \P^{\alpha^1}\circ X^{-1}$ and $\mu^2 = \P^{\alpha^2}\circ X^{-1}$ (with $\P^{\alpha}$ defined in \eqref{eq:Def.MFG.general}), it suffices to show that for each $\lambda \in (0,1)$ there is $\alpha \in \A$ such that
\begin{equation*}
    \frac{d\P^{\alpha}}{d\P} = \lambda \frac{d\P^{\alpha^1}}{d\P} + (1-\lambda)\frac{d\P^{\alpha^2}}{d\P}.
\end{equation*}
Define the (true, strictly positive) martingales $M^{\alpha^i}_t \ce \E[\frac{d\P^{\alpha^i}}{d\P}|\mathcal{F}_t], t \in [0, T]$ for $i \in \{1, 2\}$, and let $M_t \ce \E[\lambda \frac{d\P^{\alpha^1}}{d\P} + (1-\lambda)\frac{d\P^{\alpha^2}}{d\P} |\mathcal{F}_t] = \lambda M^{\alpha^1}_t + (1-\lambda)M^{\alpha^2}_t$. Applying It\^o's formula to $\log(M_t)$ yields
\begin{align*}
    M_t & = \exp\pa{\intt\sigma^{-1}(s, X)\frac{\lambda M^{\alpha^1}_sb(s, X, \alpha^1_s) + (1 - \lambda)M^{\alpha^2}_sb(s, X, \alpha^2_s)}{M_s}dW_s - \frac{1}{2}\intt \frac{\qv{M}_s}{M^2_s}ds}\\
    & = \mathcal{E}\pa{\int_0^\cdot\sigma^{-1}(s, X)\sqbra{\tilde{\lambda}_sb(s, X, \alpha_s^1) + (1-\tilde{\lambda}_s)b(s, X, \alpha_s^2)}dW_s}_t
\end{align*}
where $\tilde{\lambda}_s := \frac{\lambda M^{\alpha^1}_s}{M_s} \in (0, 1)$. Since $b$ is continuous in $a$ and $A$ is an interval, the image of $A$ under $b(t, \bx, \cdot)$ is also an interval in $\R$ for fixed $(t, \bx)$, hence also convex. Then $\P\otimes dt$ almost everywhere, the set $$\cbra{a \in A: b(t, X, a) = \tilde{\lambda}_t b(t, X, \alpha_t^1) + (1-\tilde{\lambda}_t)b(t, X, \alpha_t^2)}$$ is closed and non-empty. Therefore, a standard measurable selection argument allows us to find $\alpha \in \A$ such that $$M_T = \mathcal{E}\pa{\int_0^\cdot\sigma^{-1}(s, X)b(s, X, \alpha_s)dW_s}_T = \frac{d\P^{\alpha}}{d\P}$$ which implies convexity of $\M_K^0$.

It remains to show that $\M_K^0$ is relatively compact. We first show this in the weak topology, or equivalently, tightness by Prokhorov's theorem (see e.g. \cite[Theorem 16.3]{bookKallenberg02}). Take a sequence $\mu^n \in \M_K^0$ and let $\alpha^n$ be the corresponding control. For each $n \in \N$, define 
\begin{equation*}
    M^n_t = M_t(\alpha^n) \ce \mathcal{E}\pa{\int_0^{\cdot} \sigma^{-1}(s, X)b(s, X, \alpha^n_s) dW_s}_t.
\end{equation*}
By uniformly boundedness of $\sigma^{-1}b$, the process $M^n$ is a $({\P,\mathbb{G}})$--martingale. 
We define $\P^n$ by $\frac{d\P^n}{d\P} := M^n_T$. Therefore, for $n \in \N$ and $0 \leq s \leq t \leq T$, by Cauchy-Schwarz inequality,
\begin{equation*}
    \E^{\P^n}\sqbra{\abs{X_t - X_s}^4} \leq \sqrt{\E\sqbra{\pa{M_T^n}^{2}}} \sqrt{\E\sqbra{\abs{X_t - X_s}^8}}.
\end{equation*}
Standard arguments using Gr\"{o}nwall's inequality give $$\E\sqbra{(M^n_T)^{2}} \leq \exp(K^2T).$$
By BDG inequality, Assumption \ref{Assumption:SD3} and Lemma \ref{Lemma:fandXintegrability}, allowing $C > 0$ to change from line to line, we have:

\begin{align*}
    \E\big[\abs{X_t - X_s}^8\big] & \leq C\E\sqbra{\abs{\int_s^t|\sigma(u, X)|^2 du}^4} \leq C\E\sqbra{\abs{\int_s^t1 + \norm{X_{\cdot \wedge u}}^2_\infty du}^4}\\
    & \leq C(t-s)^4 + C\E\Big[(t-s)^4\sup_{u \in [s, t]}|X_u|^8\Big] \leq C(t-s)^4.
\end{align*}
\orange{Hence we have
$$
    \E^{\P^n}\sqbra{\abs{X_t - X_s}^4}\leq C(t-s)^2.
$$
}
By Kolmogorov-Chentsov tightness criterion (see  e.g. \cite[Corollary 16.9]{bookKallenberg02}), $\M_K^0$ is tight, hence weakly pre-compact in $\M$. 
Now using \cite[Theorem 6.9]{villani09}, Wasserstein pre-compactness follows if we show the uniform integrability property
\begin{equation}\label{label:uniformintegrability0}
    \lim_{R \to \infty}\sup_{\mu \in \M^0_K}\int_{\{\bx\in \X: \norm{\bx}_\infty > R\}} \norm{\bx}_\infty \mu(d\bx) =0.
\end{equation}
By definition of $\M_K^0$, it is equivalent to take the supremum over $\A$. For each $\alpha \in \A$, define $\P^{\alpha} \sim \P$ again with density $M_T(\alpha)$. Applying Cauchy-Schwarz inequality yields
\begin{equation*}
    \sup_{\alpha \in \A}\E^{\P^\alpha}\sqbra{\norm{X}_\infty^2} \leq \sup_{\alpha \in \A}\sqrt{\E\sqbra{\pa{M_T(\alpha)}^{2}}} \sqrt{\E\sqbra{\norm{X}_\infty^4}} < \infty.
\end{equation*}
Then (\ref{label:uniformintegrability0}) follows by dominated convergence theorem. 
\end{proof}

\begin{remark}\label{Remark:Theta}
For the control processes, a natural choice of spaces to work with would be $\Theta$ defined in (\ref{Definition:Theta}). However, the control processes do not have the required regularity for us to directly consider measures on $A$. 
It turns out to be more favorable if we look at $\mathcal{P}([0, T] \times \mathcal{P}(A))$ instead. We are now considering the space of \emph{relaxed controls} which contains the standard controls as Dirac measures. Therefore, a MFG equilibrium in Definition 
\ref{Definition:ModelNE} can be equivalently stated in terms of this larger space as well. It is worth noting that by considering relaxed controls, Campi and Fischer \cite{CampiFischer18} proved that $\M_K^0$ is indeed compact.
\end{remark}
We slightly abuse the notation and redefine $\Theta$ as the subset of $\mathcal{P}([0, T] \times \mathcal{P}(A))$ whose first projection is Lebesgue measure $dt$ on $[0, T]$. Again, any $\theta \in \Theta$ can be characterized, with $dt$ a.s. uniqueness, by $\{\theta_t \in \mathcal{P}(\mathcal{P}(A))\}_{t \in [0, T]}$ such that $\theta(dt, dq) = \theta_t(dq)dt$. We naturally extend any bounded measurable function $F: \mathcal{P}(A)\to \R$ to $\underline{F}: \mathcal{P}(\mathcal{P}(A)) \to \R$ by 
\begin{equation}\label{label:extendfunctions}
\underline{F}(\theta) \ce \int_{\mathcal{P}(A)}\theta(dq)F(q).    
\end{equation}
In particular, $\underline{F}(\delta_q) = F(q)$ for $q \in \mathcal{P}(A)$. Let us recall that Assumption \ref{Assumption:C1} ensures separability between $a$ and $q$ in the cost $f$. Therefore, we can drop the underline to simplify notation. Endow $\Theta$ with the stable topology, which is the weakest topology making the map $\theta \mapsto \int\phi d\theta$ continuous, for each bounded measurable function $\phi: [0, T] \times \mathcal{P}(A) \to \R$ that is continuous in the measure variable for each $t$. Since $A$ is convex, compact and metrizable, so is $\Theta$.

\subsection{Proof of Theorem \ref{Theorem:MFGExistence0}}\label{Subsection:MFGExistence0}
Now we come back to the proof of Theorem \ref{Theorem:MFGExistence0}. We shall go through the following four steps:
\begin{enumerate}
    \item For each given $(\theta, \mu)$, we first solve the optimization step in \eqref{eq:Def.MFG.general} using BSDEs and obtain some optimal control $\hat{\alpha}^{\theta, \mu} \in \A$.
    \item We show that the mapping $(\theta, \mu) \mapsto \hat{\alpha}^{\theta, \mu}$ is continuous in Lemma \ref{Lemma:alphaL2conv0}.
    \item We introduce a mapping whose fixed points correspond to MFG equilibria and show that it is continuous.
    \item Finish the proof by applying a fixed point theorem (Theorem \ref{Theorem:Brouwer}) and obtain the decomposition of optimal control into before and after burst time parts.
\end{enumerate}
\noindent\emph{Step 1: Obtaining Optimal Control.}\\
Let us recall the Hamiltonian $H$ defined by \eqref{hamiltonian} and the minimized Hamiltonian $h$ defined by \eqref{minhamiltonian}.
We first fix measure inputs $(\theta, \mu) \in \Theta \times \M$. For each $\alpha \in \A$, we consider the BSDE:
\begin{equation}\label{BSDEforH}
Y^\alpha_t = g(X,\bar\tau(\mu), \tau) + \int_t^T H(s, X, \theta_s, \bar\tau(\mu), \tau, Z^\alpha_s, \alpha_s)ds - \int_t^T Z^\alpha_sdW_s - \int_t^T U^\alpha_sdM_s
\end{equation}
where $M$ is the martingale defined by the compensator process for exogenous burst in \eqref{compensator MG}. 
Observe that this equation admits a unique solution $(Y^\alpha, Z^\alpha, U^\alpha)$.
This follows by Theorem \ref{Theorem:BSDEmain} since $H$ is Lipschitz--continuous in $z$, and
Assumptions \eqref{Assumption:C}, \eqref{Assumption:SD} together with Lemma \ref{Lemma:fandXintegrability} guarantee that Assumptions \eqref{Assumption:PEA} are satisfied. 
Let us recall $\P^\alpha \sim \P$ defined in \eqref{eq:Def.MFG.general}.
Changing the probability measure to $\P^\alpha$ yields
\begin{equation*}
    Y^\alpha_t = g(X,\bar\tau(\mu), \tau) + \int_t^T f(s, X, \theta_s, \bar\tau(\mu), \tau, \alpha_s)ds - \int_t^T Z^\alpha_s dW^\alpha_s -  \int_t^T U^\alpha_s dM_s,
\end{equation*}
and by martingale property, we have
\begin{align*} 
    J^{\theta,\mu}(\alpha) &= \E^{\P^{\alpha}}\bigg[g(X, \bar\tau(\mu), \tau) + \int_0^Tf(s, X,\theta,\bar\tau(\mu),\tau, \alpha_s )\,ds \bigg]\\
                          & = \E^{\P^{\alpha}}\sqbra{Y_0^\alpha} = \E\sqbra{Y_0^\alpha}.
\end{align*}
The last equality follows from the fact that $\P^{\alpha}$ agrees with $\P$ at $t = 0$. 
Now let $(\widehat{Y}^{\theta, \mu},\widehat{Z}^{\theta, \mu}, \widehat{U}^{\theta, \mu})$ be the unique solution (again obtained thanks to Theorem \ref{Theorem:BSDEmain}) of the BSDE
\begin{equation}
\label{eq:WeakMFG_BSDE}
    \widehat Y^{\theta, \mu}_t = g(X,\bar\tau(\mu), \tau) + \int_t^T h(s, X, \theta_s,  \bar\tau(\mu), \tau, \widehat Z^{\theta, \mu}_s)ds - \int_t^T \widehat Z^{\theta, \mu}_sdW_s - \int_t^T \widehat U^{\theta, \mu}_sdM_s.
\end{equation}
Since $h \leq H$ pointwise, we get that $\widehat{Y}^{\theta, \mu}_t \leq Y^\alpha_t$ $\P$--a.s. by Proposition \ref{Prop:BSDEcomparison}. 
Recalling Assumption \ref{Assumption:C3}, we can find $\hat{\alpha}^{\theta, \mu} \in \A$ where
    \begin{equation}
    \label{eq:charact.hat.alpha}
        \hat{\alpha}^{\theta, \mu}_t = \hat{a}(t, X, \widehat{Z}^{\theta, \mu}_t) = \argmin_{a \in A} H(t, X, \theta_t, \bar\tau(\mu), \tau, \widehat{Z}^{\theta, \mu}_t, a)\quad \P\text{--a.s.}
    \end{equation}
Since $h(s, X, \theta_s,  \bar\tau(\mu), \tau, Z^{\hat{\alpha}^{\theta, \mu}}_s) = H(s, X, \theta_s,  \bar\tau(\mu), \tau, Z^{\hat\alpha}_s, \hat{\alpha}^{\theta, \mu}_s)$, we have $$(\widehat Y^{\theta, \mu}, \widehat Z^{\theta, \mu}, \widehat U^{\theta, \mu}) = (Y^{\hat{\alpha}^{\theta, \mu}}, Z^{\hat{\alpha}^{\theta, \mu}},U^{\hat{\alpha}^{\theta, \mu}}).$$
In particular, for any other $\alpha \in \A$:
$$J^{\theta, \mu}(\hat{\alpha}) =  \E^{\P^{\hat{\alpha}^{\theta, \mu}}}\sqbra{Y_0^{\hat{\alpha}^{\theta, \mu}}} = \E\sqbra{Y_0^{\hat{\alpha}^{\theta, \mu}}} \leq \E\sqbra{Y_0^{\alpha}} =\E^{\P^{\alpha}}\sqbra{Y_0^{{\alpha}}} = J^{\theta, \mu}(\alpha).$$
Therefore $\hat{\alpha}^{\theta, \mu}$ is optimal for each $(\theta, \mu) \in \Theta \times \M$.\\

\noindent\emph{Step 2: Proving Continuity of $(\theta, \mu) \mapsto \hat{\alpha}^{\theta, \mu}$.}\\
We first show some properties on the function $\hat{a}$.
\begin{lemma}\label{Lemma:alphacont}
    The unique minimizer $\hat a$ of the Hamiltonian in \ref{Assumption:C3} satisfies $\hat{a}_t = \hat{a}(t, \bx, z)$ for a function $\hat a$ that is jointly measurable in all variables. For each $(t, \bx) \in [0, T] \times \X$, the function $z \mapsto \hat{a}(t, \bx, z)$ is continuous in $z$. 
    Moreover, the mapping $\mathcal{P}(A) \times \R \ni (q, z) \mapsto h(t, \bx, q, \bar\tau, \tau,  z)$ is also continuous for fixed $\bar\tau, \tau \in [0, T]\times \R_+$.
\end{lemma}

\begin{proof}
    The separability condition in \ref{Assumption:C1} ensures that the unique minimizer $\hat{\alpha}$ is a function of $(t, \bx, z)$. 
    The continuity conditions in \ref{Assumption:C1} ensure that for each $(t, \bx)$, the Hamiltonian is jointly continuous in $(q,z, a)$. 
    Since $A$ is compact, we can apply Berge's maximum theorem (e.g. \cite[Theorem 17.31]{bookInfDimAnalysis06}) to conclude that for each $(t, \bx, \bar\tau, \tau)$, $\hat{a}(t, \bx, \cdot)$ is continuous in $z$, and $h(t, \bx, \cdot,  \bar\tau, \tau, \cdot)$ is continuous. 
    In particular, $\hat{a}$ is a Carath{\'e}odory function, which is jointly measurable \cite[Lemma 4.51]{bookInfDimAnalysis06}.
\end{proof}

\begin{lemma}\label{Lemma:alphaL2conv0}
    The mapping $\Theta \times \M \ni (\theta, \mu) \mapsto \hat{\alpha}^{\theta, \mu}$ is continuous in  $\H^2_\mathbb{G}$.
\end{lemma}
\begin{proof}
    Take a sequence $\mu^n \ciw \mu$ in  $\mathcal{M}$ and $\theta^n \to \theta$ weakly in $\Theta$. Let us recall that the MFG endogenous burst times are deterministic, as opposed to their $N$-player game counterparts. By Assumption \ref{Assumption:BT2}, we get $\bar{\tau}(\mu^n) \nto \bar{\tau}(\mu)$, and so $\tau^*_n \nto \tau^*$ where $\tau^*_n \ce \bar{\tau}(\mu^n)\wedge \tau$ and $\tau^* \ce \bar{\tau}(\mu) \wedge \tau$. Let $\hat\alpha^{\theta^n,\mu^n}$ and $\hat\alpha^{\theta,\mu}$ be the optimal controls corresponding to $(\theta,\mu)$ and $(\theta^n,\mu^n)$, respectively. Let $(Y, Z, U), (Y^n, Z^n, U^n) \in \mathcal{S}^2_{\mathbb G} \times \mathcal{H}^2_{\mathbb G} \times \mathcal{H}^2_{\mathbb{G}}(D)$ denote the unique solutions to \eqref{eq:WeakMFG_BSDE} with data $(\theta, \mu)$ and  $(\theta^n, \mu^n)$, respectively. 

   Suppose that $(Z^n)_n$ converges to $Z$ in $\mathcal{H}^2_{\mathbb{G}}$. Then, $(Z^n)_n$ converges to $Z$ in $dt \times \P$ measure. 
   By Lemma \ref{Lemma:alphacont} and bounded convergence theorem, it follows that $(\hat\alpha^{\theta^n, \mu^n})_n$ converges to $\hat{\alpha}^{\theta, \mu}$ in $\mathcal{H}^2_{\mathbb{G}}$. 
   Therefore, it suffices to show
    \begin{equation*}
        \|Z^n- Z\|^2_{\mathcal{H}^2_{\mathbb{G}}} = \E\bigg[\intT|Z_t^n - Z_t|^2dt\bigg] \nto 0.
    \end{equation*}
    Let us recall that by Theorem \ref{Theorem:BSDEmain} and Remark 
    \ref{Remark:AssumptionC}, we have $Y_t = Y_t^0\ind{t < \tau} + Y_t^1(\tau)\ind{\tau \leq t}$ and $Z_t = Z_t^0\ind{t \leq \tau} + Z_t^1(\tau)\ind{\tau < t}$, where $(Y^0,Z^0)$ and $(Y^1(\cdot), Z^1(\cdot))$ satisfy the Brownian BSDEs
    \begin{equation}\label{eq:BSDEsolutionBSDE.Hamiltonian}
    \begin{cases}
        Y_t^1(\upeta) = g^1(X, \upeta \wedge \bar{\tau}(\mu)) + \int_t^T h^1(s, X, \theta_s^n, Z^1_s(\upeta))ds - \int_t^TZ_s^1(\upeta)dW_s & t \in [\upeta \wedge T, T]\\
        Y_t^0 = g^0(X, \bar{\tau}(\mu)) + \int_t^T \Big(h^0(s, X, \theta_s, \bar{\tau}(\mu), Z_s^0)\\
        \hspace{6cm}+ k_sY_s^1(s) - k_sY_s^0\Big)ds - \int_t^TZ^0_sdW_s & t \in[0, T].
    \end{cases}
    \end{equation} 
    For $\upeta \in [0, T]$, define
    $$\Delta_n g^0 \ce g^0(X, \bar{\tau}(\mu^n)) - g^0(X, \bar{\tau}(\mu)), \quad \Delta_n g^1(\upeta) \ce g^1(X, \upeta \wedge \bar{\tau}(\mu^n)) - g^1(X, \upeta \wedge \bar{\tau}(\mu)),$$ 
    $$\Delta_n h^0_s \ce h^0(s, X, \theta_s^n, \bar{\tau}(\mu^n), Z^0_s) - h^0(s, X, \theta_s, \bar{\tau}(\mu), Z_s^0)$$
    $$\Delta_n h^1_s(\upeta) \ce h^1(s, X, \theta_s^n, Z^1_s(\upeta)) - h^1(s, X, \theta_s, Z_s^1(\upeta)).$$
    With the decomposition of minimized Hamiltonian and Proposition \ref{Prop:BSDEstability}, we have for some $C > 0$
    \begin{equation}\label{uglyequationZ}
    \begin{split}
        \|Z^n- Z\|^2_{\mathcal H^2_{\mathbb{G}}} & \leq C\E\bigg[|\Delta_n g^0|^2 + \intT \abs{\Delta_n h^0_s}^2ds\bigg]\\ 
        &\quad  + C \E\bigg[ |\Delta_n g^1(\tau)|^2 + \int_{\tau \wedge T}^T \abs{\Delta_n h^1_s(s)}^2ds\bigg]\\
        &\quad + C\E\bigg[ \intT|\Delta_n g^1(\upeta)|^2 d\upeta + \intT\int_{\upeta \wedge T}^T \abs{\Delta_n h^1_s(\upeta)}^2ds d\upeta\bigg].
    \end{split}
    \end{equation}
    We check convergence of each term on the right hand side. Under Assumptions \ref{Assumption:BT2} and \ref{Assumption:C2}, since $(\mu^n)_n$ converges to $\mu$, by bounded convergence theorem we have
    \begin{equation*}
        \E\bigg[\abs{\Delta_n g^0}^2 + \abs{\Delta_n g^1(\tau)}^2 + \intT \abs{\Delta_n g^1(\upeta)}^2d\upeta\bigg] \nto 0.
    \end{equation*}
    By \ref{Assumption:C1} and Lemma \ref{Lemma:alphacont}, we can express $\Delta_nh^0, \Delta_nh^1$ as
\begin{equation*}
    \begin{split}
        \Delta_nh^0_s & = f_b(s, X, \theta^n_s) - f_b(s, X, \theta_s) + f_c(s, X, \theta^n_s)\ind{0 \leq s < \bar{\tau}(\mu^n)} - f_c(s, X, \theta_s)\ind{0 \leq s < \bar{\tau}(\mu)}\\
        \Delta_nh^1_s(\upeta) & = f_b(s, X, \theta^n_s) - f_b(s, X, \theta_s).
    \end{split}
\end{equation*}
For the term involving $\Delta_nh_s^0$,
\begin{subequations}
\begin{align}
    \E\bigg[\intT \abs{\Delta_nh^0_s}^2ds\bigg] & \leq \E\bigg[\intT \abs{f_b(s, X, \theta^n_s) - f_b(s, X, \theta_s)}^2ds\bigg] \label{Deltah0_1}\\
    &\quad  + \E\bigg[\intT \abs{f_c(s, X, \theta^n_s)\ind{0 \leq s < \bar{\tau}(\mu^n)} - f_c(s, X, \theta_s)\ind{0 \leq s < \bar{\tau}(\mu)}}^2ds\bigg] \label{Deltah0_2}.
\end{align}
\end{subequations}

We call the term on the right hand side of \eqref{Deltah0_1} $C_n$. Denote by $I_n$ the interval $[\bar{\tau}(\mu^n)\wedge \bar{\tau}(\mu), \bar{\tau}(\mu^n)\vee \bar{\tau}(\mu))$. It is easily verified that $\abs{\ind{0 \leq t < \bar{\tau}(\mu^n)}-\ind{0 \leq t < \bar{\tau}(\mu)}}^2 = \ind{t \in I_n}$. Using triangle inequality on \eqref{Deltah0_2} we get
\begin{align*}
    \E\bigg[\intT \abs{\Delta_nh^0_s}^2ds\bigg] & \le C_n + 2\E\bigg[\intT \abs{f_c(s, X, \theta_s^n)\ind{0 \leq s < \bar{\tau}(\mu^n)} - f_c(s, X, \theta_s)\ind{0 \leq s < \bar{\tau}(\mu^n)}}^2ds\bigg]\\
    & \qquad \ \ + 2\E\bigg[\intT \abs{f_c(s, X, \theta_s)\ind{0 \leq s < \bar{\tau}(\mu^n)} - f_c(s, X, \theta_s)\ind{0 \leq s < \bar{\tau}(\mu)}}^2ds\bigg]\\
    & \le C_n + 2\E\bigg[\intT \abs{f_c(s, X, \theta_s^n) - f_c(s, X, \theta_s)}^2ds\bigg]\\
    & \qquad \ \ + 2\E\bigg[\int_{I_n}\abs{f_c(s,X, \theta_s)}^2ds\bigg].
\end{align*}
By Assumption \ref{Assumption:BT2}, Lemma \ref{Lemma:fandXintegrability} and bounded convergence theorem, all three terms converge to $0$ as $n \to \infty$. 

As for $\Delta_nh^1_s$, note that it does not depend on $\upeta$ by Remark \ref{Remark:AssumptionC}. 
Therefore we have
\begin{equation*}
    \begin{split}
        & \E \bigg[\int_{\tau \wedge T} ^T \abs{\Delta_nh^1_s(s)}^2ds + \intT \int_{\upeta \wedge T}^T\abs{\Delta_nh^1_s(\upeta)}^2ds d\upeta\bigg] \leq (T+1)\E\bigg[\intT \abs{\Delta_nh_s^1}^2 ds\bigg]\\
        & = (T+1)\E\bigg[\intT\abs{f_b(s, X, \theta^n_s) - f_b(s, X, \theta_s)}^2ds\bigg] = (T+1)C_n \nto 0.
    \end{split}
\end{equation*}
Since all terms in (\ref{uglyequationZ}) converge to $0$, we have $\E\sqbra{\intT|Z_t^n - Z_t|^2dt} \nto 0$.
This concludes the argument. 
\end{proof}

\noindent \emph{Step 3: Existence of Fixed Points.}\\
We are now ready to apply the fixed point theorem. Define the mapping $\Psi: \Theta\times \M \to \Theta \times \M_K^0$ by 
\begin{equation}
\label{eq:def.Psi}
    \Psi(\theta, \mu) := \pa{\delta_{\P^{\theta, \mu} \circ \pa{\hat\alpha_t^{\theta,\mu}}^{-1}}(dq)dt, \quad \P^{\theta,\mu} \circ X^{-1}},
\end{equation}
where $\P^{\theta, \mu}$ is short for $\P^{\hat{\alpha}^{\theta, \mu}}$ defined in \eqref{eq:Def.MFG.general}. Notice that by construction, the mapping $\Psi$ indeed takes values in $\Theta \times \M_K^0$ where $K$ is the bound on $\sigma^{-1}b$ in Assumption \ref{Assumption:SD4}. A fixed point of $\Psi$ corresponds to a MFG equilibrium for fixed entry time. 
Let us recall the following version of Brouwer's Fixed Point Theorem (see e.g. \cite[Corollary 17.56]{bookInfDimAnalysis06}):
\begin{theorem}[Brouwer-Schauder-Tychonoff]\label{Theorem:Brouwer}
Let $\mathcal K$ be a non-empty compact convex subset of a locally convex Hausdorff space, and let $\Psi: \mathcal K \to \mathcal K$ be a continuous function. Then the set of fixed points of $\Psi$ is compact and nonempty.
\end{theorem}
We will use the following immediate corollary:
\begin{corollary}\label{Corollary:Brouwer}
Let $\mathcal{P}$ be a locally convex metric space and $\mathcal{K} \subset \mathcal{P}$ be convex and relatively compact in $\mathcal{P}$. If $\Psi: \mathcal{P} \to \mathcal{K}$ is continuous, then fixed points of $\Psi$ exist in $\mathcal{K}$.
\end{corollary}
\begin{proof}
    Let $\overline{\mathcal{K}} \subseteq \mathcal{P}$ be the closure of $\mathcal{K}$ in $\mathcal{P}$, which by assumption is compact and convex. 
    Let $\Psi_{\overline{\mathcal{K}}}$ be the restriction of $\Psi$ to $\overline{\mathcal{K}}$. 
    Apply Theorem \ref{Theorem:Brouwer} to get a fixed point of $\Psi_{\overline{\mathcal{K}}}$. 
    Since the mapping only takes values in $\mathcal{K}$, the fixed point has to be in $\mathcal{K}$. Hence it is also a fixed point for $\Psi$.
\end{proof}
In light of Corollary \ref{Corollary:Brouwer} and Lemma \ref{Lemma:Mk0PreCompact}, to this end, we only need to show that $\Psi$ is sequentially continuous on $(\Theta, \M)$.

\begin{lemma}\label{Lemma:PsiContinuous0}
    The function $\Psi$ is (sequentially) continuous.
\end{lemma}
\begin{proof}
    Take a sequence $$\Theta \times \M \ni (\theta^n, \mu^n) \nto (\theta, \mu) \in \Theta \times \M.$$
    For the control variable, we follow the proof of \cite[Theorem 3.5]{CarmonaLacker15} and show a stronger convergence of $\Psi(\theta^n, \mu^n)$ in total variation to $\Psi(\theta, \mu)$. Let us recall that the total variation distance $d_{TV}$ for two probability measures $\nu, \nu'$ is defined by $$\dst{TV}{\nu, \nu'} := \sup_{\phi: |\phi| \leq 1}\int\phi d(\nu - \nu').$$ 
    Let $\P^{\theta, \mu}:= \P^{\hat\alpha^{\theta,\mu}}$  and  $\P^{\theta^n, \mu^n}:= \P^{\hat\alpha^{\theta^n,\mu^n}}$ be the respective equivalent measures defined in \eqref{eq:Def.MFG.general}. 
    By Pinsker's inequality, it suffices to show that
    \begin{equation}
        \mathcal  H(\P^{\theta, \mu} \ | \ \P^{\theta^n, \mu^n})\nto 0,
    \end{equation}
    where $\mathcal{H}(\nu_1 | \nu_2)$ denotes the relative entropy
    \[\mathcal{H}\pa{\nu_1 | \nu_2} := \begin{cases}
        \int \log \frac{d\nu_1}{d\nu_2}d\nu_1 & \text{if } \nu_1 \ll \nu_2\\
        +\infty & \text{otherwise}.
    \end{cases}\]
    By definition of $\P^{\theta^n, \mu^n}$ and $\P^{\theta, \mu}$, we have 
    $$
        \frac{d\P^{\theta^n, \mu^n}}{d\P^{\theta, \mu}} = \mathcal{E}\pa{\int_0^\cdot \pa{\sigma^{-1}(t, X)b(t, X, \hat{\alpha}_t^{\theta^n, \mu^n}) - \sigma^{-1}(t, X)b(t, X, \hat\alpha_t^{\theta, \mu})}dW_t }_T.
    $$
    Denote by $\E^{\theta, \mu}$ and $\E^{\theta^n, \mu^n}$ the expectation under $\P^{\theta, \mu}$ and $\P^{\theta^n, \mu^n}$, respectively. 
    Since $\sigma^{-1}b$ is bounded, we compute
    \begin{align*}
        \mathcal H(\P^{\theta, \mu} \ | \ \P^{\theta^n, \mu^n}) & = - \E^{\theta, \mu}\bigg[\log \frac{d\P^{\theta^n, \mu^n}}{d\P^{\theta, \mu}}\bigg]\\
        & = \frac{1}{2}\E^{\theta,\mu}\bigg[\intT\sigma(t, X)^{-2}\abs{b(t, X, \hat{\alpha}_t^{\theta^n, \mu^n}) - b(t, X, \hat{\alpha}_t^{\theta, \mu})}^2dt\bigg].
    \end{align*}
    By Lemma \ref{Lemma:alphaL2conv0}, continuity of $b$ in $a$, Assumption \ref{Assumption:SD4} and the fact that $\P \sim \P^{\theta,\mu}$, we can apply bounded convergence theorem and conclude that $\mathcal H(\P^{\theta, \mu} \ | \ \P^{\theta^n, \mu^n}) \nto 0$. In particular,
    \begin{equation}\label{Palpha conv}
        \dst{TV}{\P^{\theta^n, \mu^n}\circ \pa{\hat\alpha_t^{\theta^n, \mu^n}}^{-1}, \P^{\theta, \mu}\circ \pa{\hat\alpha_t^{\theta^n, \mu^n}}^{-1}} \leq \dst{TV}{\P^{\theta^n, \mu^n}, \P^{\theta, \mu}} \nto 0, \quad dt\text{-a.e.}
    \end{equation}
    On the other hand, Lemma \ref{Lemma:alphaL2conv0} implies
    $$
        \P^{\theta, \mu} \circ (\hat\alpha_t^{\theta^n,\mu^n})^{-1} \nto \P^{\theta, \mu} \circ (\hat\alpha_t^{\theta,\mu})^{-1} \text{ in }dt\text{-measure}.$$
    Together with (\ref{Palpha conv}) and triangle inequality, we get $$\P^{\theta^n, \mu^n}\circ \pa{\hat\alpha_t^{\theta^n, \mu^n}}^{-1} \nto \P^{\theta, \mu} \circ (\hat\alpha_t^{\theta,\mu})^{-1} \text{ in $dt$-measure},$$
    which implies $\delta_{\P^{\theta^n, \mu^n} \circ\pa{\hat\alpha_t^{\theta^n,\mu^n}}^{-1}}(dq)dt \nto \delta_{\P^{\theta, \mu} \circ\pa{\hat\alpha_t^{\theta,\mu}}^{-1}}(da)dt$  in $\Theta$ by bounded convergence theorem. 
    
    Note that for the state variable, we need Wasserstein convergence in $\M$, which is not implied by convergence in total variation metric. \orange{Let $M^n$ (resp. $M^\infty$) denote the stochastic exponential that defines the density for $\P^{\theta^n, \mu^n}$ (resp. $\P^{\theta, \mu}$) with respect to $\P$. Then we have
    \begin{align*}
        \W&\pa{\P^{\theta^n, \mu^n} \circ X^{-1}, \P^{\theta, \mu} \circ X^{-1}} = \sup_{\psi \in \mathrm{Lip}(\X, 1)}\pa{\E^{\theta^n , \mu^n}[\psi(X)] - \E^{\theta, \mu}[\psi(X)]}\\
        & = \sup_{\psi \in \mathrm{Lip}(\X, 1), \psi(0) = 0}\E\sqbra{(M^n_T - M^\infty_T)\psi(X)} \leq \E\sqbra{|M^n_T - M^\infty_T|\norm{X}_{\infty}}.
    \end{align*}
    Using the inequality $|e^{a} - e^{b}| \leq |a - b||e^{a} + e^{b}|$ and H\"{o}lder's inequality, we can find some $\tilde{p} > 2$ and $C>0$ such that 
    \begin{align*}
        \E\sqbra{|M^n_T - M_T|\norm{X}_{\infty}} & \leq C\E\bigg[\intT\sigma(t, X)^{-2}\abs{b(t, X, \hat{\alpha}_t^{\theta^n, \mu^n}) - b(t, X, \hat{\alpha}_t^{\theta, \mu})}^2dt\bigg]^{1/2}\\
        & \qquad \times \pa{\E\sqbra{|M^n_T + M^\infty_T|^{\tilde{p}}}\E\sqbra{\norm{X}_\infty^{\tilde{p}}}}^{1/\tilde{p}}.
    \end{align*}
    Applying Ito's formula and Gr\"{o}nwall's inequality under Assumption \ref{Assumption:SD4} yields for all $p \in \N$ that
    $\sup_{n \in \N \cup \{\infty\}}\E\sqbra{|M^n_T|^p} < \infty$. By continuity of $b$, $\H^2$-convergence of the optimal controls and Lemma \ref{Lemma:uniformboundXtilde}, we can conclude by bounded convergence theorem.} Continuity of $\Psi$ is then obtained.
    \end{proof}
    \noindent\emph{Step 4: Decomposition of Optimal Control.}\\
    To conclude the proof of Theorem \ref{Theorem:MFGExistence0},
    we can thus apply Corollary \ref{Corollary:Brouwer} with $\mathcal{P} = \Theta \times \M$ and $\mathcal{K} = \Theta \times \M_K^0$, 
    this yields a fixed point $(\theta, \mu)\in \Theta\times\M_K^0$ of the mapping $\Psi$ with corresponding optimal control $\hat{\alpha} \in \A$.
    Hence, $(\hat{\alpha}, \theta, \mu)$ is a MFG equilibrium for fixed entry time.
    Moreover, by \eqref{eq:charact.hat.alpha} this optimal control satisfies $\hat\alpha_t = \hat a(t, X, Z_t)$, where $Z$ solves the BSDE \eqref{eq:WeakMFG_BSDE}.

    Note that the desired decomposition in Theorem \ref{Theorem:MFGExistence0} is by the actual burst time $\tau^*$, not the exogenous burst time $\tau$, with which we have been working. However, Lemma \ref{Lemma:tau*decomp} implies that the two decompositions coincide. In particular, we have $$\hat{\alpha}_t = \hat{a}(t, X, Z_t^0)\ind{0 \leq t \leq \tau^*} + \hat{a}(t, X, Z_t^1(\tau^*))\ind{t > \tau^*}$$
    with $(Z^0,Z^1(\cdot))$ solving \eqref{eq:BSDEsolutionBSDE.Hamiltonian}. The measurability conditions follow from Theorem \ref{Theorem:BSDEmain}. Hence Theorem \ref{Theorem:MFGExistence0} is proved.

\section{Existence Result for Varying Entry Times}\label{Section:RET}

In this section, we look at MFGs with varying entry times. Let us recall that players enter the game in the time frame $[0, \eta]$ independently from each other with law $\nu$. The entry of the players creates a jump in the state process. Therefore, we need to generalize our setup, and then we state and prove the existence result in this setting. The game is still solved on the enlarged filtration $\mathbb{G}$. In other words, when the agent enters the game, she knows whether the burst of the bubble has occurred. We do not force entry before burst, and the model reduces to a standard optimal execution problem on $\{\T \geq \tau^*\}$. 

\subsection{MFG Setup with Varying Entry Times}\label{Subsection:RESetup}

Let us recall from Section \ref{Subsection:MFGModel} that we start with the product probability space $(\Omega \times [0, \eta], \F \otimes \B([0, \eta]), \P \otimes \nu)$ that carries independent variables $\iota$ and $\tau$ with respective laws $\lambda_0$ and $\nu_\tau$. It also carries an independent Brownian motion $W$, and $\mathbb{F}$ is the filtration generated by the Brownian motion and the initial state, completed by $\P$-null sets.
Using the same procedure described in Section \ref{Subsection:PEforExoBurst}, we progressively enlarge this filtration to $\mathbb{G}$ and make the exogenous burst time $\tau$ a totally inaccessible stopping time. 

Naturally, we extend the domains of $b, \sigma$, $f$ and $g$ to allow for state trajectories in $\X^*$. Assumptions \eqref{Assumption:SD}, \eqref{Assumption:BT}, and \eqref{Assumption:C} are still well--defined, except that the driftless state equation in \eqref{DriftlessState} becomes (and is still denoted) $X$, the unique strong solution on the product space $[0,T]\times \Omega\times [0,\eta]$ of 
    \begin{equation}\label{DriftlessState*}
        dX_t = \sigma(t, X)dW_t, \quad X_{\T} = \iota, \quad \text{and} \quad X_t = 0 \text{ for } t \in [0, \T).
    \end{equation}
By Lemma \ref{Lemma:uniformboundXtilde}, for all $p \geq 1$ we have $\nu$-almost surely
\begin{equation}
    \label{fandXintegrability1}\E^{\P}\sqbra{ \norm{X}_{\infty}^p}  < \infty
\end{equation}
and thus, by \ref{Assumption:C1}, $\nu$-almost surely we have
\begin{equation}
    \label{fandXintegrability2}
    \sup_{t \in [0, T], q \in \mathcal{P}(A), a \in A}\E^\P\sqbra{\abs{f_a(t, X, a)}^2 + \abs{f_b(t, X, q)}^2 + \abs{f_c(t, X, q)}^2} < \infty.
\end{equation}

Now, we describe the MFG in the current random entry time setting. Let us recall the set of admissible controls $\A^*$ from Definition \ref{Definition:AdmissibleControl}. Given $\alpha \in \A^*$, $(\theta, \mu)\in \Theta\times \M$ and $X$ satisfying \eqref{DriftlessState*}, we put
\begin{equation}\label{Definition:Palpha*}
    \frac{d\P^{\alpha}}{d\P\otimes \nu} := \mathcal{E}\Big(\int_0^\cdot \sigma^{-1}(s, X)b(s, X, \alpha_s)d W_s \Big)_T
 \end{equation} 
 and consider the cost
\begin{equation}\label{Definition:MFGobjectiveProductSpace}
    J^{\theta, \mu}(\alpha) \ce \E^{\P^{\alpha}}\sqbra{g(X,\bar\tau(\mu), \tau) + \int_\T^Tf(s, X, \theta_s, \bar{\tau}(\mu), \tau, \alpha_s)ds}.
\end{equation}
Note that the probability measure $\P^{\alpha}$, in contrast to Section \ref{Section:PE}, is now a measure on the product space $\Omega\times[0,\eta]$.
Let $X(t^*)$ denote the conditional process $X(\cdot, \cdot, t^*)$. By the disintegration theorem, there is a family of measures $\P^{\alpha; t^*}$ on $\Omega$
such that
\begin{equation}\label{Definition:MFGobjective}
    J^{\theta, \mu}(\alpha) = \E^\nu\Biggl[\underbrace{\E^{\P^{\alpha; t^*}}\bigg[g(X(t^*),\bar\tau(\mu), \tau) + \int_{t^*}^Tf(s, X(t^*), \theta_s, \bar{\tau}(\mu), \tau, \alpha_s(t^*))ds\bigg]}_{\eqqcolon J^{\theta,\mu, t^*}(\alpha(t^*))}\Biggr].
\end{equation}
Conditioning on $\T = t^*$, the problem can be solved similarly to the fixed--entry case. \orange{We make a note that the uniqueness of the disintegration and the product structure of our probability space together imply that the family of measures $\P^{\alpha;t^*}$ coincides $\nu$-a.s. with $\P^{\alpha(t^*)}$, the probability measures resulting from Girsanov transformation corresponding to $\alpha$ conditioned on entry time $\T = t^*$. In this section, we use these two notations interchangeably. }

We will solve the problem $t^*$-by-$t^*$ and collect our conditional strategies by marginalization. The equality \eqref{Definition:MFGobjective} shows that the optimality of $\alpha$ is preserved if $\alpha(t^*)$ is optimal for $\nu$-almost every $t^*$. The joint measurability condition for admissibility requires justification (see Lemma \ref{Lemma:Zmeasurability}). We make the following continuity assumption to guarantee admissibility of the patched solution.
\begin{assumption}{ET}\label{Assumption:ET}
In addition to Assumption (\ref{Assumption:C}), we further assume that for fixed $(t, q, a) \in [0, T] \times \mathcal{P}(A)\times A$, the decomposed running costs $f_a(t, \cdot, a)$, $f_b(t, \cdot, q)$, and $f_c(t, \cdot, q)$, as well as the decomposed terminal costs $g^0(\cdot, t)$ and $g^1(\cdot, t)$, are continuous in $\bx$ on $(\X^*, d_{\X^*})$.
\end{assumption}
\begin{remark}\label{Remark:a_hatCont}
    Let us recall from Lemma \ref{Lemma:alphacont} that the minimizer function $\hat{a}$ of the Hamiltonian $H$ in \eqref{hamiltonian} is a measurable function of $(t, \bx, z)$ and continuous in $z$. If we assume \eqref{Assumption:ET} in addition to Assumptions \eqref{Assumption:C} and \eqref{Assumption:SD}, specifically the continuity of $f_a$ in $\bx$, then $H$ is jointly continuous in $(\bx, z)$. Therefore, by Berge's maximum theorem \cite[Theorem 17.31]{bookInfDimAnalysis06}, the minimizer function $\hat{a}$ is also jointly continuous in $(\bx, z)$ for fixed $t \in [0, T]$.
\end{remark}

Given joint measurability of optimal control, by the disintegration theorem, the following definition of MFG equilibrium is equivalent to Definition \ref{Definition:ModelNE} but more suitable for our $t^*$-to-$t^*$ approach.

\begin{definition}\label{definition:MFGequilibrium}
A MFG equilibrium for random entry time is a triplet $(\hat{\alpha}, \theta, \mu) \in \A^* \times \Theta \times \M$ such that for $\nu$-almost each $t^* \in [0,\eta]$, $\hat{\alpha}(t^*)$ minimizes the conditional objective $J^{\theta,\mu, t^*}(\alpha(t^*))$ in \eqref{Definition:MFGobjective}, and that $(\theta, \mu)$ satisfies the fixed point condition $(\theta_t, \mu) = (\theta_{t; \text{MFG}}, \mu_{\text{MFG}})\ dt$-a.e. where
   \begin{equation}\label{Definition:muMFG}
       \mu_{\text{MFG}}(B) \ce \int_{[0,\eta]} \P^{\hat{\alpha}(t^*)}\circ \pa{X(t^*)}^{-1} (B)\nu(dt^*),  \quad B \in \B(\X^*),
   \end{equation}
   and
   \begin{equation}\label{Definition:thetaMFG}
      \theta_{t;MFG}(E) \ce \int_{[0,\eta]} \P^{\hat{\alpha}(t^*)}\circ \pa{\hat{\alpha}_t(t^*)}^{-1}(E) \nu(dt^*), \quad E \in \B(A), \quad t \in [0, T].
   \end{equation} 
\end{definition}

We are now ready to state the main result of this section.
\begin{theorem}\label{Theorem:MFGExistence}
    Under Assumptions (\ref{Assumption:SD}), (\ref{Assumption:BT}), (\ref{Assumption:C}) and (\ref{Assumption:ET}), there exists a MFG equilibrium $(\hat{\alpha}, \Theta, \mu) \in \A^* \times \theta \times \M$ with random entry time. Moreover, the optimal control can be decomposed as 
    $$
        \hat{\alpha}_t(t^*) = \hat{\alpha}_t^+(t^*) \ind{t \leq \tau^*} + \hat{\alpha}^-_t(t^*,\tau^*) \ind{t > \tau^*}\quad dt\otimes\P\otimes\nu\text{--a.s.}
    $$
    where the pre-burst control $\hat{\alpha}^+$ is $\mathscr{P}(\mathbb{F})\otimes \B([0, \eta])$-measurable, and the post-burst control $\hat{\alpha}^-$ is $\mathscr{P}(\mathbb{F})\otimes \B([0, \eta]) \otimes \B([0, T])$--measurable.
\end{theorem}

\subsection{Preliminaries and Regularity Results}
As we mentioned before, we solve the problem $t^*$-by-$t^*$ and verify the fixed point conditions \eqref{Definition:muMFG} and \eqref{Definition:thetaMFG} using regular conditional probabilities. For $t^* \in [0,\eta]$, and $\alpha(t^*) \in \A(t^*)$, for simplicity we denote $\alpha^{t^*}:= \alpha(t^*)$. Just as Definition \ref{Definition:Mk0} in the fixed-entry section, we define the space of laws of the inventory process, for both $X^{t^*}$ and $X$.

With $K \geq 0$ being the bound in Assumption \ref{Assumption:SD4}, \orange{define
\begin{equation}\label{MKt*Definition}
    \M_K^{t^*} \ce \cbra{\P^{\alpha^{t^*}} \circ (X^{t^*})^{-1}: \alpha^{t^*} \in \A(t^*)} \subset \mathcal{P}(\X^{t^*}),
\end{equation}}
which is consistent with the Definition \ref{Definition:Mk0} of $\M^0_K$ with $t^* = 0$. In light of Lemma \ref{Lemma:Mk0PreCompact}, the set $\M^{t^*}_K$ is convex and pre-compact in $\mathcal{P}(\X^{t^*})$ for each $t^* \in [0, \eta]$. 
Similarly, for laws of the unconditional state processes, \orange{define
\begin{equation}\label{Definition:MK}
    \begin{split}
    \M_K & \ce \cbra{\P^{\alpha} \circ X^{-1}: \alpha \in \A^*} \subset \calP(\X^*).
    \end{split}
\end{equation}}
Note that for each $\alpha \in \A^*$ we can obtain its corresponding $\mu^\alpha = \P^{\alpha} \circ X^{-1} \in \M_K$ also by marginalizing the conditional measures given $\T$, namely $$\mu^\alpha(B) \ce \int_{[0,\eta]} \mu^{\alpha, t^*}(B)\nu(dt^*)  \quad \forall B \in \B(\X^*),$$
where \orange{$\mu^{\alpha, t^*} \ce \P^{\alpha^{t^*}} \circ (X^{t^*})^{-1}\in \M^{t^*}_K$} is the probability kernel for the state variable.

For the law of control process, for $t^* \in [0, \eta]$ define
\begin{equation*}
    \Theta^{t^*} \ce \cbra{\theta^{t^*} \in \Theta: \theta^{t^*}_t = \delta_0 \text{ for } t \in [0, t^*)}
\end{equation*}
where $\delta_0$ is the Dirac measure concentrated at $0$.

\subsection{Proof of Theorem \ref{Theorem:MFGExistence}}

Fix some $(\theta, \mu) \in \Theta \times \M$. For a given entry time $\T = t^*$, consider the BSDE on $[t^*, T]$ under the enlarged filtration $\mathbb{G}$
\begin{equation}\label{WeakMFG_BSDE*}
    Y^{\hat{\alpha}^{t^*}}_t = g(X^{t^*}, \bar{\tau}(\mu), \tau) + \int_t^T h(s, X^{t^*}, \theta_s, \bar{\tau}(\mu), \tau, Z^{\hat{\alpha}^{t^*}}_s)ds - \int_t^T Z^{\hat{\alpha}^{t^*}}_sdW_s - \int_t^T U^{\hat{\alpha}^{t^*}}_sdM_s.
\end{equation}
By Theorem \ref{Theorem:BSDEmain}, this BSDE is uniquely solvable with solutions in the form of \eqref{BSDEsolution}. 
By \emph{Step 1} in Section \ref{Subsection:MFGExistence0}, with $\T = t^*$ fixed, the strategy $\hat{\alpha}^{t^*} \in \A(t^*)$ obtained as the unique minimizer of the Hamiltonian \eqref{hamiltonian} along the solution $Z^{\hat{\alpha}^{t^*}}$ minimizes the conditional objective $J^{\theta, \mu}(\alpha, t^*)$ defined in \eqref{Definition:MFGobjective}. We can combine the $t^*$--by--$t^*$ optimal strategies by defining $\hat{\alpha}(t^*) \ce \hat{\alpha}^{t^*}$.
In order for $\hat{\alpha}$ to be admissible, we need to check that it is jointly measurable in $(t, \omega)$ and $t^*$.
We start with the following estimates.
\begin{lemma} \label{Lemma:unifboundZt*}
    Let $(Y^{\hat{\alpha}^{t^*}}, Z^{\hat{\alpha}^{t^*}}, U^{\hat{\alpha}^{t^*}}) \in \mathcal{S}^2_{\mathbb G}[t^*, T] \times \mathcal{H}_\mathbb{G}^2[t^*, T] \times \mathcal{H}_{\mathbb{G}, D}^2[t^*, T]$ denote the unique solution of \eqref{WeakMFG_BSDE*}. Then, we have the uniform bound 
    $$\sup_{t^* \in [0,\eta]}\E\bigg[|Y_t^{\hat{\alpha}^{t^*}}|^2 + \int_{t^*}^T \big|Z_s^{\hat{\alpha}^{t^*}}\big|^2ds\bigg] < \infty.$$
\end{lemma}

\begin{proof} 
    We apply the results in Section \ref{Subsection:PEforExoBurst} for each $t^* \in [0,\eta]$. 
    The solution $Y^{\hat{\alpha}^{t^*}}$ has the form similar to \eqref{BSDEsolution}, namely
    \begin{equation*}
        Y^{\hat{\alpha}^{t^*}}_t = Y_{t}^{0; \hat{\alpha}^{t^*}}\ind{t < \tau} + Y_t^{1; \hat{\alpha}^{t^*}}(\tau)\ind{\tau \leq t} 
    \end{equation*}
    where for all $\upeta \in [t^*, T]$, $(Y^{0; \hat{\alpha}^{t^*}}, Z^{0;\hat{\alpha}^{t^*}})$ and $(Y^{1;\hat{\alpha}^{t^*}}(\upeta), Z^{1;\hat{\alpha}^{t^*}}(\upeta))$ is the unique solution in $\S^2_\mathbb{F} \times \H^2_\mathbb{F} \times \S^2_\mathbb{F}[\upeta, T] \times \H^2_\mathbb{F}[\upeta, T]$ to the following iterative system of Brownian BSDEs: 
    \begin{equation}\label{eq:IterativeBSDE}
    \begin{cases}
        Y_t^{1; \hat{\alpha}^{t^*}}(\upeta) = g^1(X^{t^*}, \tau^*) + \int_t^T h^1(s, X^{t^*}, \theta_s, Z_s^{1; \hat{\alpha}^{t^*}}(\upeta))ds \\
        \hspace{3cm} - \int_t^TZ_s^{1; \hat{\alpha}^{t^*}}(\upeta)dW_s, & \upeta \leq t \leq T\\
        Y_t^{0; \hat{\alpha}^{t^*}} = g^0(X^{t^*}, \bar{\tau}(\mu)) + \inttT h^0(s, X^{t^*}, \theta_s, \bar{\tau}(\mu),Z^{0; \hat{\alpha}^{t^*}}_s)ds\\      \hspace{3cm} + \inttT k_s(Y_s^{1; \hat{\alpha}^{t^*}}(s) - Y_s^{0; \hat{\alpha}^{t^*}})ds
    - \int_t^TZ^{0; \hat{\alpha}^{t^*}}_sdW_s, &  t^* \leq t \leq T.
    \end{cases}
    \end{equation} 
    We bound each term uniformly in $t^*$. 
    In light of \eqref{fandXintegrability1} and \eqref{fandXintegrability2}, adopting the proof of Lemma \ref{integrabilityY1Z1} yields a bound $C_Y$ for $\E[|Y^{1; \hat{\alpha}^{t^*}}_t(\upeta)|^2]$ and a bound $C_Z$ for $\E\sqbra{\int_{\upeta \wedge T}^T \abs{Z_s^{1; \hat{\alpha}^{t^*}}(\upeta)}^2ds}$, both of which are uniform in $t^*, t$ and $\upeta$.

    Now we bound the expressions involving $Y^{0;\hat{\alpha}^{t^*}}$ and $Z^{0;\hat{\alpha}^{t^*}}$. Let $C_K$ denote the bound of $k_s$ on $[0, T]$. Apply It\^{o}'s formula to $(Y_t^{0; \hat{\alpha}^{t^*}})^2$ and take expectation. Let us recall that due to boundedness of $A$, the driver $h^0$ has linear growth in $z$. So there is $C > 0$ that depends on the bounds in \eqref{fandXintegrability1} and \eqref{fandXintegrability2} such that for any $\varepsilon > 0$ we have 
    \begin{align*}
        \E\sqbra{\abs{Y^{0; \hat{\alpha}^{t^*}}_t}^2} & = \E\sqbra{|g^0(X^{t^*}, \bar{\tau}(\mu))|^2 + \int_{t^*}^T 2Y_s^{0; \hat{\alpha}^{t^*}}h^0(s, X^{t^*}, \theta_s, \bar{\tau}(\mu),Z^{0; \hat{\alpha}^{t^*}}_s)ds}\\
        & +\E\sqbra{\int_{t^*}^T 2Y_s^{0; \hat{\alpha}^{t^*}}\pa{k_sY_s^{1;\hat{\alpha}^{t^*}}(s) - k_sY_s^{0;\hat{\alpha}^{t^*}}}ds} - \E\sqbra{\int_{t^*}^T \abs{Z^{0; \hat{\alpha}^{t^*}}_s}^2ds} \\
        & \leq \E\sqbra{|g^0(X^{t^*}, \bar{\tau}(\mu))|^2 + \int_{t^*}^T\frac{1}{\varepsilon}\abs{Y_s^{0; \hat{\alpha}^{t^*}}}^2+ \varepsilon\abs{h^0(s, X^{t^*}, \theta_s, \bar{\tau}(\mu),Z^{0; \hat{\alpha}^{t^*}}_s)}^2ds}\\
        &  + \E\sqbra{\int_{t^*}^T\frac{1}{\varepsilon}\abs{Y_s^{0; \hat{\alpha}^{t^*}}}^2+ \varepsilon\abs{k_sY_s^{1;\hat{\alpha}^{t^*}}(s)}^2 - 2k_s\abs{Y_s^{0; \hat{\alpha}^{t^*}}}^2ds - \int_{t^*}^T \abs{Z^{0; \hat{\alpha}^{t^*}}_s}^2ds}\\
& \leq C + \E\sqbra{\pa{\frac{2}{\varepsilon}+2C_K}\int_{t^*}^T \abs{Y_s^{0; \hat{\alpha}^{t^*}}}^2 ds} + TC_K^2C_Y^2\\
&+C\varepsilon T + (C\varepsilon - 1) \E\sqbra{\int_{t^*}^T \abs{Z^{0; \hat{\alpha}^{t^*}}_s}^2ds}.
\end{align*}
Choose $\varepsilon$ small enough such that $C\varepsilon \leq 1$ to remove the last term. For example, take $\varepsilon = 1/C$ and apply Gr\"{o}nwall's inequality. We have $$\E\sqbra{\abs{Y^{0; \hat{\alpha}^{t^*}}_t}^2} \leq (C+T + TC_K^2C_Y^2)\exp\pa{2T(C+C_K)}.$$
Call this $C_{YY}$. Plug this back into the previous inequality, and choose any $\varepsilon < \frac{1}{C}$ we get $$(1-C\varepsilon)\E\sqbra{\int_{t^*}^T \abs{Z^{0; \hat{\alpha}^{t^*}}_s}^2ds} \leq C+\pa{\frac{2}{\varepsilon}+2C_K}TC_{YY} + TC_K^2C_Y^2 + C\varepsilon T.$$
Divide both sides by $(1-C\varepsilon) > 0$, get a uniform bound which we call $C_{ZZ}$. Putting these together, we have 
$$\sup_{t^* \in [0,\eta]}\E\sqbra{|Y_t^{\hat{\alpha}^{t^*}}|^2 + \int_{t^*}^T \abs{Z_s^{\hat{\alpha}^{t^*}}}^2ds} \leq C_Y + C_{YY} + C_Z + C_{ZZ} < \infty.
$$
\end{proof} 

\begin{lemma}\label{Lemma:Zmeasurability}
    The patched solution $\hat{\alpha}$ given by $\hat{\alpha}(\cdot, t^*) = \hat{\alpha}^{t^*}(\cdot)$ is admissible.
\end{lemma}
\begin{proof} 
    Since $\hat{\alpha}^{t^*} \in \A(t^*)$ for every $t^* \in [0, t^*]$, we only need to show $\hat{\alpha}$ is $\mathscr{P}(\mathbb{G})\otimes \B([0, \eta])$-measurable. We first prove a regularity result for the mapping $t^* \mapsto Z^{\hat{\alpha}^{t^*}}$. For each $t^*$, extend $Z^{\hat{\alpha}^{t^*}}$ to a process in $\mathcal{H}_{\mathbb{G}}^2$ by setting $Z^{\hat{\alpha}^{t^*}}_s = 0$ for $s \in [0, t^*)$. Then, we claim that for any sequence $t^*_n \downarrow t^*$, the process $Z^{\hat{\alpha}^{t^*_n}}$ converges to $Z^{\hat{\alpha}^{t^*_n}}$ in $\P \otimes dt$-measure.
    Abbreviate $Z^n = Z^{\hat{\alpha}^{t^*_n}}$ and $Z^* =  Z^{\hat{\alpha}^{t^*}}$, where $\hat{\alpha}^{t^*_n}$ (respectively $\hat{\alpha}^{t^*}$) is the minimizer of Hamiltonian corresponding to $t^*_n$ (respectively $t^*$). 
    We show that each subsequence of $\{Z^n\}_{n \in \N}$ has a further subsequence converging to $Z$, $dt \otimes \P$-almost everywhere. By Lemma \ref{Lemma:Xt*cont}, $X^{t^*_n}$ converges to $X^{t^*}$ in probability on $\X^*$, so there exists a subsequence that converges almost surely. Without loss of generality, we take $\{t^*_n\}$ to be this subsequence that gives almost sure convergence for $X^{t^*_n}$. We then want to show that
    $$\limn \norm{Z^n - Z^*}_{\mathcal{H}_{\mathbb{G}}^2} = \limn \E\bigg[\intT|Z_t^n - Z_t^*|^2dt\bigg] = 0.$$ 
    First note that the integrand is $0$ on $[0, t^*)$ and $Z_t^n = 0$ on $[t^*, t^*_n)$. 
    We break the integral into two pieces: 
    $$\E\bigg[\intT|Z_t^n - Z_t^*|^2dt\bigg] = \E\bigg[\int_{t^*}^{t^*_n}|Z_t^* |^2dt\bigg] + \E\bigg[\int_{t^*_n}^{T}|Z_t^n - Z_t^*|^2dt\bigg].$$
    The first term goes to $0$ by Lemma \ref{Lemma:unifboundZt*} and bounded convergence theorem. For the second term, we first use Assumption \eqref{Assumption:ET}, \eqref{Assumption:C} and dominating convergence theorem to deduce
    \begin{equation*}
        \E\sqbra{\intT \abs{h^0(t, X^{t^*}, \theta_t, \bar{\tau}(\mu), Z_t^*) - h^0(t, X^{t^*_n}, \theta_t, \bar{\tau}(\mu), Z_t^*)}^2dt} \nto 0.
    \end{equation*}
    Using similar arguments on $h^1, g^0, g^1$, we can applying the stability result in Proposition \ref{Prop:BSDEstability} to conclude $\limn \norm{Z^n - Z^*}_{\mathcal{H}_{\mathbb{G}}^2} = 0$.
    
    Now let us recall from Lemma \ref{Lemma:alphacont} and Remark \ref{Remark:a_hatCont} that $\hat{a}$ is continuous in $(\bx, z)$. Therefore, by continuous mapping theorem, every subsequence of $\hat{\alpha}^{t^*_n}$ also has a further subsequence converging almost surely to $\alpha^{t^*}$. Then, if we view $\hat{\alpha}$ as a function on $\tilde{\Omega} \times [0, \eta] \to \R$ where $\tilde{\Omega} = [0, T] \times \Omega$, for all $t^* \in [0, \eta]$ it satisfies both of the following conditions:
    \begin{enumerate}
        \item $\hat{\alpha}(\cdot, t^*)$ is $\mathscr{P}(\mathbb{G})$-measurable.
        \item $\hat{\alpha}(\cdot, t^*_n)$ converges to $\hat{\alpha}(\cdot, t^*)$ in $dt \otimes \P$-measure for any decreasing sequence $t^*_n \downarrow t^*$.
    \end{enumerate}
    By (a slight variation of) \cite[Lemma A.2]{TangpiWang23}, we can conclude that $\hat{\alpha}$ is in fact $\mathscr{P} \otimes \B([0, \eta])$ measurable. The only difference is that \cite[Lemma A.2]{TangpiWang23} requires left instead of right continuity. The proof can be easily adapted for the other direction. Admissibility now follows by definition.
\end{proof}

Lemma \ref{Lemma:Zmeasurability} justifies solving the problem $t^*$-by-$t^*$. As in Remark \ref{Remark:Theta}, we redefine $\Theta$ as the subset of $\mathcal{P}([0, T] \times \mathcal{P}(A))$ whose first marginal is Lebesgue measure $dt$ on $[0, T]$. 
The set $\Theta$ is equipped with the stable topology. 
The control part of the MFG equilibrium in Definition \ref{definition:MFGequilibrium} corresponds to Dirac measures in this larger space. Extend functions on $\mathcal{P}(A)$ to $\mathcal{P}(\mathcal{P}(A))$ as in (\ref{label:extendfunctions}).

We can now define the mapping $\Psi: \Theta \times \M \to \Theta \times \M_K$ by
$$\Psi(\theta, \mu) := \pa{\delta_{\theta_{t; MFG}}(dq)dt, \quad \mu_{MFG}}$$
where $\mu_{MFG}$ and $\theta_{t; MFG}$ are defined in (\ref{Definition:muMFG}) and (\ref{Definition:thetaMFG}) respectively. In light of Corollary \ref{Corollary:Brouwer}, for existence of fixed points, it suffices to show that $\M_K$ is convex and relatively compact in $\M$, and $\Psi$ is continuous on $\Theta \times \M$. 
\begin{lemma}\label{Lemma:MkPreCompact}
$\M_K$ is convex and relatively compact in $\M$ for each $K \geq 0$.
\end{lemma}
\begin{proof}
    Convexity of $\M_K$ follows from the same argument used in the proof of Lemma \ref{Lemma:Mk0PreCompact} but applied on the product space. We prove relative compactness in the weak topology first. Let us recall from \eqref{supmetric} the metric space $(\X^*, d_{\X^*})$. \orange{Define the continuous counterpart $\overline{X}$ of the uncontrolled state process $X$ by $$\overline{X}_t \ce \iota \ind{t < \T} + X_t\ind{t \geq \T}.$$}
    Take a sequence $\mu_n \in \M_K$. By definition of $\M_K$ (see \eqref{Definition:MK}), let $\alpha^n \in \A^{*}$ be the corresponding control to $\mu_n$. 
    Let us recall from the proof of Lemma \ref{Lemma:Mk0PreCompact} that there exists a constant $C > 0$ independent of $t^*$ such that for all $n \in \N$ and all $t^*$ with $t^* \leq s \leq t \leq T$, we have
    \orange{\begin{equation*}
    \E^{\P^{\alpha^n(t^*)}}\Big[\big|\overline{X}^{t^*}_t - \overline{X}^{t^*}_s\big|^4\Big] = \E^{\P^{\alpha^n(t^*)}}\Big[\big|X^{t^*}_t - X^{t^*}_s\big|^4\Big] \leq C(t-s)^2.
    \end{equation*}}
    But observe that if $s \leq t^* \leq t \leq T$, we still have
    \begin{equation*}
    \E^{\P^{\alpha^n(t^*)}}\Big[\big|\overline{X}^{t^*}_t - \overline{X}^{t^*}_s\big|^4\Big] = \E^{\P^{\alpha^n(t^*)}}\Big[\big|X^{t^*}_t - X^{t^*}_{t^*}\big|^4\Big] \leq C(t-t^*)^2 \leq C(t-s)^2.
    \end{equation*} 
    Therefore the inequality holds in general for all $t^*$, which implies that
    \orange{\begin{equation*}
        \E^{\P^{\alpha^n}}\Big[\big|\overline{X}_t - \overline{X}_s\big|^4\Big] \leq C(t-s)^2,
    \end{equation*}
    so Kolmogorov--Chentsov tightness criterion (e.g. \cite[Corollary 16.9]{bookKallenberg02}) is satisfied. By Prohorov's theorem (see e.g. \cite[Theorem 16.3]{bookKallenberg02}) we have that the sequence $\P^{\alpha^n} \circ \overline{X}^{-1}$ is relatively compact in the weak topology. We can then extract a weakly converging subsequence indexed by $\{n_m\}_{m \geq 1}$.}
    
    \orange{Now we apply \cite[Theorem 13.5]{billingsleyBook99} to show that the subsequence $\mu_{n_m} \stackrel{m\to \infty}{\Longrightarrow} \mu_n$} where ``$\Longrightarrow$'' denotes weak convergence. The first condition on finite dimensional distributions convergence in \cite[Theorem 13.5]{billingsleyBook99} follows immediately from the convergence of the laws of the continuous counterparts. It remains to show that there exists $B \geq 0, A > 1/2$ and a non-decreasing, continuous function $\psi$ on $[0, T]$ such that for all $m \geq 1$, $\lambda > 0$ and $r \leq s \leq t \leq T$, one has
    \orange{\begin{equation}\label{eq:ConvCondition}
        \P^{\alpha^{n_m}}\pa{\abs{X_s - X_r} \wedge \abs{X_t - X_s} \geq \lambda} \leq \frac{1}{\lambda^{4B}}[\psi(t) - \psi(r)]^{2A}.
    \end{equation}}
    Observe that because there is only one jump, at least one of the two time segments will fall on the same side of the jump. Therefore, we always have \orange{$$\E^{\P^{\alpha^{n_m}}}\sqbra{\abs{X_s - X_r}^4 \wedge \abs{X_t - X_s}^4} \leq C(t-r)^2.$$}
    Then condition (\ref{eq:ConvCondition}) follows by Markov's inequality, with $A = 1, B = 1, \psi(t) = \sqrt{C}t.$
    
    It remains to show relative compactness in Wasserstein space. By \cite[Theorem 6.9]{villani09}, this follows from uniform integrability:
    \begin{equation}\label{label:uniformintegrability}
        \lim_{R \to \infty}\sup_{\mu \in \M_K}\int_{\{d_{\X^*}(\bx, \mathbf{0}) > R\}} d_{\X^*}(\bx, \mathbf{0}) \mu(d\bx) =0,
    \end{equation}
    which in turn follows from Lemma \ref{Lemma:uniformboundXtilde} and bounded convergence theorem.
\end{proof}

Now we show that $\Psi$ is sequentially continuous. Take a sequence $(\mu^n)$ converging to  $\mu$ in  $\mathcal{M}$ and $(\theta^n)$ converging to $\theta$ weakly in $\Theta$. Let $\alpha^{\theta, \mu} \in \A^*$ (resp. $\alpha^{\theta^n, \mu^n} \in \A^*$) denote the optimal control with inputs $(\mu, \theta)$ (resp. $(\mu^n, \theta^n)$). To simplify notation, again we let $\P^{\theta, \mu; t^*}$ (resp. $\P^{\theta^n, \mu^n; t^*}$) denote the probability measure $\P^{\alpha^{\theta,\mu}(t^*)}$ (resp. $\P^{\alpha^{\theta^n,\mu^n}(t^*)}$) defined in \eqref{eq:Def.MFG.general}.
We collect convergence results analogous to Lemma \ref{Lemma:alphaL2conv0} from the fixed entry time section.
\begin{lemma}\label{Lemma:conv}
The sequence of optimal strategies $\hat\alpha^{\theta^n, \mu^n; t^*} \nto \hat{\alpha}^{\theta, \mu; t^*}$ in $\H^2_\mathbb{G}$ uniformly in $t^*$.
Moreover, we also have the following convergence uniformly in $t^*$:
\begin{equation}
\label{Palphaconvinmeasure}\P^{\theta^n, \mu^n; t^*}\circ \pa{\hat\alpha_t^{\theta^n, \mu^n; t^*}}^{-1} \nto \P^{\theta, \mu; t^*}\circ \pa{\hat\alpha_t^{\theta, \mu; t^*}}^{-1} \text{ in } dt\text{-measure}.
\end{equation}
\end{lemma}
We omit the proof since it is essentially the same as that of Lemma \ref{Lemma:alphaL2conv0}. In particular, by (\ref{fandXintegrability1}), (\ref{fandXintegrability2}) and Lemma \ref{Lemma:uniformboundXtilde}, the bounds used in those proofs are uniform in $t^*$.
\begin{lemma}\label{Lemma:PsiContinuous}
The mapping $\Psi$ is (sequentially) continuous.
\end{lemma}
\begin{proof}
    We first show $\mu^n_{\text{MFG}} \ciw \mu_{\text{MFG}}$.
    Let us recall from \eqref{Definition:muMFG} that for all $B \in \B(\X^*)$,
   \begin{align*}
        \mu^n_{\text{MFG}}(B) & = \int_{[0,\eta]} \P^{\theta^n, \mu^n; t^*}\circ \pa{X^{t^*}}^{-1} (B)d\nu(t^*)\\
        \mu_{\text{MFG}}(B) & = \int_{[0,\eta]} \P^{\theta, \mu; t^*}\circ \pa{X^{t^*}}^{-1} (B)d\nu(t^*).
   \end{align*}
  \orange{Let $M^{n, t^*}$ (resp. $M^{\infty, t^*}$) denote the stochastic exponential that defines the density for $\P^{\theta^{n}, \mu^n; t^*}$ (resp. $\P^{\theta, \mu; t^*}$). A simple application of Gr\"{o}nwall's inequality yields $$\sup_{t^* \in [0,\eta], n \in \N^+}\E\sqbra{|M^{n, t^*}_T|^2} \leq \exp(K^2T)$$ where $K > 0$ is the bound of $|\sigma^{-1}b|$. Then we have
    \begin{align*}
        \W\pa{\mu^n_{\text{MFG}}, \mu_{\text{MFG}}} & = \sup_{\psi \in \mathrm{Lip}(\X^*, 1)}\pa{\int_{[0,\eta]}\E^{\theta^n , \mu^n; t^*}[\psi(X^{t^*})] - \E^{\theta, \mu; t^*}[\psi(X^{t^*})]d\nu(t^*)}\\
        & = \sup_{\psi \in \mathrm{Lip}(\X^*, 1), \psi(0) = 0}\int_{[0, \eta]}\E\sqbra{(M^{n, t^*}_T - M^{\infty, t^*}_T)\psi(X^{t^*})}d\nu(t^*) \\
        & \leq \int_{[0,\eta]}\sqrt{\E\sqbra{|M^{n, t^*}_T - M^{\infty, t^*}_T|^2}}\sqrt{\E\sqbra{\norm{X^{t^*}}_{\infty}^2}}d\nu(t^*) \nto 0,
    \end{align*}
where convergence follows from bounded convergence theorem.}

    Next, we turn to the control variables.
    Define, for a Borel subset $E$ of $A$,
    \begin{align*}
        \vartheta^n_{t}(E) & \ce \int_{[0,\eta]} \P^{\theta^n, \mu^n; t^*}\circ \pa{\hat{\alpha}^{\theta^n,\mu^n; t^*}_t}^{-1}(E) \nu(dt^*) \in \mathcal{P}(A)\\
        \vartheta_{t}(E) & \ce \int_{[0,\eta]} \P^{\theta, \mu; t^*}\circ \pa{\hat{\alpha}^{\theta,\mu; t^*}_t}^{-1}(E) \nu(dt^*) \in \mathcal{P}(A).
    \end{align*}
    Our goal is to show $\theta^n_{t; MFG}(dq, dt) \ce \delta_{\vartheta^n_{t}}(dq)dt$ converges to $\theta_{t; MFG}(dq, dt) \ce \delta_{\vartheta_{t}}(dq)dt$ in $\Theta$ with respect to the stable topology. 
    By bounded convergence theorem, it suffices to show that $\vartheta^n \nto \vartheta$ in $dt$ measure, or equivalently, every subsequence of $\vartheta^n_t$ has a further subsequence that weakly converges to $\vartheta_t$ for almost every $t$. Since $\theta^n$ and $\mu^n$ are arbitrary converging sequences, it suffices to find a converging subsequence. By \eqref{Palphaconvinmeasure} in Lemma \ref{Lemma:conv}, we can find a subsequence $(n_m)_{m \geq 1}$ such that for all $t^* \in [0,\eta]$, the conditional laws $\P^{\theta^{n_m}, \mu^{n_m}; t^*}\circ \pa{\hat{\alpha}^{\theta^{n_m},\mu^{n_m}; t^*}_t}^{-1}$ converge weakly in $\mathcal{P}(A)$ to $\P^{\theta, \mu; t^*}\circ \pa{\hat{\alpha}^{\theta,\mu; t^*}_t}^{-1}$ for almost every $t \in [0, T]$. For any $t \in [0, T]$ such that this weak convergence holds, for any open set $G \subset A$, by Fatou's lemma we have 
    \begin{align*}
        \liminf_{m\to \infty} \vartheta_t^{n_m}(G) & = \liminf_{m\to\infty}\int_{[0, \eta]}\P^{\theta^{n_m}, \mu^{n_m};t^*}\circ \pa{\hat{\alpha}^{\theta^{n_m},\mu^{n_m}; t^*}_t}^{-1}(G) \nu(dt^*)\\
        & \geq \int_{[0, \eta]}\liminf_{m\to\infty}\P^{\theta^{n_m}, \mu^{n_m}; t^*}\circ \pa{\hat{\alpha}^{\theta^{n_m},\mu^{n_m}; t^*}_t}^{-1}(G) \nu(dt^*)\\
        & \geq \int_{[0,\eta]} \P^{\theta, \mu; t^*}\circ \pa{\hat{\alpha}^{\theta,\mu; t^*}_t}^{-1}(G) \nu(dt^*) =  \vartheta_t(G)
    \end{align*}
    which is equivalent to weak convergence in $\mathcal{P}(A)$.
\end{proof}
We can now apply Corollary \ref{Corollary:Brouwer}. The pre-and-post burst decomposition of the optimal control given $t^*$ can be argued similarly to step 4 of the proof in Section \ref{Subsection:MFGExistence0}. Specifically, for each $t^* \in [0, \eta]$, we have $$\hat{\alpha}^{+}_t(t^*) \ce \hat{a}(t, X^{t^*}, Z^{0; t^*}) \quad \hat{\alpha}^{-}_t(t^*, \upeta) \ce \hat{a}(t, X^{t^*}, Z^{1; t^*}(\upeta))$$
where $Z^{0; t^*}$ and $Z^{1; t^*}(\cdot)$ are part of the unique solution to the iterative system of BSDEs \eqref{eq:IterativeBSDE} solved on $[t^*, T]$. By Theorem \ref{Theorem:BSDEmain}, $\hat{\alpha}^{+}(t^*)$ and $\hat{\alpha}^{-}(t^*, \cdot)$ are $\mathscr{P}(\mathbb{F})$ and $\mathscr{P}(\mathbb{F}) \otimes \B([0, T])$-measurable, respectively. Following the same argument as in Lemma \ref{Lemma:Zmeasurability}, we can also obtain the joint measurability requirements for the patched, decomposed components $\hat{\alpha}^{+}$ and $\hat{\alpha}^{-}$. Hence Theorem \ref{Theorem:MFGExistence} is proved.

Note that there is no explicit assumption on the distribution of entry times $\nu$. We now prove a simple consequence of imposing a continuity condition on $\nu$, which will be used in the following section when we revisit the bubble model. 
\orange{
\begin{lemma}\label{Lemma:Xcont}
If $\nu$ is atomless on $(0, \eta]$, then $t \mapsto \E^{\P^{\alpha}}[X_t]$ is continuous for every $\alpha \in \A^*$.
\end{lemma}
\begin{proof}
    Take a sequence $t_n \to t \in [0, T]$. Without loss of generality, we can take $t > 0$. Since for each $t^*$, the process $X^{t^*}$ is a.s. continuous except at $t = t^*$, the atomless condition for $\nu$ implies
    \begin{align*}
        \limn \E^{\P^{\alpha}}\sqbra{X_{t_n}} &= \limn \int_{[0, \eta]}\E^{\P^{\alpha; t^*}}\sqbra{X^{t^*}_{t_n}}d\nu(t^*) = \int_{[0, \eta]} \E^{\P^{\alpha; t^*}}\sqbra{\limn X^{t^*}_{t_n}\ind{t^* \neq t}}d\nu(t^*)\\
        & = \int_{[0, \eta]} \E^{\P^{\alpha; t^*}}\sqbra{X^{t^*}_{t}\ind{t^* \neq t}}d\nu(t^*) = \int_{[0, \eta]}\E^{\P^{\alpha; t^*}}\sqbra{X^{t^*}_{t}}d\nu(t^*) = \E^{\P^{\alpha}}\sqbra{X_{t}}.
    \end{align*}
\end{proof}
}

\section{Revisiting Bubble Model}\label{Section:ModelRevisit}

Let us now come back to the bubble riding game introduced in Section \ref{Subsection:MFGModel} that motivates the MFG studied in Sections \ref{Section:PE} and \ref{Section:RET}.
In this final section of the article, we provide the proof of Theorem \ref{Theorem:ModelExistence} and discuss the numerical simulations alluded to in Section \ref{Subsection:MFGModel}. Then, we discuss the link between the MFG and the motivating, finite-player model.

\subsection{Revisiting Theorem \ref{Theorem:ModelExistence}}
The proof of Theorem \ref{Theorem:ModelExistence} is a direct application of Theorem \ref{Theorem:MFGExistence}. We state and prove the full version of the theorem below. Suppose the closed interval of allowed (or practical) trading speed is $A = [\underline A, \overline A]$, with $\underline A \leq 0 \leq \overline A$. Define $A' = [-2\kappa\sigma\overline A, -2\kappa \sigma\underline A]$. For $z \in \R$, denote by $z|_{A'}$ the projection of $z$ on the interval $A'$.
\begin{theorem}
\label{thm.bubble.random.entry}
    In addition to Assumptions \eqref{Assumption:MA} and \eqref{Assumption:BT1model}, assume that $F_\T(0) > 0$, and that $F_\T(\cdot)$ is continuous on $[0, \eta]$. Then, there exists a MFG equilibrium $(\hat{\alpha}, \theta, \mu)$ for the bubble riding game with varying entry time. 
    Moreover, the optimal control $\hat{\alpha}(\cdot)$ can be decomposed as 
    $$
        \hat{\alpha}_t(t^*) = \hat{\alpha}_t^+(t^*) \ind{t \leq \tau^*} + \hat{\alpha}^-_t(t^*,\tau^*) \ind{t > \tau^*}\quad dt\otimes\P\otimes\nu\text{--a.s.}
    $$
    where the pre-burst control $\hat{\alpha}^+$ is $\mathscr{P}(\mathbb{F})\otimes \B([0, \eta])$-measurable, and the post-burst control $\hat{\alpha}^-$ is $\mathscr{P}(\mathbb{F})\otimes \B([0, \eta]) \otimes \B([0, T])$--measurable.
    In particular, 
    \begin{equation}\label{eq:minimizerofHam}
        \hat\alpha^+_t(t^*) = -\frac{1}{2\kappa\sigma}(Z^{0; t^*}|_{A'})\quad \text{and}\quad \hat\alpha^-_t(t^*,\tau^*) = -\frac{1}{2\kappa\sigma}(Z^{1; t^*}|_{A'})
    \end{equation}
    where $(Y^{0; t^*},Z^{0; t^*})$ and $(Y^{1; t^*}(\cdot),Z^{1; t^*}(\cdot))$ solve the iterative system of Brownian BSDEs
    \begin{equation}\label{eq:bsde.bubble.game.random.entry}
    \begin{cases}
        Y_t^{1; t^*}(\upeta) = X^{t^*}_{\tau^*}\beta_{\tau^*}\gamma_{\tau^*} + c(X^{t^*}_T)^2 + \int_t^T h^1(s, X^{t^*}, \theta_s, Z^{1; t^*}_s)ds\\
       \hspace{3cm} - \int_t^TZ_s^{1; t^*}(\upeta)dW_s, & \upeta \leq t \leq T\\
        Y_t^{0; t^*} = X^{t^*}_{\bar\tau(\mu)}\beta_{\bar\tau(\mu)}\gamma_{\bar\tau(\mu)} + c(X^{t^*}_T)^2  + \inttT h^0(s, X^{t^*}, \theta_s, \bar{\tau}(\mu),Z^{0; t^*}_s)ds\\      \hspace{3cm} + \inttT k_s(Y_s^{1; t^*}(s) - Y_s^{0; t^*})ds
    - \int_t^TZ^{0; t^*}_sdW_s, &  t^* \leq t \leq T,
    \end{cases}
    \end{equation}
    with 
    \begin{align*}
        h^0(t, \bx, \theta_t, \bar{\tau}(\mu), z) & = \frac{1}{4\kappa\sigma^2}(z|_{A'})^2 - \frac{z}{2\kappa\sigma^2}(z|_{A'}) +\phi \bx_t^2 - \bx_t\qv{\rho, \theta_t}_{F_\T} - \bx_tb(t, F_\T(t))\ind{0\leq t < \bar\tau(\mu)}\\
        h^1(t, \bx, \theta_t, z) & = \frac{1}{4\kappa\sigma^2}(z|_{A'})^2 - \frac{z}{2\kappa\sigma^2}(z|_{A'}) + \phi \bx_t^2 - \bx_t\qv{\rho, \theta_t}_{F_\T}.
\end{align*}
\end{theorem}

\begin{proof}[Proof of Theorem \ref{thm.bubble.random.entry}]
    We check that the conditions of Theorem \ref{Theorem:MFGExistence} are satisfied by the model specifications. Assumption \eqref{Assumption:SD} is satisfied given that $b(t, x, a) = a$ and $\sigma(t, x) = \sigma > 0$. Assumption \eqref{Assumption:ET} holds since the cost functions are not path dependent. Therefore, it remains to verify Assumptions \eqref{Assumption:C} and \ref{Assumption:BT2}.\\
    
    \emph{Assumption \eqref{Assumption:C}}. The running cost (\ref{runningcost}) can be decomposed as
\begin{equation*}
    f(t, X, q, \tau^*, a) = f^0(t, X, q, \bar{\tau}, a)\ind{0 \leq t < \tau} + f^1(t, X, q, a)\ind{\tau \leq t}  
\end{equation*}
where $f^0:  [0, T] \times \X \times \mathcal{P}(A) \times [0, T] \times A \to \R$ and $f^1:  [0, T] \times \X \times \mathcal{P}(A)\times A \to \R$  are
\begin{equation}\label{modelf0f1}
    \begin{split}
    f^0(t, \bx, q, \bar{\tau}, a) & = \kappa a^2 + \phi \bx_t^2 - \bx_tb(t, F_\T(t))\ind{0\leq t < \bar\tau} - \bx_t\qv{\rho, q}_{F_\T}\\
    f^1(t, \bx, q, a) & = \kappa a^2 + \phi \bx_t^2 - \bx_t\qv{\rho, q}_{F_\T}.
    \end{split}
\end{equation} 
It is easy to verify that the corresponding Hamiltonian, quadratic in $a$, is minimized by \eqref{eq:minimizerofHam}. Matching $f^0$ and $f^1$ with the form in Assumption \ref{Assumption:C1}, we get
\[
\begin{cases}
    f_a(t, \bx,a) & = \kappa a^2 + \phi \bx_t^2\\
    f_b(t, \bx, q) & = -\bx_t \qv{\rho, q}_{F_\T}\\
    f_c(t, \bx, q) & = -\bx_tb(t, F_\T(t)).
\end{cases}
\]
Similarly, the terminal cost (\ref{terminalcost}) can also be expressed as
\begin{align}
\notag
    X^{t^*}_{\tau^*}\beta_{\tau^*}\gamma_{\tau^*}+ c(X^{t^*}_T)^2 & = \pa{X^{t^*}_{\bar\tau(\mu)}\beta_{\bar\tau(\mu)}\gamma_{\bar\tau(\mu)} + c(X^{t^*}_T)^2}\ind{\tau > T} + \pa{X^{t^*}_{\tau^*}\beta_{\tau^*}\gamma_{\tau^*} + c(X^{t^*}_T)^2}\ind{\tau \leq T}\\
    \label{modelxi0xi1}
    & \eqqcolon g^0(X^{t^*}, \bar\tau(\mu))\ind{\tau > T} + g^1(X^{t^*},\tau^*)\ind{\tau \leq T}.
\end{align}
Assumption \ref{Assumption:MA2} implies both \ref{Assumption:C1} and \ref{Assumption:C2}.\\

\emph{Assumption \ref{Assumption:BT2}}. 
Let $\overline{\M}_K$ be the closure of $\M_K$. In light of Lemma \ref{Lemma:MkPreCompact} and Corollary \ref{Corollary:Brouwer}, we only need to check continuity of $\bar{\tau}(\cdot)$ on $\overline{\M}_K$ instead of $\M$. By continuity assumption on $F_{\T}$ and Lemma \ref{Lemma:Xcont}, all laws in $\M_K$ have continuous mean trajectory. Taking closure in Wasserstein space preserves the continuity of the
mean processes. Take a sequence $(\mu^n)_n$ in $\overline{\M}_K$ converging in the Wasserstein sense to $\mu \in \overline{\M}_K$. Let $\bar{\mu}^n$ and $\bar{\mu}$ be the entry-weighted average processes defined in \eqref{eq:MFG_averages}. Since $F_\T(0) > 0$, the process $\bar{\mu}^n$ is continuous for each $n$ and converges uniformly to $\bar\mu$, which is also continuous.

Define the trajectory of the running minimum of the average inventory on $[0, T]$: $$m_t \ce \inf_{s \in [0, t]}\bar \mu_s \quad\text{and}\quad m^n_t \ce \inf_{s \in [0, t]}\bar \mu^n_s.$$
For each $n$, by continuity arguments we have $$\sup_{t \in [0,T]}|m_t - m^n_t| \leq \sup_{t \in [0, T]}\sup_{s \in [0, t]}\abs{\bar\mu^n_s - \bar\mu_s} = \sup_{t \in [0, T]}\abs{\bar{\mu}_t - \bar \mu_t^n},$$
which converges to $0$ by uniform convergence of $\bar{\mu}^n$ to $\bar{\mu}$.
Therefore, $m^n$ also converges to $m$ uniformly. Let us recall the definition $$\bar{\tau} \ce \inf\cbra{t > 0: m_t \leq \zeta_t} \wedge T$$
where $\zeta$ is a continuous and strictly increasing function of $t$ with $\zeta_0 \in (0, \E[\iota])$. First we take care of the case where $m_t > \zeta_t$ for all $t \in [0, T]$. By uniform convergence of $m^n \to m$, we must also have $m^n_t > \zeta_t$ on $[0, T]$ for sufficiently large $n$ and we have $\bar{\tau}(\mu^n) = \bar{\tau}(\mu) = T$ eventually. Now we look at the case where the infimum is attained. Suppose there exists a subsequence of $\bar{\tau}(\mu^n)$ that converges to $\bar{\tau}' < \bar{\tau}(\mu)$. Without loss of generality we can assume the whole sequence converges to $\bar{\tau}'$. Since $\bar{\tau}' < \bar{\tau}(\mu)$, the running minimum $m$ has not reached the threshold yet, so we must have $m_{\bar{\tau}'} > \zeta_{\bar{\tau}'}$. By continuity of $m^n$ and $\zeta$, we have $\limn m^n_{\bar{\tau}(\mu^n)} = \limn \zeta_{\bar{\tau}(\mu^n)} = \zeta_{\bar{\tau}'}$. However, $m^n$ converges to $m$ uniformly, which yields a contradiction $\limn m^n_{\bar{\tau}(\mu^n)} = m_{\bar{\tau}'} > \zeta_{\bar{\tau}'}$. 

Now suppose there is a subsequence of $\bar{\tau}(\mu^n)$ that converges to $\bar{\tau}' > \bar{\tau}(\mu)$. Again without loss of generality, we can assume the whole sequence to converge to $\bar{\tau}'$. By continuity $m_{\bar{\tau}} = \zeta_{\bar{\tau}}$. Monotonicity of $m$ and strict monotonicity of $\zeta$ imply $$m_{\bar{\tau}'} \leq m_{\bar{\tau}(\mu)} = \zeta_{\bar{\tau}(\mu)} <  \zeta_{\bar{\tau}'}.$$
However, by uniform convergence we get contradiction again $$0 = \limn (m^n_{\bar{\tau}(\mu^n)} - \zeta_{\bar{\tau}(\mu^n)}) = m_{\bar{\tau}'} - \zeta_{\bar{\tau}'}<0. $$
Therefore $\bar{\tau}(\cdot)$ is sequentially continuous. Applying Theorem \ref{Theorem:MFGExistence} concludes the proof.
\end{proof}

\begin{remark}
Here we provide a simple example to show that strict monotonicity assumption of $\zeta$ is necessary for the continuity of $\bar\tau(\cdot)$. Consider $T = 2, \zeta_t = 0$ for all $t$. Define $\mu = \delta_{m}$ where $m_t = (1-t)\ind{t \in [0, 1]}$. For each $n$, define $\mu^n = \delta_{m^n}$ where
\begin{equation*}
    m^n_t = \begin{cases}
    m_t & \text{ if } t \in \sqbra{0, 1 - \frac{1}{n}}\\
    \frac{2 - t}{n+1} & \text{ if } t \in \left (1- \frac{1}{n}, 2 \right ]
    \end{cases}.
\end{equation*}
Since $m^n$ converges to $m$ uniformly, we have $\mu^n \nto \mu$ in Wasserstein distance. However, $\bar{\tau}(\mu^n) = 2$ for all $n$ and $\bar{\tau}(\mu) = 1$.

Note that the strict monotonicity is likely not a necessary assumption for the existence of equilibrium. In particular, the numerical experiments in the next section see no difference when we set a constant threshold $\zeta_t = \zeta_0$ vs $\zeta_t = \zeta_0 + \varepsilon t$ for a very small $\varepsilon > 0$.
\end{remark}

\subsection{Numerical Experiments}\label{Subsection:Numerical}
In this section, we provide some observations from the numerical experiments on the fixed entry setup with linear price impact. Specifically, for some $\delta > 0$ we take $$F_{\T}(t) = \ind{t \geq 0}\quad \text{and}\quad \rho(a) = \delta a.$$
Notice that the setup becomes almost linear-quadratic (LQ). It is well--known that explicit solutions to a true linear-quadratic game can be obtained via solving Ricatti ordinary differential equations (e.g. \cite{LQMFGs1, LQMFGs2}). 
However, the interaction through endogenous burst time and the presence of a random exogenous burst time prevent us from directly following this traditional route. Some useful observations from the LQ framework such as the form of the solutions can be borrowed nonetheless. We numerically solve the iterative BSDE \eqref{eq:bsde.bubble.game.random.entry} by discretization and use the implicit scheme described in \cite{PengXuNumerics11} to take advantage of the LQ structure.\footnote{See \href{https://github.com/peterwangshichun/VaryingEntry_Numerics}{this link} for codes.} Within each iteration we randomly draw an initial inventory $\iota$, a sample path of Brownian motion $W$ and an exogenous burst time $\tau$.

We assume a linear exogenous burst intensity $k_t = kt$, $k > 0$. The parameter $k$ should be viewed as a measure of common concern (or population prior) about the external shock. A large value of $k$ indicates that traders believe that the exogenous burst will happen sooner, which should lead to a more pessimistic MFG equilibrium strategy. For the bubble trend function, we match the bubble component to the LPPL model \eqref{eq:LPPL} with critical time $1.01T$ (as opposed to $T$, to avoid unbounded derivative for numerical purposes) and set $C = 0$ for simplicity. Since we assume a fixed entry setup, the model reduces to a simple power law, and the bubble component reads $$\gamma_t = \int_0^t b(u)du \approx \exp\left(A -B_0(1.01T - t)^\ell\right) - P_0$$ where $\ell \in (0, 1)$ controls price growth, and $B_0 > 0$ controls the scale of the power law.
With $\gamma_0=0$, we have $$A = B_0(1.01T)^\ell + \ln(P_0),$$
which gives us 
$$
    \gamma_t = P_0\exp\left(B_0(1.01T)^\ell -B_0(1.01T - t)^\ell\right) - P_0.
$$
Matching derivatives yields $$b(t) = P_0\exp\left(B_0(1.01T)^\ell -B_0(1.01T - t)^\ell\right)(\ell B_0 (1.01T -t)^{\ell - 1}).$$
Let us recall that $\beta_t$ controls the loss amplitude when the crash happens. For simplicity we set $\beta_t = 1$ for all $t$, which indicates that the price drops to exactly the fundamental value at burst. The initial inventory is set to be $\iota \sim N(10, 2)$. We fix the values of the following parameters: $\sigma = 1$, a constant threshold\footnote{Technically, one should use the increasing threshold function $\zeta_t = 2 + \varepsilon t$, with $\varepsilon > 0$ very small. However, this has no impact on the numerical results as mentioned in the previous remark.} function $\zeta = 2$ for endogenous burst, $T = 1$, quadratic terminal penalty $c = 10$, running inventory cost $\phi = 0.1$ and initial price $P_0 = 10$. By varying the other parameters, we inspect 5 specific scenarios (see Table \ref{table:paravalue} below). 

\begin{table}[h]
\centering
\begin{tabular}{|c||c | c | c | c | c | c | c||} 
 \hline
 Scenarios & $\kappa$ & $\delta$  & $k$ & $B_0$ & $\ell$ \\ [0.5ex] 
 \hline\hline
 \emph{Default} & 0.5 & 0.5 & 2 & $\log(20)$ & 0.5\\
 \hline
 \emph{BigBubble} & 0.5 & 0.5 &  2 & $1.3*\log(20)$ & 0.65\\
 \hline
 \emph{NoBubble} & 0.5 & 0.5 & 2 & $0$ & - \\
 \hline
 \emph{FearExo} & 0.5 & 0.5 &  5 & $\log(20)$ & 0.5\\
 \hline
 \emph{LowImpact} & 0.1 & 0.3 &  2 & $\log(20)$ & 0.5\\
 \hline
\end{tabular}
\vspace*{3mm}
\caption{Parameter Values under Different Scenarios}
\label{table:paravalue}
\end{table}

The equilibrium strategies from the default setting are plotted in Figure (\ref{fig:Default}). On average, the equilibrium execution strategy begins with a ``riding'' phase, where traders sell at a much lower speed compared to when there is no bubble (Figure \ref{fig:NoBubble}). Then they gradually start ``attacking'' the bubble until burst. By construction, the equilibrium strategy always jumps at $\bar{\tau}(\mu)$ because the equilibrium endogenous burst time is anticipated by all traders. Therefore, a sudden drop in trading speed takes place right after $\bar{\tau}(\mu)$ as players shift from profit-taking to standard optimal execution. However, if an early exogenous burst occurs, the players react in the opposite direction and immediately start dumping their shares at a higher pace (see for instance the exogenous burst paths in Figure \ref{fig:DefaultB}, \ref{fig:BigBubble} and \ref{fig:LowImpact}). These reactions are consistent in all scenarios where a bubble is present.
\begin{figure}
    \centering
    \begin{subfigure}[h]{0.47\textwidth}
        \centering
        \includegraphics[width=\textwidth]{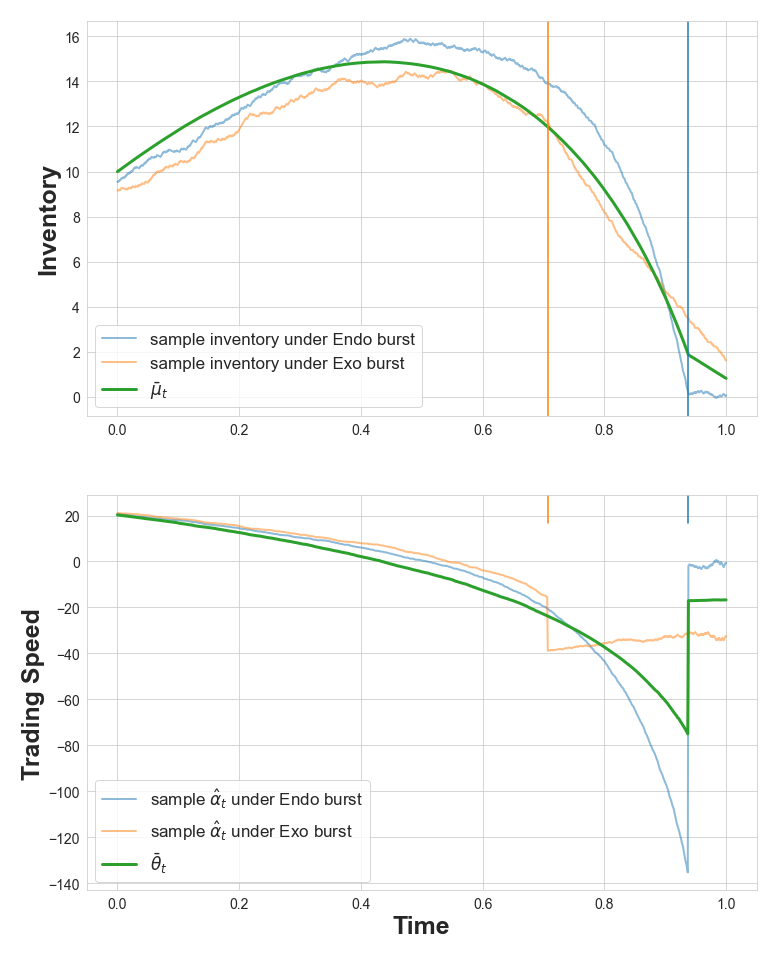}
        \caption{BigBubble: $J$= $-657.5$\\$B_0 = 1.3\log(20), \ell = 0.65$}
        \label{fig:BigBubble}
    \end{subfigure}
        \hfill
    \begin{subfigure}[h]{0.47\textwidth}
        \centering
        \includegraphics[width=\textwidth]{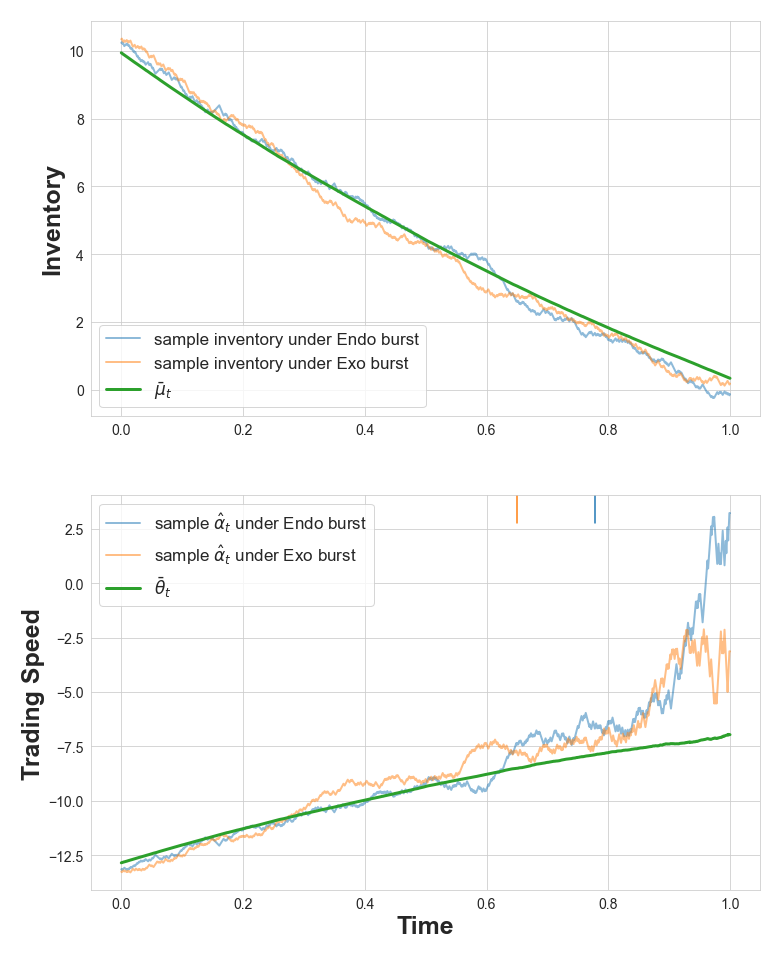}
        \caption{NoBubble: $J= 79.5$\\
        $B_0 =0$}
        \label{fig:NoBubble}
    \end{subfigure}
        \begin{subfigure}[h]{0.47\textwidth}
        \centering
        \includegraphics[width=\textwidth]{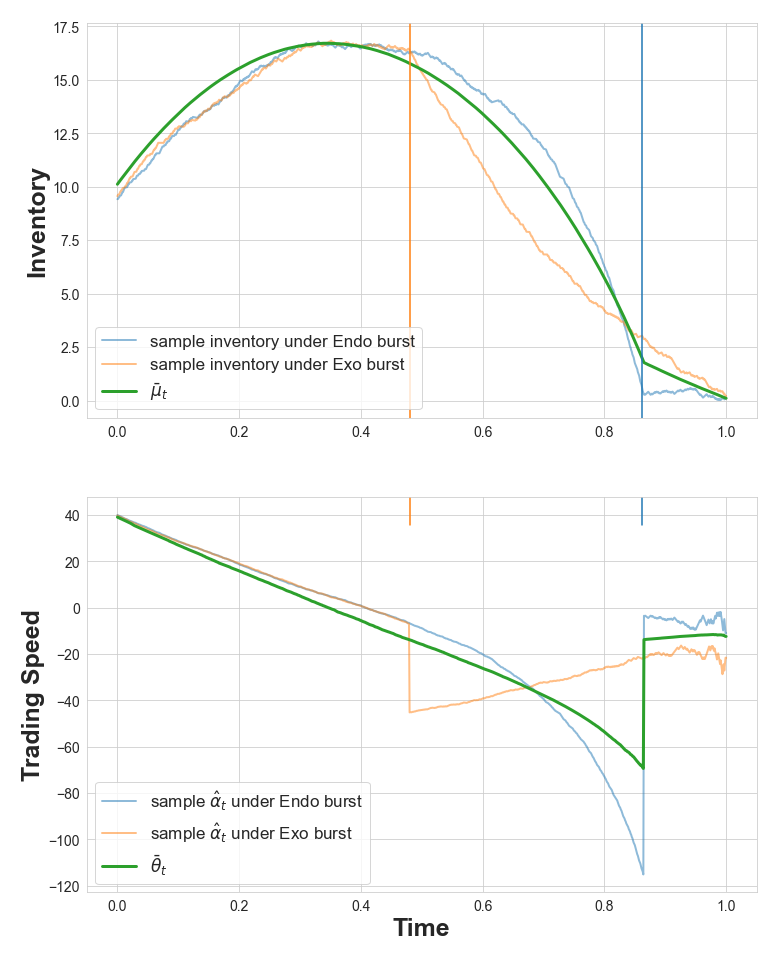}
        \caption{LowImpact: $J= -239.9$\\
        $\kappa = 0.1, \delta = 0.3$}
        \label{fig:LowImpact}
    \end{subfigure}
    \hfill
    \begin{subfigure}[h]{0.47\textwidth}
        \centering
        \includegraphics[width=\textwidth]{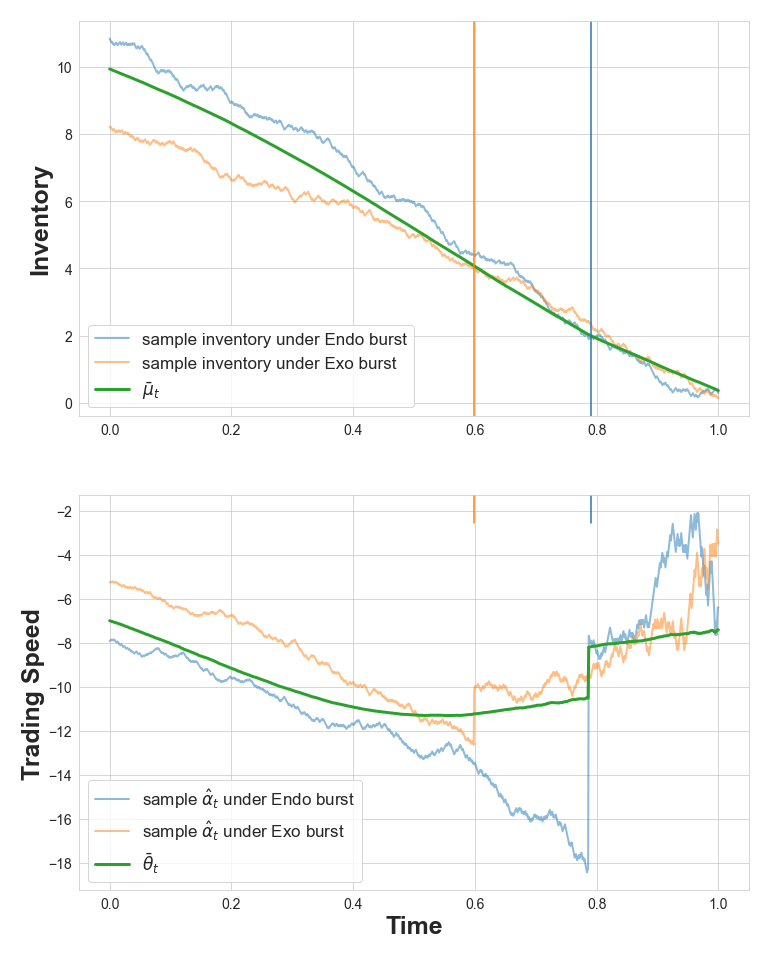}
        \caption{FearExo: $J= 18.4$\\
        $k = 5$}
        \label{fig:FearExo}
    \end{subfigure}

    \caption{Equilibrium Strategy under Different Scenarios}
    \label{fig:StandardPlot2}
\end{figure}

When the growth of the bubble is strong, traders may initially buy more shares to take advantage of the upward trend, which they use to offset the higher liquidation costs (see Figure \ref{fig:BigBubble} and \ref{fig:LowImpact}). However, an interesting phenomenon occurs in the \emph{BigBubble} setting where traders collectively delay the endogenous burst, compared to the \emph{Default} setting, in order to take further advantage of the explosive growth while keeping the transaction costs reasonable. On the contrary, traders in the \emph{LowImpact} scenario will trigger the burst sooner, even compared to the \emph{Default} case, by selling faster. In the \emph{FearExo} setting, we increase the level of concern towards an exogenous shock, parameterized by $k > 0$. Observe that the profit-taking phase before the exogenous burst is no longer present. Instead, players adopt a more conservative strategy by reducing inventory much earlier in the game. 

\begin{figure}[ht]
     \centering
     \begin{subfigure}[h]{0.49\textwidth}
         \centering
         \includegraphics[width=\textwidth]{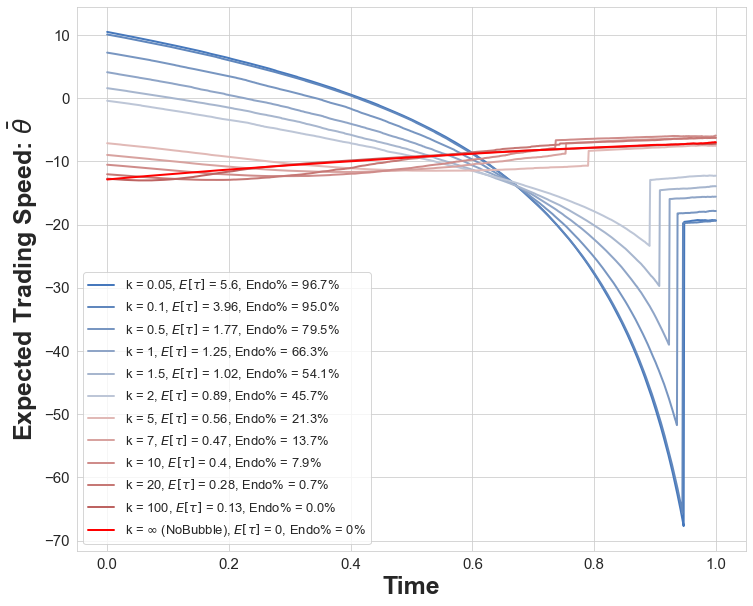}
         \caption{Equilibrium Trading Speed}
         \label{fig:theta_k}
     \end{subfigure}
     \hfill
     \begin{subfigure}[h]{0.49\textwidth}
         \centering
         \includegraphics[width=\textwidth]{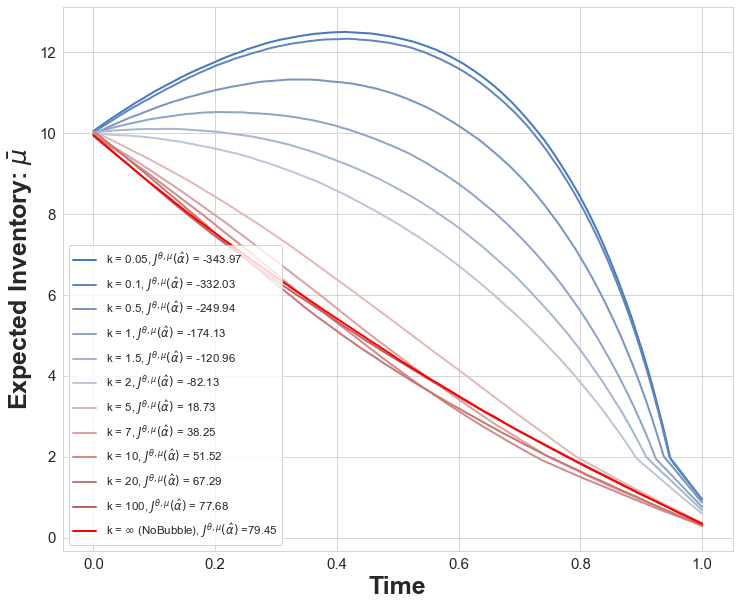}
         \caption{Equilibrium Inventory Level}
         \label{fig:mu_k}
     \end{subfigure}
        \caption{MFG Equilibrium under Different Priors on Exogenous Burst}
        \label{fig:Varyk}
\end{figure}
To better illustrate this change in mindset as the traders' prior on the exogenous burst varies, we fix all parameters to the \emph{Default} setting and plot the equilibrium strategies at different values of $k$ (see Figure \ref{fig:Varyk}). When $k$ is small (i.e. traders expect the exogenous burst to happen far in the future), we see a \emph{concave} shape in the equilibrium liquidation trajectory which resembles a ``riding before attacking'' strategy. 
When $k\to0$, players believe that there is almost no risk of an external shock.
If strong enough, this optimism alone could trigger players to initially buy shares to take advantage of the bubble. However, as we increase $k$, the concavity becomes less and less pronounced and traders eventually shift to a pessimistic strategy for which the equilibrium liquidation trajectory becomes \emph{convex}, resembling a ``risk-aversion'' strategy. The traders are thus reluctant to ride the bubble because they consider the bubble as a liability, rather than an opportunity. As $k \to \infty$, the bubble bursts closer to $0$ which eventually reduces to the \emph{NoBubble} scenario.
It would be interesting to determine the critical point at which players shift from the optimistic to the pessimistic strategy.
Observing the simulated results, it seems that the change of regimes occurs around the value of $k$ at which the equilibrium endogenous burst time is equal to the expected exogenous burst time.
\subsection{MFG Approximation of Finite-player Games}\label{Subection:Approximation}
In this section, we elaborate on the point made in Remark \ref{Remark:Approximation} that an approximate Nash equilibrium for the finite-player game can be constructed by a MFG equilibrium. Since the result is largely already presented in \cite[Sections 4 and 8]{CarmonaLacker15}, we shall only go over the main arguments to avoid repetition. 

First write the $N$-player game also in the weak formulation. Let us recall the probability setup introduced in Section \ref{Subsubsection:ProbSetup}. Take independent copies $X^i$ of the driftless state process $X$ in \eqref{eq:driftlessInv_model} with corresponding entry times $\T^i$, Brownian motions $W^i$ and enlarge for $i \in \N$, all on the same probability space. Let us recall the spaces of admissible controls $\A, \A(t^*)$ on $[0, T] \times \Omega$ and $\A_i, \A^N$ on the product space $[0, T] \times \Omega \times [0, \eta]^N$ from Definition \ref{Definition:AdmissibleControlN}. Given $\balpha = (\alpha^1, \dots, \alpha^N)\in \A^N$, define $\P^{\balpha}_N \sim \P \otimes \nu^N$ by $$\frac{d\P^{\balpha}_N}{d\P \otimes \nu^N} \ce \mathcal{E}\pa{\int_0^{\cdot}\sumN\sigma^{-1}\alpha^i_sdW^i_s}_T.$$
By Girsanov's theorem, under $\P^{\balpha}_N$, for each $i \leq N$ the processes $W^{\alpha^i}_t \ce W^i_t - \intt \sigma^{-1}\alpha^i_s ds$ form an $N$-dimensional Brownian motion, and $X^i$ satisfies $$dX^i_t = \alpha^i_tdt + \sigma dW^{\alpha^i}_t \text{ for } t \geq \T^i \quad \text{and} \quad X^i_t = 0 \as \text{ for } t \in [0, \T^i).$$
Let $\mu^N$ denote the empirical measure of the state processes and $\theta^N_t(\balpha)$ the empirical measure of some given controls $\balpha \in \A^N_N$ at time $t$. A key difference between the finite-player setup and the MFG is that the running cost depends on the random variable $\bT$ via the empirical CDF $F_{\bT}^{N}$. 

The weak formulation of the unconditional version of the objective \eqref{eq:Objective_N} for player $i$ with a given control $\balpha$ is 
\begin{equation*}
    \begin{split}
        J^{N, i}(\balpha) & \ce \E^{\P^{\balpha}_N}\left [\int_{\T^i}^{\tau^*(\mu^N)} f^+(t, X^i_t, F^N_{\bT}(t), \theta^N_t(\balpha), \alpha^i_t)dt\right.\\
& \left .+ \int_{\tau^*(\mu^N)}^T f^-(t, X_t^i, \theta^N_t(\balpha), \alpha_t^i)dt + g(X^i, \tau^*(\mu^N))\right],       
    \end{split}
\end{equation*}
where $f^+$ and $f^-$ are defined in \eqref{running_cost_f+} and \eqref{running_cost_f-}, and $g$ is defined in \eqref{terminalcost} as $$g(X, \upeta) \ce X_{\upeta}\beta_{\upeta}\gamma_{\upeta}+ c(X_T^i)^2.$$ 
Under slight abuse of notation, by $\balpha(\bt)$ we mean $(\alpha^1(\bt), \dots, \alpha^N(\bt))$. Denote by $J^{N, i}(\balpha; \bt)$ as the conditional expectation given $\bT = \bt$, that is
\begin{equation*}
    \begin{split}
        J^{N, i}(\balpha; \bt) & \ce \E^{\P^{\balpha}_N}\left [\int_{t^i}^{\tau^*(\mu^N)} f^+(t, X^i_t, F^N_{\bT}(t; \bt), \theta^N_t(\balpha(\bt)), \alpha^i_t(\bt))dt\right.\\
& \left .+ \int_{\tau^*(\mu^N)}^T f^-(t, X_t^i, \theta^N_t(\balpha(\bt)), \alpha_t^i(\bt))dt + g(X^i, \tau^*(\mu^N)) \right].       
    \end{split}
\end{equation*}

\begin{definition}\label{Definition:eps-Nash}
    Given $\varepsilon > 0$, an $\varepsilon$-Nash equilibrium for the $N$-player game is a set of admissible controls $\balpha \in \A^N$ such that for every $i \in \{1, \dots, N\}$, 
    \begin{equation}\label{eq:eps-Nash}
        J^{N, i}(\balpha) \leq J^{N, i}(\balpha^{\beta; -i}) + \varepsilon \quad \forall \beta \in \A_i.
    \end{equation}
\end{definition}
In order to construct an $\varepsilon$-Nash equilibrium, first note that since $\sigma > 0$, the process $(\iota^i, W^i_t)_{0 \leq t \leq T}$ generates the same filtration as $X^i$. Therefore, given a MFG equilibrium $(\hat{\alpha}, \hat{\theta}, \hat{\mu})$ from Theorem \ref{thm.bubble.random.entry}, the open-loop control $\hat{\alpha}$ can be written in a closed loop form, namely a deterministic, measurable function $\varphi$ of $(t, X, \T)$. We briefly outline the proof to show that for each $\varepsilon >0$, there exists $M_\varepsilon \in \N$ such that for all $N \geq M_{\varepsilon}$, the strategy profile $\hat{\balpha} =(\hat{\alpha}^1, \dots, \hat{\alpha}^N)$ is an $\varepsilon$-Nash equilibrium for the $N$-player game where $\hat{\alpha}^i_t = \varphi(t, X^i, \T^i)$. In other words, we can construct a strategy profile where each player only uses her own information, that is $X^i$ and $\T^i$, such that it can approximate a true Nash equilibrium arbitrarily well in the sense of Definition \ref{Definition:eps-Nash}. We only need to show the inequality \eqref{eq:eps-Nash} for the first player by symmetry. 

We now fix $t^1 \in [0, \eta]$ and define $\bt^1 \ce (t^1, \T^2, \dots, \T^N)$. Since we are fixing the entry time of just the first player, the notation $\hat{\balpha}^{\beta; -1}$ makes sense for $\beta \in \A(t^1)$ instead of $\beta \in \A_1$. In light of \cite[Lemma 8.2]{CarmonaLacker15}, the key step in the proof is to show that 
\begin{equation}\label{eq:convergence_limsup}
    \lim_{N \to \infty}\sup_{\beta \in \A(t^1)}\abs{J^{N, 1}(\hat{\balpha}^{\beta; -1}; \bt^1) - \hat{J}^{N}(\beta)} = 0
\end{equation}
where the intermediate objective function is defined as
\begin{equation*}
        \hat{J}^{N, 1}(\beta) \ce \E^{\P^{\balpha}_N}\left [\int_{t^1}^{\tau^*(\hat{\mu})} f^+(t, X^1_t, F_{\T}(t), \hat{\theta}_t, \beta_t)dt + \int_{\tau^*(\hat{\mu})}^T f^-(t, X_t^1, \hat{\theta}_t, \beta_t)dt + g(X^1, \tau^*(\hat{\mu})) \right].
\end{equation*}
Since $t^1$ is fixed, we are basically in the same setting as \cite{CarmonaLacker15} if we note from the proof of Theorem \ref{thm.bubble.random.entry} that $\tau^*(\cdot)$ is almost surely Wasserstein-continuous, and that $F^N_{\bT}(\cdot; \bt^1) \cas F_\T$ uniformly on $[0, \eta]$. Here we need to additionally assume that the bubble trend function $b$ is continuous.

Moreover, the limit \eqref{eq:convergence_limsup} also holds uniformly for all $t^1$ because $t^1$ only plays two roles. As entry time, it does not affect any bounds or coefficients (see e.g. \eqref{fandXintegrability1}, \eqref{fandXintegrability2}, and Lemma \ref{Lemma:uniformboundXtilde}). As a part of $F_{\bT}^N$, it is also fine since it only takes values in a compact interval. Therefore, $M_\varepsilon$ does not depend on $t^1$. We leave out the full proof as it mostly follows from Section 8 in \cite{CarmonaLacker15}.

\section{Concluding Remarks}\label{Section:Conclusion}
In this paper, we propose a model on optimal bubble riding under a stochastic differential game setting, in which traders realize the presence of the bubble at possibly different times before they enter the market and begin liquidating their assets. Each trader attempts to take advantage of the uptrend knowing that a future crash will happen. This formulation leads to a MFG with varying entry times, where players interact through both permanent price impact and the endogenous burst condition. We then prove the existence of MFG equilibria in a generic framework by a fixed point argument. To incorporate the exogenous burst, we adopt the techniques of progressive enlargement of filtration from credit risk literature, which enables us to decompose the MFG equilibrium into pre-burst and post-burst components that only require market information and time of exogenous burst once it occurs.

We revisit the model with simulation results by numerically solving the BSDE in the weak formulation of MFGs. The equilibrium control is found to have a different shape from the classical optimal execution problems, showing both a riding phase and an attacking phase in the liquidation strategy. While the uptrend provides motivation for herding strategies, the inevitable, random burst time propels the traders to eventually attack the bubble. We also notice that traders could collectively delay the endogenous burst in the presence of a strong enough bubble. 

Another intriguing result arises when we analyze how traders' view on the exogenous burst influences their equilibrium strategy. As the traders expect an earlier exogenous burst, there is a critical point where the concave shaped liquidation strategies become convex. It would be interesting to further investigate this regime change. Lastly, we give a proof outline of the inverse convergence result, showing that the MFG equilibrium can be used to construct an approximate Nash equilibrium. As we are the first to look at MFGs with varying entry times, there are many future directions to take. Some examples include having players choose their entry time based on their risk preference, allowing inference on $k$, or incorporating a leader-follower framework to model both institutional and retail traders. We also refer the readers to the accompanying paper \cite{TangpiWang23} for an extension with price-dependent entry times and common noise.

  \appendix



\section{Appendix}\label{Section:Appendix}
The appendix consists of three sections. 
We first present the proofs of results on BSDEs under enlarged filtration used in the proof of our MFG existence results.
These follow by adaptations of \citet{KharroubiLim14}. Then we provide some regularity results needed for the varying entry times section, as well as a brief discussion on common noise.

\subsection{BSDE Solutions under Enlarged Filtration}
In this section we study the generic BSDE \eqref{BSDE ProEnl0} and prove the results given in \ref{sec:BSDE-progress}.
We begin with a uniform integrability result for the following BSDE: for a fixed $\upeta \in \R_+$, consider  $$Y_t^1(\upeta) = \xi^1(\upeta) + \int_t^T G^1(s, Z_s^1(\upeta), \upeta)ds - \int_t^TZ_s^1(\upeta)dW_s, \qquad \upeta \wedge T \leq t \leq T.$$
Under Assumption \eqref{Assumption:PEA}, by classical results of Pardoux and Peng \cite[Theorem 4.1]{PardouxPeng90}, for each $\upeta \in \R_+$ there exists a unique solution $(Y^1(\upeta), Z^1(\upeta)) \in \S_{\mathbb{F}}^2[\upeta \wedge T, T] \times \H_{\mathbb{F}}^2[\upeta \wedge T, T]$.

\begin{lemma}\label{integrabilityY1Z1}
There exists a constant $C > 0$ such that for all $\upeta \in \R_+$ and all $t \in [\upeta \wedge T, T]$, we have
$$\sup_{\upeta \in [0, T]}\E\sqbra{\abs{Y_t^1(\upeta)}^2 + \int_{\upeta \wedge T}^T |Z^1_s(\upeta)|^2ds} < \infty.$$
In particular, $\sup_{s \in [0, T]}\E[|Y^1_s(s)|] < \infty$.
\end{lemma}
\begin{proof}
Apply It\^{o}'s formula to $|Y^1_t(\upeta)|^2$ and take expectations. By \eqref{Assumption:PEA} and Young's inequality we get 
\begin{align*}
\E\sqbra{\abs{Y^1_t(\upeta)}^2} & = \E\sqbra{|\xi^1(\upeta)|^2} + \E\sqbra{\int_{\upeta \wedge T}^T 2Y_s^1(\upeta)G^1(s, Z^1_s(\upeta), \upeta)ds} - \E\sqbra{\int_{\upeta \wedge T}^T \abs{Z^1_s(\upeta)}^2ds} \\
& \leq \E\sqbra{|\xi^1(\upeta)|^2 + \int_{\upeta \wedge T}^T \frac{1}{\varepsilon}\abs{Y_s^1(\upeta)}^2 ds + \int_{\upeta \wedge T}^T \varepsilon|G^1(s,Z_s^1(\upeta),\upeta)|^2ds - \int_{\upeta \wedge T}^T \abs{Z^1_s(\upeta)}^2ds}\\
& \leq C + \frac{1}{\varepsilon}\E\sqbra{\int_{\upeta \wedge T}^T \abs{Y_s^1(\upeta)}^2 ds} + C\varepsilon T + (C\varepsilon - 1)\E\sqbra{\int_{\upeta \wedge T}^T  |Z^1_s(\upeta)|^2ds}
\end{align*}
for any $\varepsilon > 0$. 
Choose $\varepsilon = \frac{1}{C}$ to remove the last term. Applying Gr\"{o}nwall's inequality yields
$$\E\sqbra{\abs{Y^1_t(\upeta)}^2} \leq (C + T)\exp(CT) \eqqcolon C_Y.$$
Plug this uniform bound back to the previous inequality and choose any $\varepsilon < \frac{1}{C}$, we get
$$(1 - C\varepsilon)\E\sqbra{\int_{\upeta \wedge T}^T  |Z^1_s(\upeta)|^2ds} \leq C + \frac{TC_Y}{\varepsilon} + C\varepsilon T.$$
Dividing both sides by $(1 - C\varepsilon) > 0$ yields result, as $C$ does not depend on $\upeta$.
\end{proof}

\begin{proposition}[Existence of BSDE Solution]\label{Prop:BSDEexistence}
Under Assumptions \ref{Assumption:BT1} and \eqref{Assumption:PEA}, there exists a solution $(Y, Z, U) \in \mathcal{S}^2_{\mathbb G} \times \mathcal{H}_\mathbb{G}^2 \times \mathcal{H}_{\mathbb{G}, D}^2$ to the BSDE \eqref{BSDE ProEnl0} satisfying \eqref{BSDEsolution}.
\end{proposition}
\begin{proof}
We cannot directly apply \cite[Theorem 3.1]{KharroubiLim14} since we do not assume that $\xi$ is almost surely bounded. However, adopting the proof of this result, it suffices to check the existence of solution to the following recursive system of Brownian BSDEs:
\begin{equation}
\begin{cases}
    Y_t^1(\upeta) = \xi^1(\upeta) + \int_t^T G^1(s, Z_s^1(\upeta), \upeta)ds - \int_t^TZ_s^1(\upeta)dW_s & \text{ for } \upeta \wedge T \leq t \leq T\\
    Y_t^0 = \xi^0 + \int_t^T\pa{G^0(s, Z^0_s) + k_sY_s^1(s) - k_sY_s^0}ds - \int_t^TZ^0_sdW_s & \text { for } 0 \leq t \leq T.
\end{cases}
\end{equation}
The second equation is well defined due to Lemma \ref{integrabilityY1Z1}. 
By Assumption \eqref{Assumption:PEA} and \cite[Theorem 4.1]{PardouxPeng90}, the first equation has a unique solution for all $\upeta \in \R_+$: $(Y^1(\upeta), Z^1(\upeta)) \in \mathcal{S}^2_{\mathbb F}[\upeta \wedge T, T] \times \mathcal{H}_\mathbb{F}^2[\upeta \wedge T, T]$. 
Moreover, by \cite[Proposition C.1]{KharroubiLim14}, we obtain $(Y^1, Z^1)$ as a $\mathcal{P}(\mathbb{F}) \otimes \B(\R_+)$-measurable process. 
Similarly, by Assumption \ref{Assumption:BT1} in addition to \eqref{Assumption:PEA}, the second equation has a unique solution $(Y^0, Z^0) \in \mathcal{S}^2_{\mathbb F} \times \mathcal{H}_\mathbb{F}^2$. 
Then by \cite[Theorem 3.1]{KharroubiLim14}, we arrive at a solution $(Y, Z, U) \in \mathcal{S}^2_{\mathbb G} \times \mathcal{H}_\mathbb{G}^2 \times \mathcal{H}_{\mathbb{G}, D}^2$ on $[0, T]$ satisfying \eqref{BSDEsolution}.
\end{proof}

\begin{remark}
    In \cite{KharroubiLim14}, the authors make the assumption that $Y^i$ is an essentially bounded process (on $[0, T]$, $i = 0, 1$) for a stronger conclusion that $Y$ is also essentially bounded. Since we do not have a bounded terminal cost, we obtain an integrability property of the solution instead of boundedness. The proof that $(Y, Z, U)$ is in $\mathcal{S}^2_{\mathbb G} \times \mathcal{H}_\mathbb{G}^2 \times \mathcal{H}_{\mathbb{G}, D}^2$ easily follows from the original proof, thus we omit it here.
\end{remark}

\subsubsection{Proof of Proposition \ref{Prop:BSDEcomparison}}
By \cite[Lemma 2.1]{KharroubiLim14} we have the representation of $(Y, Z, U)$ as a sum of pre-- and post--burst processes:
\begin{align*}
    Y_t & = Y_t^0 \ind{t < \tau} + Y_t^1(\tau)\ind{t \geq \tau}\\
    Z_t & = Z_t^0 \ind{t \leq \tau} + Z_t^1(\tau)\ind{t > \tau}\\
    U_t & = U_t^0 \ind{t \leq \tau}
\end{align*}
where $Y^0, Z^0$ are $\mathbb F$-predictable and $U^0$ is $\mathbb F$-progressively measurable. Note that $U$ only has the pre--burst component by Remark (\ref{Remark:convention}). 
Respectively, we also get representation of $(Y', Z', U')$ by $(Y^{'0}, Y^{'1}, Z^{'0}, Z^{'1}, U^{'0})$. By \cite[Theorem 12.23]{bookSemiMGTheory} and Remark \ref{Remark:convention}, $(Y^1, Z^1)
$ (resp. $(Y^{'1}, Z^{'1})$) is a solution to the Brownian BSDE
\begin{align*}
    Y^1_t(\tau) & = \xi + \int_t^T G(s,Z_s^1(\tau))ds - \int_t^T Z_s^1dW_s, \qquad \tau \wedge T \leq t \leq T\\
    \biggl(\text{resp. } Y^{'1}_t(\tau) & = \xi' + \int_t^T G'(s,Z_s^{'1}(\tau))ds - \int_t^T Z_s^{'1}dW_s , \qquad \tau \wedge T \leq t \leq T)\biggr).
\end{align*}
By standard comparison theorem for BSDEs, we have $ Y_t^1 \leq  Y_t^{'1}$, see e.g. \cite[ Corollary 4.4.2]{DarlingPardoux97}. 
Now define drivers $F$ and $F'$ by
\begin{align*}
    F(t, y, z) & = G(t, z) + Y_t^1(t) - y\\
    F'(t, y, z) & = G'(t, z) + Y_t^{'1}(t) -y.
\end{align*}
We thus have $F \leq F'$. 
According to \cite[Theorem 4.1]{KharroubiLim14}, we only need to check that the drivers $F$ and $F'$ satisfy a comparison theorem for a standard Brownian BSDE. This can be verified again by \cite[Corollary 4.4.2]{DarlingPardoux97} under Assumption \eqref{Assumption:PEA}.\qed

\begin{proposition}[Uniqueness of BSDE Solution]\label{Prop:BSDEuniqueness}
Under Assumptions \ref{Assumption:BT1} and \ref{Assumption:PEA}, there exists at most one solution $(Y, Z, U) \in \mathcal{S}^2_{\mathbb G} \times \mathcal{H}_\mathbb{G}^2 \times \mathcal{H}_{\mathbb{G}, D}^2$ to the BSDE (\ref{BSDE ProEnl0}).
\end{proposition}
The proof of the uniqueness result largely follows from Proposition \ref{Prop:BSDEcomparison}, and we refer readers to the proof of \cite[Theorem 4.2]{KharroubiLim14} for the main argument. The only difference is that we use \cite[Corollary 4.4.2]{DarlingPardoux97} as the basis of the comparison principle. As a result, we require the driver to be Lipschitz in $Z$ but do not need any convexity assumptions as in \cite[Theorem 4.2]{KharroubiLim14}. 

Putting the results together, we have proven Theorem \ref{Theorem:BSDEmain}.

\begin{proposition}[Stability]\label{Prop:BSDEstability}
Let $(\xi, G)$ and $(\xi', G')$ be coefficients that satisfy \eqref{Assumption:PEA}. Suppose \ref{Assumption:BT1} is also satisfied and $(Y, Z, U)$, $(Y', Z', U')$ are the solutions to BSDE \eqref{BSDE ProEnl0} with coefficients $(\xi, G)$ and $\xi', G'$ respectively. Define $\Delta Y = Y - Y', \Delta Z = Z - Z'$ and 
\begin{align*}
    \Delta \xi^0 \ce \xi^0 - \xi^{'0}, & \quad \Delta \xi^1(\upeta) \ce \xi^1(\upeta) - \xi^{'1}(\upeta)\\
    \Delta G^0(s, z) \ce G^0(s, z) - G^{'0}(s, z), &\quad \Delta G^1(s, z, \upeta) \ce G^1(s, z, \upeta) - G^{'1}(s, z, \upeta).
\end{align*}  
Then there exists a constant $C>0$ such that 
\begin{equation}
    \begin{split}
        \E\sqbra{\intT\abs{\Delta Y_t}^2 + \abs{\Delta Z_t}^2dt} & \leq C\E\sqbra{|\Delta \xi^0|^2 + \intT \abs{\Delta G^0(s, Z^0_s)}^2ds}\\ 
        & + C \E\sqbra{|\Delta \xi^1(\tau)|^2 + \int_{\tau \wedge T}^T \abs{\Delta G^1(s, Z^1_s(\tau), \tau)}^2ds}\\
        & + C\E\sqbra{\intT |\Delta \xi^1(\upeta)|^2 d\upeta + \intT\int_{\upeta \wedge T}^T \abs{\Delta G^1(s, Z^1_s(\upeta), \upeta)}^2dsd\upeta}.
    \end{split}
\end{equation}
\end{proposition}
\begin{proof}
By uniqueness, we can construct $(Y, Z, U)$ and $(Y', Z', U')$ using the procedure in the proof of Proposition \ref{Prop:BSDEexistence} and get the decomposition as in (\ref{BSDEsolution}). 
First note that we then have
\begin{align*}
    \abs{\Delta Y_t}^2 & = \abs{\Delta Y_t^0\ind{t < \tau} + \Delta Y^1_t(\tau)\ind{\tau \leq t}}^2 = \abs{\Delta Y_t^0}^2\ind{t < \tau} + \abs{\Delta Y_t^1(\tau)}^2\ind{\tau \leq t}\\
    \abs{\Delta Z_t}^2 & = \abs{\Delta Z_t^0\ind{t \leq \tau} + \Delta Z^1_t(\tau)\ind{\tau < t}}^2 = \abs{\Delta Z_t^0}^2\ind{t \leq \tau} + \abs{\Delta Z_t^1(\tau)}^2\ind{\tau < t}
\end{align*}
with 
\begin{align*}
    \Delta Y^0 \ce Y^0 - Y^{'0}, & \quad \Delta Y^1(\upeta) \ce Y^1(\upeta) - Y^{'1}(\upeta)\\
    \Delta Z^0 \ce Z^0 - Z^{'0},  &\quad \Delta Z^1(\upeta) \ce Z^1(\upeta) - Z^{'1}(\upeta).
\end{align*}
Therefore we have 
\begin{align*}
    \intT\abs{\Delta Y_t}^2 + \abs{\Delta Z_t}^2dt & = \intT\pa{\abs{\Delta Y^0_t}^2 + \abs{\Delta Z^0_t}^2}\ind{t \leq \tau}dt\\
    & + \intT\pa{\abs{\Delta Y^1_t(\tau)}^2 + \abs{\Delta Z^1_t(\tau)}^2}\ind{t > \tau}dt\\
    & \leq \intT\pa{\abs{\Delta Y^0_t}^2 + \abs{\Delta Z^0_t}^2}dt + \int_{\tau\wedge T}^T\pa{\abs{\Delta Y^1_t(\tau)}^2 + \abs{\Delta Z^1_t(\tau)}^2}dt.
\end{align*}
First we look at $(\Delta Y^1, \Delta Z^1)$. Since $G^1$ is assumed to be Lipschitz in $Z$ uniformly in $\upeta$, using \cite[Proposition 4.4]{DarlingPardoux97} we can find a constant $C_1 > 0$ independent of $\upeta$ such that for all $\upeta \in \R_+$:
\begin{subequations}
\begin{equation}\label{stabilitydecomposed:1}
\begin{split}
   & \E\sqbra{\abs{\Delta Y_\upeta^1(\upeta)}^2 + \int_{\upeta \wedge T}^T\abs{\Delta Y^1_s(\upeta)}^2 + \abs{\Delta Z^1_s(\upeta)}^2ds}\\
   & \qquad \leq C_1 \E\sqbra{|\Delta \xi^1(\upeta)|^2 + \int_{\upeta \wedge T}^T \abs{\Delta G^1(s, Z^1_s(\upeta), \upeta)}^2ds}.
\end{split}
\end{equation}

In particular, with integrability of $\abs{\Delta Y_s^1(s)}^2$ guaranteed by Lemma \ref{integrabilityY1Z1}, we have 
\begin{equation}\label{stabilitydecomposed:yss}
    \E\sqbra{\intT \abs{\Delta Y_\upeta^1(\upeta)}^2d\upeta} \leq C_1\E\sqbra{\intT |\Delta \xi^1(\upeta)|^2 d\upeta + \intT\int_{\upeta \wedge T}^T \abs{\Delta G^1(s, Z^1_s(\upeta), \upeta)}^2dsd\upeta}.
\end{equation}

Similarly for $(\Delta Y^0, \Delta Z^0)$, by Lemma \ref{integrabilityY1Z1} and Assumption \ref{Assumption:BT1}, applying \cite[Proposition 4.4]{DarlingPardoux97} again, we can find $C_0 > 0$ such that
\begin{equation}\label{stabilitydecomposed:0} 
    \E\sqbra{\intT\abs{\Delta Y^0_s}^2 + \abs{\Delta Z^0_s}^2ds} \leq C_0\E\sqbra{|\Delta \xi^0|^2 + \intT \abs{\Delta G^0(s, Z^0_s)}^2ds + \intT\abs{\Delta Y_s^1(s)}^2ds}.
\end{equation}
Conclude by combining \eqref{stabilitydecomposed:1}, \eqref{stabilitydecomposed:yss} and \eqref{stabilitydecomposed:0}.
\end{subequations}
\end{proof}
\begin{lemma}\label{Lemma:tau*decomp}
    Under Assumptions \eqref{Assumption:BT}, \eqref{Assumption:SD}, and \eqref{Assumption:C}, for each $(\theta, \mu) \in \Theta \times \M$, let $(Y, Z, U)$ denote the unique solution to the BSDE \eqref{eq:WeakMFG_BSDE}.
    Then $Z$ can be decomposed as $$Z_t = Z_t^0\ind{t \leq \tau^*} + Z_t^1(\tau^*)\ind{t > \tau^*},$$
where $Z^0$ and $Z^1(\cdot)$ are the $\mathbb{F}$-predictable components from the decomposition by $\tau$ given in \eqref{Theorem:BSDEmain}.
\end{lemma}

\begin{proof}
    Let $(Y, Z, U)$ be the unique solution to the BSDE 
    \eqref{eq:WeakMFG_BSDE}. 
    Let us recall that  $\tau^* =  \tau\wedge  \bar\tau(\mu)$.
    From the decomposition property of the BSDE solution in Theorem \ref{Theorem:BSDEmain}, we have 
    \begin{equation}\label{label:Zdecomp}
    \begin{split}
    Z_t &= \pa{Z^0_t\ind{t \leq \tau} + Z^{1}_t(\tau)\ind{t > \tau}}\ind{\tau \leq \bar{\tau}(\mu)} + \pa{Z^0_t\ind{t \leq \tau} + Z^{1}_t(\tau)\ind{t > \tau}}\ind{\tau > \bar{\tau}(\mu)}\\
    & = \pa{Z_t^0\ind{t \leq \tau^*}+ Z^1_t(\tau^*)\ind{t > \tau^*}}\ind{\tau \leq \bar{\tau}(\mu)} \\
    & +\pa{Z_t^0\ind{t \leq \tau^*} + Z_t^0\ind{\tau^* < t \leq \tau} - Z^1_t(\tau)\ind{\tau^* < t \leq \tau}+ Z_t^1(\tau)\ind{t > \tau^*}} \ind{\tau > \bar{\tau}(\mu)}.
    \end{split}
\end{equation}
    In other words, on the event $\{\tau \leq \bar{\tau}(\mu)\}$, i.e. $\tau^* = \tau$, we have the desired decomposition. On the event $\{\tau > \bar{\tau}(\mu)\}$, i.e. $\tau^* = \bar{\tau}$, we need to show that $Z^0_t$ and $Z^1_t(\tau)$ coincide for $t \in (\tau^*, \tau]$. 
    First note that by Assumption \ref{Assumption:C2} we have $\xi^1(\tau) = \xi^0$, which no longer depends on $\tau$. Let us recall that $h^1$ does not depend on $\upeta$ either. Hence for all $\upeta \in \R_+$, $(Y^1(\upeta), Z^1(\upeta))$ solves the same BSDE
\begin{equation*}
    Y^1_t = \xi^0 + \inttT h^1(s, X, \theta_s, Z^1_s)ds - \inttT Z^1_sdW_s, \qquad 0 \leq t \leq T
\end{equation*}
which  is independent of $\upeta$. Under \ref{Assumption:C1}, observe that we always have
\begin{equation}\label{label:samebursteffectf}
    f^0(t, \bx, q, \bar{\tau}, a)\ind{t \geq \tau^*} = f^1(t, \bx, q, a)\ind{t \geq \tau^*}.
\end{equation}

Therefore, on $t > \bar{\tau}(\mu)$ and for all $(\bx, q, \bar{\tau}, z) \in \X \times \mathcal{P}(A) \in [0, T] \times \R$ one has $$h^1(t, \bx, q, z) = h^0(t, \bx, q, \bar{\tau}, z).$$
Specifically, $(Y^1, Z^1)$ satisfies on $(\bar{\tau}(\mu), T]$ the equation $$Y^1_t = \xi^0 + \inttT h^0(s, X, \theta_s, \bar{\tau}(\mu), Z^1_s)ds - \inttT Z^1_sdW_s.$$
Let us recall that $(Y^0, Z^0)$ uniquely solves 
\begin{equation}\label{label:particularBSDE0}
    Y^0_t = \xi^0 + \inttT \pa{h^0(s, X, \theta_s, \bar{\tau}(\mu), Z^0_s) + k_sY_s^1 - k_sY^0_s} ds - \inttT Z^0_sdW_s.
\end{equation}
It is immediate to see that $(Y^1, Z^1)$ also satisfy \eqref{label:particularBSDE0} on $(\bar{\tau}, T]$. Therefore, by uniqueness of the solution we have $$Z_t^0\ind{\tau^* < t \leq \tau} = Z^1_t \ind{\tau^* < t \leq \tau}.$$
Note again that on the event $\{\tau > \bar{\tau}(\mu)\}$ the process $Z^1_t$ does not depend on $\tau$. In particular $Z^1_t = Z^1_t(\tau^*)$. 
Then \eqref{label:Zdecomp} yields the desired result.
\end{proof}

\subsection{Regularity Results for Varying Entry Times}
\begin{lemma}\label{Lemma:uniformboundXtilde}
    For $p \in \N$, there exists $C > 0 $ such that for $\nu$-almost all $t^* \in [0, \eta]$:
    \orange{$$\E^{\P}\sqbra{\norm{X(t^*)}_\infty^p} + \sup_{\alpha \in \A^*}\E^{\P^{\alpha; t^*}}\Big[\|X(t^*)\|^p_\infty\Big] \leq C$$}
    where $X^{t^*} = X(\cdot, \cdot, t^*)$ is the driftless state \eqref{DriftlessState*} with varying entry time $t^*$.
\end{lemma}
\begin{proof}
    The linear growth condition of $\sigma$ provides a bound for the first summand by standard arguments. The bound is uniform in $t^*$ because $0 \leq t^* \leq \eta \leq T$. For each $\alpha \in \A^{*}$, we recall that $\P^{\alpha; t^*}$ is defined by
    \begin{equation*}
        M^{\alpha; t^*}_t \ce \mathcal{E}\pa{\int_0^{\cdot} \sigma^{-1}(s, X^{t^*})b(s, X^{t^*}, \alpha^{t^*}_s) dW_s}_t\quad \text{and}\quad \frac{d\P^{\alpha; t^*}}{d\P} \ce M^{\alpha; t^*}_T.
    \end{equation*}
    By uniform boundedness of $|\sigma^{-1}b|$, standard arguments give
    $$\E^{\P}\sqbra{|M^{\alpha; t^*}_T|^{2}} \leq \exp(K^2T) \text{ for all } \alpha \in \A^*, t^* \in [0, \eta]$$
    Applying Cauchy-Schwarz, we have:
    \orange{\begin{equation}\label{eq:t*uniformbound}
        \E^{\P^{\alpha; t^*}}\Big[\|X^{t^*}\|^p_\infty\Big] \leq \sqrt{\E^{\P}\sqbra{\abs{M_T^{\alpha; t^*}}^{2}}} \sqrt{\E^{\P}\sqbra{\norm{X^{t^*}}_\infty^{2p}}} \leq C
    \end{equation}}
    where $C$ only depends on $p, K, T$ and the bound in Lemma \ref{Lemma:fandXintegrability}, but not on $t^*$ or $\alpha$. Also, we note that $$\sup_{\alpha \in \A^*}\E^{\P^{\alpha}}\Big[\|X\|^{p}_\infty\Big] = \sup_{\alpha \in \A^*}\E^{\P^{\alpha}}\sqbra{\E^{\P^{\alpha; t^*}}\sqbra{\|X\|^{p}_\infty \ | \ \T = t^*}} \leq C.$$
\end{proof}

\begin{lemma}\label{Lemma:MetricSpace}
    The space $(\X^*, d_{\X^*})$ defined in Section \ref{Subsection:MFGModel} is a separable and complete metric space. 
\end{lemma}
\begin{proof}
    It is easy to check that $d_{\X^*}$ is indeed a metric. Separability follows from that of $\X$. For completeness, take a Cauchy sequence $\{\bx^n(t^n)\}_{n \geq 1}$ in $\X^*$. Then by definition of $d_{\X^*}$, the sequence of entry times $t^n$ must be Cauchy and therefore converges to some $t^* \in [0, \eta]$. Similarly, the continuous counterparts $\bar{\bx}^n(t^n)$ is a Cauchy sequence in $\X$ and thus converges to some $\bar{\bx}^* \in \X$. In particular, $\bar{\bx}^*$ must be constant on $[0, t^*]$. Then it is not hard to show that the Cauchy sequence $\bx^n(t^n)$ converges to $\bx^*(t^*) \in \X^*$ with respect to the metric $d_{\X^*}$, where $\bx^*_t(t^*) \ce \ind{t \geq t^*}\bar{\bx}^*_t$. 
\end{proof}
\begin{lemma}\label{Lemma:Xt*cont}
    If a sequence $t^n \nto t^*$ in $[0, \eta]$, then for all $p \geq 1$, $\E\sqbra{d_{\X^*}\pa{X^{t^n}, X^{t^*}}^p} \nto 0$. 
\end{lemma}

\begin{proof}
    First assume $t^n \uparrow t^*$ in $[0,\eta]$. 
    Since each trajectory has the same initial value, we have 
\begin{equation*}
    \bar{X}^{t^n}_t - \bar{X}^{t^*}_t = \int_{t^n}^t\sqbra{\sigma(s, X^{t^n}) - \ind{s \geq t^*}\sigma(s, X^{t^*})}dW_s.
\end{equation*}
Applying BDG and Jensen's inequality and the Lipschitz property of $\sigma$, allowing $C > 0$ to change from line to line, we take expectation with respect to $\P$ and get
\begin{align*}
    \E\sqbra{\sup_{s\in [0, t]}|\bar{X}^{t^n}_s - \bar{X}^{t^*}_s|^{p}} & \leq C\E\sqbra{\int_{t^n}^t \abs{\sigma(s, X^{t^n}) - \ind{s \geq t^*}\sigma(s, X^{t^*})}^pds}\\
    & \leq C\E\sqbra{\int_{t^n}^{t^*}|\sigma(s, \bar{X}^{t^n})|^pds} + C\E\sqbra{\int_{t^*}^t \abs{\sigma(s, \bar{X}^{t^n}) - \sigma(s, \bar{X}^{t^*})}^pds}\\
    &\leq C\int_{t^n}^{t^*}\E\sqbra{|\sigma(s, \bar{X}^{t^n})|^p}ds + C\ell_\sigma^p\int_{t^*}^{t}\E\Big[\sup_{u \in [0, s]}|\bar{X}^{t^n}_u - \bar{X}^{t^*}_u|^p\Big]ds.
\end{align*}
The first term, which we denote as $C_n$, goes to $0$ due to the linear growth of $\sigma$ and (\ref{fandXintegrability1}). Applying Gr\"{o}nwall's lemma yields
\begin{equation*}
    \E\sqbra{d_{\X^*}\pa{X^{t^n}, X^{t^*}}^p} \leq 2^{p-1}C_n\exp(C\ell_\sigma^p |t^n - t^*|) + 2^{p-1}|t^n - t^*|^p \nto 0.
\end{equation*}
The case $t^n\downarrow t^*$ follows by the same argument.
\end{proof}

\subsection{MFGs with Common Noise}\label{Subsection:CommonNoise}
Let us recall that the filtration that we work with in this paper does not include the noise term driving the price dynamics $W^0$, but only the individual inventory information $W$. In this section we elaborate on Remark \ref{Remark:CommonNoise} to justify this seemingly restrictive assumption. Unlike the idiosyncratic noise $W$, the common noise remains even when $N \to \infty$. As a consequence, $\theta$ and $\mu$ become random laws which require a different definition for MFG equilibrium. In particular, they now become the conditional laws of controls and state given $W^0$. 

Let $(\Omega, \F, \P)$ also carry another Wiener process $W^0$ that is independent from everything else. Denote by $\mathbb{F}^c \ce (\F_t^c)_{t \in [0, T]}$ the natural filtration now including common noise and by $\mathbb{G}^c$ its progressive enlargement by $\tau$, where $\F^c_t$ is the $\P$-completion of $\sigma((\iota, W_s, W^0_s)_{s \in [0, t]})$. Also define $\mathbb{F}^{W^0}\ce (\F^{W^0}_t)_{t \in [0, T]}$ to be the natural filtration generated by $W^0$. Let $\A^{*, c}(t^*)$ and $\A^{*, c}$ follow Definition \ref{Definition:AdmissibleControl} but with common noise (i.e. $\mathbb{G}^c$-progressive, A-valued processes with appropriate entry times).

\begin{definition}\label{Definition:CommonNoise}
    A MFG equilibrium with common noise for random entry times is a triple $(\hat\alpha, \theta^c,\mu^c)$ where the random probability measures $(\theta^c, \mu^c): \Omega \to \Theta \times \M$ and the optimal control $\hat\alpha \in \A^{*, c}$ satisfy 
    \begin{enumerate}
        \item $\hat\alpha$ minimizes over $\A^{*, c}$ the objective $J^{\theta^c,\mu^c}$ defined in \eqref{Definition:MFGobjectiveProductSpace}.
        \item $\mu^c$ is a version of the conditional law of $X$ given $W^0$ under $\P^{\hat\alpha}$: $$\mu^c_t(\cdot) = \int_{[0, \eta]}\P^{\hat{\alpha}(t^*)}\pa{X^{t^*}_t \in \cdot | \F_t^{W^0}}\nu(dt^*) \text{ for almost every } t.$$
        \item $\theta^c$ is a version of the conditional law of $\hat\alpha$ given $W^0$ under $\P^{\hat\alpha}$: $$\theta^c_t(\cdot) =  \int_{[0, \eta]}\P^{\hat{\alpha}(t^*)} \pa{\hat\alpha_t(t^*) \in \cdot | \F_t^{W^0}}\nu(dt^*) \text{ for almost every } t.$$
    \end{enumerate}
\end{definition}

In contrast to the full-fledged versions of MFGs with common noise described in \cite{CarmonaBookII}, our setup has the following distinctive features:
\begin{enumerate}
    \item Individual state variable $X$ does not involve $W^0$.
    \item Running cost function $f$ and terminal cost $g$ do not depend on common noise $W^0$ (or the common state $P$).
    \item Entry time and burst time do not depend on $W^0$.
\end{enumerate}

\begin{proposition}\label{Prop:CommonNoise}
    A MFG equilibrium $(\hat\alpha, \theta, \mu)$ without common noise in the sense of Definition \ref{definition:MFGequilibrium} is also a MFG equilibrium with common noise in the sense of Definition \ref{Definition:CommonNoise} (where $(\theta, \mu): \Omega \to \Theta \times \M$ is deterministic).
\end{proposition}

\begin{proof}
    We check the conditions in Definition \ref{Definition:CommonNoise} are satisfied by $(\hat{\alpha}, \theta, \mu)$. First note that $\hat{\alpha}$ is still admissible since $\mathbb{F} \subset \mathbb{F}^c$ and thus $\mathbb G\subset \mathbb G^c$. Since $\hat\alpha$ does not depend on $W^0$, neither does $X$. Therefore the conditional law properties of $\mu$ and $\theta$ reduce to the unconditional versions, which follow from the fact that $(\hat\alpha, \theta, \mu)$ is an equilibrium. It remains to show that $\hat\alpha$ minimizes $J^{\theta, \mu}$. Using arguments in Section \ref{Section:RET}, it suffices to show optimality $t^*$-by-$t^*$. 
    
    Let $(Y^{t^*}, Z^{t^*}, U^{t^*})$ be the unique solution to the BSDE \eqref{WeakMFG_BSDE*} with input $(\theta, \mu)$. Assumption \ref{Assumption:C3}, Remark \ref{Remark:AssumptionC} and Proposition \ref{Prop:BSDEcomparison} imply that $\hat{\alpha}_t(t^*) = \hat{a}(t, X^{t^*}, Z^{t^*}_t)$ for almost every $t$. Now consider the following BSDE on $[t^*, T]$ solved on the filtration $\mathbb{G}^c$:
    \begin{equation}\label{BSDE_CommonNoise}
    \begin{split}
            Y^{t^*, c}_t & = g(X^{t^*}, \bar{\tau}(\mu^{c}), \tau) + \int_t^T h(s, X^{t^*}, \theta^{c}_s, \bar{\tau}(\mu), \tau,Z^{t^*, c}_s)ds\\
            & - \int_t^T Z^{t^*, c}_sdW_s - \int_t^T \mathfrak{Z}^{t^*, c}_sdW^0_s - \int_t^T U^{t^*, c}_sdM_s.\\
    \end{split}
    \end{equation}
    A solution to this BSDE is $(Y^{t^*, c}, Z^{t^*, c}, \mathfrak{Z}^{t^*, c}, U^{t^*, c}) \in \mathcal{S}^2_{\mathbb{G}^c}[t^*, T] \times \mathcal{H}_{\mathbb{G}^c}^2[t^*, T] \times \mathcal{H}_{\mathbb{G}^c}^2[t^*, T] \times \mathcal{H}_{\mathbb{G}^c, D}^2[t^*, T]$, which exists and is unique (see e.g. \cite[Theorem 53.1]{KarouiMazliakBSDE}). However, it is also easy to verify that $(Y^{t^*}, Z^{t^*}, \bzero, U^{t^*})$ is a solution by construction, so by uniqueness, $(Y^{t^*, c}, Z^{t^*, c}) = (Y^{t^*}, Z^{t^*})$ and $\hat{\alpha}_t(t^*) = \hat{a}(t, X^{t^*}, Z^{t^*, c}_t)$ for $t \in [t^*, T]$ which is still the unique minimizer of the Hamiltonian. 
    Thus, using comparison principle of BSDE \eqref{BSDE_CommonNoise}, for any $\beta\in \A^{*, c}(t^*)$ we have
    \begin{align*}
        J^{\theta, \mu, t^*}(\beta) & = \E^{\P^{\beta}}\bigg[g(X^{t^*},\bar\tau(\mu), \tau) + \int_{t^*}^Tf(s, X^{t^*}, \theta_s, \bar{\tau}(\mu), \tau, \beta)ds\bigg] \\
        & \leq \E[Y^{t^*, c}_{t^*}] = \E^{\P^{\hat\alpha(t^*)}}[Y^{t^*, c}_{t^*}] = J^{\theta, \mu, t^*}(\hat{\alpha}(t^*)).
    \end{align*}
    We have checked that $(\hat{\alpha}, \theta, \mu)$ satisfies both the optimality condition and the fixed point condition, so it is a MFG equilibrium with common noise.
\end{proof}
The argument above implies the existence, albeit a somewhat trivial one, of MFG equilibrium even if we allow our controls to depend on common noise. However, it is important to note that the key reason enabling our model to avoid common noise is that the bubble trend functions depends only on the entry time, which is specified exogenously. Proposition \ref{Prop:CommonNoise} would not be applicable for instance if individual entry also depended on the price trajectory. See our accompanying paper \cite{TangpiWang23} for a full treatment.

\bibliographystyle{plainnat}
\bibliography{references}  

\end{document}